\documentclass[11pt]{article}

\usepackage[T1]{fontenc}
\usepackage[utf8]{inputenc}
\usepackage{lmodern}
\usepackage{amsmath,amssymb,amsthm,mathtools}
\usepackage{graphicx}
\usepackage{float}
\usepackage{booktabs}
\usepackage{caption}
\usepackage{microtype}
\usepackage[a4paper,left=28mm,right=28mm,top=26mm,bottom=26mm]{geometry}
\usepackage{enumitem}
\usepackage[hidelinks]{hyperref}
\newcommand{\R}{\mathbb R}
\newcommand{\B}{\mathcal B}
\newcommand{\E}{\mathbb E}
\newcommand{\norm}[1]{\left\lVert#1\right\rVert}

\newtheorem{theorem}{Theorem}[section]
\newtheorem{proposition}[theorem]{Proposition}
\newtheorem{lemma}[theorem]{Lemma}
\newtheorem{corollary}[theorem]{Corollary}
\newtheorem{remark}[theorem]{Remark}

\title{Optimal exponential memory for sequential Euclidean connections:\\
edge-power costs and phase transitions}
\author{Pedro M. M. de Castro}
\date{September 3, 2026}

\hypersetup{
  pdftitle={Optimal exponential memory for sequential Euclidean connections: edge-power costs and phase transitions},
  pdfauthor={Pedro M. M. de Castro}
}

\begin{document}
\maketitle

\begin{abstract}
We study the edge-power cost of the labelled tree generated by the
\(\gamma\)-strategy, a constant-gain rule for sequential Euclidean
connections.  Starting with \(x_0=p_0\), each input point \(p_i\) is attached
to \(x_{i-1}\), after which the state is updated by
\(x_i=\gamma x_{i-1}+(1-\gamma)p_i\).  Retaining \(x_i\) as a labelled
auxiliary vertex subdivides the insertion segment into the spine edge
\([x_{i-1},x_i]\) and the leaf edge \([x_i,p_i]\).  The memory parameter
\(\gamma\) controls how long earlier input points influence subsequent
attachment points.  We minimize the sum of the \(\alpha\)-powers of these edge
lengths under independent uniform input and under arbitrary input sequences.

For uniform points in the unit ball, the stationary problem has a transition
at \(\alpha=1\).  Its continuous extension is minimized at the boundary for
\(0<\alpha\leq1\), and every global minimizer is interior for \(\alpha>1\).
The main result determines the finite optimizer in the joint window
\(\alpha_N=1+\varepsilon_N\), \(\varepsilon_N\log N\to\lambda\).  Below an
explicit threshold it lies on the \(N^{-1/2}\) scale.  At the threshold its
scale is \(\sqrt{\log N/(N\log\log N)}\), and above the threshold it approaches
an explicit stationary root with two computable corrections.  A second
threshold identifies which correction governs the location, and differentiated
estimates prove eventual uniqueness.  At \(\alpha=3d+8\), the linear
coefficient at the stationary endpoint changes sign and a branch of strict
local maxima enters the parameter interval.  Under arbitrary input sequences,
the optimal parameter and asymptotic worst-case edge-power cost per processed
point are explicit for \(0<\alpha\leq3\).  At high powers, \(m\)-block
inputs, which periodically repeat \(m\) copies of a point followed by \(m\)
copies of its antipode, provide explicit lower-bound witnesses.  Combined
with a separation argument, they show that the optimized cost is asymptotic
to \(2\log2/\log\alpha\).
Exact results for powers two and four, a rational recursion for all even
powers, and a high-dimensional expansion provide additional descriptions of
the optimizer.
\end{abstract}

\noindent\textbf{AI-use disclosure.}
OpenAI's ChatGPT and Codex were used as research tools to assist the author with
literature retrieval, mathematical exploration, symbolic and numerical checks,
and draft preparation and revision.  The author formulated the research
questions and mathematical framework, developed and refined the results and
arguments with this assistance, and critically reviewed the manuscript
throughout.  Responsibility for the mathematical statements and the final
manuscript rests entirely with the author.

\section{Introduction}

This section introduces the sequential connection rule and its network-design
interpretation.  It then states the questions addressed in the paper and
summarizes the resulting phase diagram and parameter choices.

Let \(p_0,p_1,\ldots,p_N\) be input points in the unit ball \(\B\) of
\(\R^d\), processed in a prescribed order.  An online rule must choose the
attachment point used at step \(i\) from the points already processed.  We
study the \(\gamma\)-strategy, the time-homogeneous constant-gain rule
\begin{equation}
 x_i=\gamma x_{i-1}+(1-\gamma)p_i,
 \qquad 0\leq\gamma\leq1.
 \label{eq:recursion}
\end{equation}
At step \(i\), the input point \(p_i\) is attached to \(x_{i-1}\), and the
updated point \(x_i\) subdivides this insertion segment.  Retaining the updated
points produces a caterpillar graph labelled by processing time.  The memory
parameter \(\gamma\) is selected before the sequence is processed and controls
how long earlier input points influence later attachments.  The rule stores one
\(d\)-dimensional point and applies the same update at every step.

This construction defines a sequential geometric network-design problem.  A
state may represent a relay node, an aggregation point, or a stored attachment
point from which the next connection is initiated.  At step \(i\), the
edge \([x_{i-1},x_i]\) records the displacement needed to update the state,
and \([x_i,p_i]\) records the remaining connection to the current input point.
The model isolates the edge-power cost of the labelled tree.  A
complete communication or facility
model may additionally charge for activating a state, transmitting through
several hops, delay, capacity, or movement.  The present objective determines
the optimal memory parameter for the isolated geometric component.

Retaining every updated point prevents the subdivision degeneracy that occurs
for edge-power costs above power one.  The precise labelled-tree convention is
given in Section~\ref{sec:model}.  The transition at power one consequently
records a genuine design tradeoff within this fixed construction.

Sequential Euclidean connection was studied by Steele
\cite{Steele1989}, and power-weighted insertion trees and stars were compared
in \cite{CastroDevillers2011}.  Random online nearest-neighbor graphs provide
a related probabilistic model in which every point is connected to a previous
point selected from the observed set \cite{PenroseWade2008,Wade2009}.
Online Steiner tree algorithms instead compare the maintained network with an
offline optimum, often through competitive ratio or bounded recourse
\cite{ImaseWaxman1991,AlonAzar1993,GuGuptaKumar2016}.  Our state is prescribed
by an affine summary of the complete observed sequence, and the decision
variable is its memory parameter.

Distance-power objectives also arise in wireless range assignment.  Recent
online and dynamic formulations minimize sums of powered transmission ranges
while maintaining a broadcast arborescence
\cite{deBergMarkovicUmboh2023,deBergSadhukhanSpieksma2024}.  Their node-range
cost is different from a sum over all realized edges.  The shared operational
question is how much geometric cost is required when a network must be updated
after each new point.  Our rule offers one stable update with a single stored
point and permits exact distributional and adversarial analysis.

Eq.\eqref{eq:recursion} is the geometric moving average of Roberts
\cite{Roberts1959}, usually called an exponentially weighted moving average.
Exponential smoothing has a long history in forecasting and stochastic
approximation \cite{Cogger1974}, while iterated random functions provide the
stationary framework for the affine recursion \cite{DiaconisFreedman1999}.
The corresponding insertion cost was analyzed in
\cite{CastroGammaInsertion}.  The present paper determines the memory
parameter that minimizes the edge-power cost of the resulting labelled
caterpillar graph.

Our first contribution is a complete finite-size phase diagram near power one.
When \(0<\alpha<1\), every finite optimizer approaches the stationary endpoint
at scale \(N^{-1/(\alpha+1)}\).  In the joint window
\(\alpha_N=1+\varepsilon_N\), \(\varepsilon_N\log N\to\lambda\), an explicit
threshold \(\lambda^\star(d)\) separates a subcritical \(N^{-1/2}\)
optimizer from
a supercritical stationary scale.  At the threshold, a Lambert profile yields
the scale \(\sqrt{\log N/(N\log\log N)}\).  Above it, the first displacement
correction has stationary and transient terms, and their dominance changes
when \(\lambda=3\lambda^\star(d)/2\).  Uniform estimates for
two derivatives prove eventual uniqueness in the critical and supercritical
windows.

Our second contribution describes the global geometry of the stationary
problem.  Its continuous extension is minimized at \(\gamma=1\) for
\(0<\alpha\leq1\), while every global minimizer lies in \((1/2,1)\) for
\(\alpha>1\).  At \(\alpha=3d+8\), the linear endpoint coefficient changes
sign.  For larger \(\alpha\), a strict local maximum enters the physical
interval from \(\gamma=1\), while the endpoint is a one-sided local minimum and
the global minimum remains interior.  This
bifurcation explains the obstruction to a global convexity proof of uniqueness.

Our third contribution concerns the adversarial edge-power objective.  We first
establish the exact adversarial insertion cost for \(0<\alpha\leq3\), a result
also developed in \cite{CastroGammaInsertion}, and then determine the optimal
edge-power parameter and value in this range.  At high
powers, periodic block inputs provide explicit lower-bound witnesses.  A new
separation argument supplies the matching upper bound and proves that the
optimized adversarial value is asymptotic to \(2\log2/\log\alpha\).  Exact quadratic and
quartic solutions, a rational recursion for all even powers, and an expansion
of the optimizer in high dimension complement these three main results.

Section~\ref{sec:model} defines the model and its probabilistic framework.
Section~\ref{sec:stationary-optimization} treats stationary optimization under
uniform random input, and Section~\ref{sec:worst-case} treats the worst-case
problem.  Their numerical comparison appears in
Section~\ref{sec:numerical-stationary-worst-case}.  Sections~\ref{sec:finite-regimes}
and~\ref{sec:high-dimensional} analyze finite-size and high-dimensional limits,
whose numerical checks are collected in Section~\ref{sec:numerical-scaling}.

\section{Model and probabilistic framework}
\label{sec:model}

This section defines the labelled caterpillar generated by the
\(\gamma\)-strategy and reduces its edge-power cost to a scalar factor times
the insertion cost.  It then introduces the stationary state and the geometric
moments used throughout the optimization results.

\subsection{Labelled caterpillar and cost reduction}

At step \(i\), retaining \(x_i\) creates the spine edge
\([x_{i-1},x_i]\) and the leaf edge \([x_i,p_i]\).  Each occurrence of
\(p_i\) and \(x_i\) carries its processing index.  Vertices created at
different steps remain distinct when their Euclidean positions coincide.  The
vertices \(x_0,\ldots,x_N\) form the spine, and each \(p_i\) is a leaf adjacent
to \(x_i\).  These are the only retained edges and vertices; additional
subdivision vertices are excluded.

Let
\begin{equation}
 L_{\alpha,N}^{(\gamma)}
 =\sum_{i=1}^{N}\norm{p_i-x_{i-1}}^\alpha
 \label{eq:insertion-cost}
\end{equation}
be the insertion cost.  The two edges generated at step \(i\) have lengths
\[
 \norm{x_i-x_{i-1}}
 =(1-\gamma)\norm{p_i-x_{i-1}}
\]
and
\[
 \norm{p_i-x_i}
 =\gamma\norm{p_i-x_{i-1}}.
\]
Consequently, the edge-power cost of the labelled tree is
\begin{equation}
 T_{\alpha,N}^{(\gamma)}
 =h_\alpha(\gamma)L_{\alpha,N}^{(\gamma)},
 \qquad
 h_\alpha(\gamma)
 =\gamma^\alpha+(1-\gamma)^\alpha.
 \label{eq:reduction}
\end{equation}

For independent uniform input, write
\[
 F_{d,\alpha,N}(\gamma)=\E T_{\alpha,N}^{(\gamma)}.
\]

Figure~\ref{fig:canonical-construction} compares the input-order path, the
labelled caterpillar, and the star centered at \(p_0\), obtained from the same
input sequence as the memory parameter varies.

\begin{figure}[t]
 \centering
 \includegraphics[width=\textwidth]{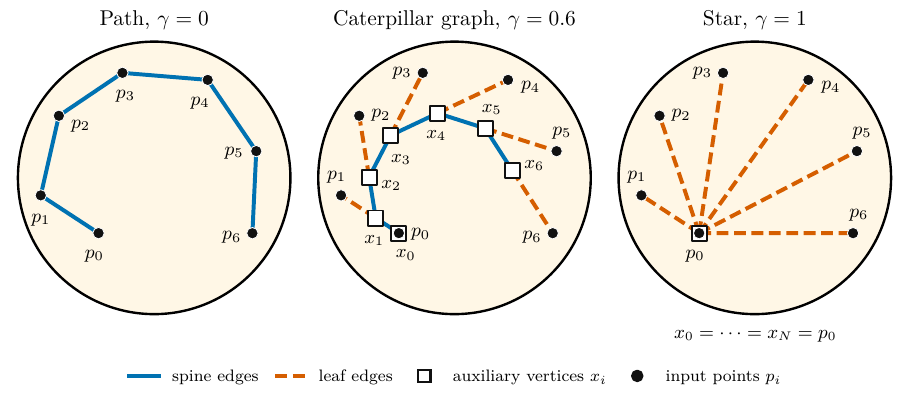}
 \caption{Effect of the memory parameter on the construction generated from
 the same ordered input.  For \(\gamma=0\), contracting the zero-length leaf
 edges gives the input-order path.  For \(0<\gamma<1\), the labelled auxiliary
 vertices \(x_i\) form the spine, and each input point \(p_i\) is joined to
 \(x_i\), giving the labelled caterpillar.  For \(\gamma=1\), contracting the
 zero-length spine edges gives the star centered at \(p_0\).  Color
 distinguishes spine edges from leaf edges, while marker shape distinguishes
 auxiliary vertices from input points.}
 \label{fig:canonical-construction}
\end{figure}

\begin{remark}
For \(\alpha=1\), Eq.\eqref{eq:reduction} is the Euclidean length of the
labelled caterpillar.  For \(\alpha\ne1\), it is the edge-power cost of
the same graph.  Subdividing an edge of length \(r\) into \(m\) equal pieces
changes its contribution to
\(m^{1-\alpha}r^\alpha\).
\end{remark}

\begin{remark}
The normalized quantity
\[
 \frac{T_{\alpha,N}^{(\gamma)}}{h_\alpha(\gamma)}
 =L_{\alpha,N}^{(\gamma)}
\]
is the insertion cost and measures attachment quality before multiplication by
the two-edge factor \(h_\alpha(\gamma)\).  We use both quantities to keep these
two effects explicit.
\end{remark}

The parameter has a direct design interpretation.  Small \(\gamma\) moves the
state strongly toward the current input point and leaves a short terminal segment.
Large \(\gamma\) moves the state slowly and leaves a longer terminal segment.
At \(\alpha=1\), the objective is total installed length.  For \(\alpha>1\),
long individual connections receive an increasing penalty, as in
distance-power network costs.  For \(0<\alpha<1\), the concave power rewards
concentration of length in fewer long connections.  The stationary, finite,
and adversarial optimizers below quantify the memory parameter selected by each of
these regimes.

\begin{remark}[The star parameter and the stationary limit]
For finite \(N\), choosing \(\gamma=1\) keeps \(x_i=p_0\) and produces the
star centered at \(p_0\).  For \(0\leq\gamma<1\) under independent uniform
input, the influence of \(p_0\) vanishes and the recursion has a unique
stationary law.  Thus the finite construction at \(\gamma=1\) and the
continuous stationary extension at the same parameter correspond,
respectively, to a star centered at \(p_0\) and a star centered at the center
of the ball.
\end{remark}

\subsection{Stationary state and geometric moment formulas}
\label{sec:inputs}

This subsection collects the probabilistic quantities used in the optimization
results.  It introduces the stationary insertion moment, records its
monotonicity and endpoint behavior, and derives the low-order and geometric
formulas needed later.  The general formulas support the tradeoff comparison
in Figure~\ref{fig:uniform-adversarial-tradeoff}, and the exact second and fourth
moments are revisited in Figure~\ref{fig:high-dimensional-first-order}.

Throughout the stationary analysis, let \(p\) be uniform in \(\B\), independent
of the stationary state
\begin{equation}
 x_\gamma=(1-\gamma)\sum_{j=0}^{\infty}\gamma^j p_j,
 \qquad 0\leq\gamma<1.
 \label{eq:stationary-state}
\end{equation}
Let
\begin{equation}
 M_{d,\alpha}(\gamma)
 =\E\norm{p-x_\gamma}^{\alpha},
 \qquad
 c_{d,\alpha}
 =\E\norm{p}^{\alpha}
 =\frac{d}{d+\alpha}.
 \label{eq:stationary-moment}
\end{equation}
The stationary insertion moment is ordered by the memory parameter.

\begin{lemma}
\label{lem:stationary-moment-order}
For every \(\alpha>0\), the function \(M_{d,\alpha}\) is nonincreasing on
\([0,1)\), satisfies \(M_{d,\alpha}(\gamma)>c_{d,\alpha}\) for
\(0\leq\gamma<1\), and
\[
 \lim_{\gamma\to1^-}M_{d,\alpha}(\gamma)=c_{d,\alpha}.
\]
For \(\alpha\geq1\), it is differentiable on \((0,1)\) and satisfies
\begin{equation}
 M_{d,\alpha}'(\gamma)<0.
 \label{eq:strict-moment-derivative}
\end{equation}
\end{lemma}

\begin{proof}
The order comparison for every positive power, its strict comparison with the
endpoint, and the endpoint limit are proved in
Appendix~\ref{app:foundational-insertion-inputs}.  That proof uses majorization
and the theorem of Olkin and Tong~\cite{OlkinTong1988}.  We give here a direct
derivative argument for \(\alpha\geq1\), which is the range used in the
interior stationary analysis below.

Let \(p_{-1},p_0,p_1,\ldots\) be independent uniform points of \(\B\), and put
\[
 a_j=(1-\gamma)\gamma^j,
 \qquad
 z=p_{-1}+\sum_{j\geq0}a_jp_j.
\]
Central symmetry gives \(M_{d,\alpha}(\gamma)=\E\norm z^\alpha\).  Define
\[
 \psi_\alpha(z)=\norm z^{\alpha-2}z,
 \qquad
 b_j=\alpha\E[\psi_\alpha(z)\mathbin{\cdot}p_j],
\]
where \(\psi_1(z)=z/\norm z\) away from the null event \(z=0\).  Since
\[
 a_j'=\frac{a_j}{\gamma}
 \left(j-\frac{\gamma}{1-\gamma}\right),
\]
differentiation and pairwise symmetrization yield
\begin{equation}
 -M_{d,\alpha}'(\gamma)
 =\frac1\gamma\sum_{0\leq i<j}
 a_ia_j(j-i)(b_i-b_j).
 \label{eq:paired-moment-derivative}
\end{equation}
To determine the sign of each difference, fix \(i<j\), let
\(u=a_i>a_j=v\), and collect all terms independent of \(p_i,p_j\) in
\(r\).  Set
\[
 z_1=r+up_i+vp_j,
 \qquad
 z_2=r+up_j+vp_i.
\]
Exchangeability gives
\begin{equation}
 b_i-b_j
 =\frac{\alpha}{2(u-v)}
 \E\bigl[(\psi_\alpha(z_1)-\psi_\alpha(z_2))
 \mathbin{\cdot}(z_1-z_2)\bigr].
 \label{eq:paired-gradient-difference}
\end{equation}
For \(\alpha>1\), \(\psi_\alpha\) is the gradient of the strictly convex
function \(z\mapsto\norm z^\alpha/\alpha\), so the integrand is nonnegative
and is positive with positive probability.  At \(\alpha=1\), monotonicity of
the subgradient of the Euclidean norm gives the same conclusion after
discarding the null set where a subgradient is not uniquely determined.
Thus \(b_i>b_j\), and Eq.\eqref{eq:paired-moment-derivative} proves
Eq.\eqref{eq:strict-moment-derivative}.  Finally,
\(x_\gamma\to0\) in \(L^2\) as \(\gamma\to1^-\), and bounded convergence
gives the endpoint value.
\end{proof}

The next two stationary moment formulas were established in
\cite{CastroGammaInsertion}.  They are recalled with a short derivation because
they are used repeatedly in the exact low-power calculations.

\begin{lemma}
Let \(s_d=d/(d+2)\).  Then
\begin{equation}
 M_{d,2}(\gamma)=\frac{2s_d}{1+\gamma}.
 \label{eq:second-moment}
\end{equation}
The fourth moment is
\begin{equation}
 M_{d,4}(\gamma)
 =
 \frac{4d}{(d+2)(d+4)}
 \frac{(d+3)(1+\gamma^2)-2\gamma^3}
 {(1+\gamma)^2(1+\gamma^2)}.
 \label{eq:fourth-moment}
\end{equation}
\end{lemma}

\begin{proof}
The state and the independent point are centered and isotropic.  For the
coefficient sequence
\[
 a_0=1,
 \qquad
 a_{j+1}=(1-\gamma)\gamma^j,
 \]
the required difference has the same distribution as
\(\sum_{j\geq0}a_jp_j\).  Therefore
\[
 \sigma_2=\sum_{j\geq0}a_j^2=\frac{2}{1+\gamma},
 \qquad
 \sigma_4=\sum_{j\geq0}a_j^4
 =1+\frac{(1-\gamma)^4}{1-\gamma^4}.
\]
Covariance summation gives \(M_{d,2}=s_d\sigma_2\).  For independent uniform
points \(p,q\),
\[
 \E\norm p^4=\frac{d}{d+4},
 \qquad
 \E(p\mathbin{\cdot}q)^2=\frac{s_d^2}{d}.
\]
Expanding the square of \(\norm{\sum_j a_jp_j}^2\) and cancelling all mixed
terms with an odd factor yields
\[
 M_{d,4}
 =s_d\sigma_2^2
 -\frac{2d}{(d+2)(d+4)}\sigma_4.
\]
Substitution and simplification give Eq.\eqref{eq:fourth-moment}.
\end{proof}

For \(x\in\B\), define the ball potential
\[
 G_{d,\alpha}(x)=\E\norm{p-x}^{\alpha}.
\]
Rotational invariance shows that it depends on \(x\) only through \(\norm x\).
For \(0\leq r\leq1\) and any unit vector \(e\), let
\[
 g_{d,\alpha}(r)=G_{d,\alpha}(re),
 \qquad \norm e=1.
\]
The value of \(g_{d,\alpha}(r)\) does not depend on the choice of \(e\).  In particular,
\begin{equation}
 M_{d,\alpha}(\gamma)
 =\E G_{d,\alpha}(x_\gamma)
 =\E g_{d,\alpha}(\norm{x_\gamma}).
 \label{eq:moment-through-potential}
\end{equation}
For \(\alpha>2\), the following bounds control the radial potential.
\begin{lemma}
\label{lem:radial-upper-envelope}
\begin{equation}
 g_{d,\alpha}(r)
 \leq c_{d,\alpha}
 +(g_{d,\alpha}(1)-c_{d,\alpha})r^2,
 \qquad
 g_{d,\alpha}(1)\leq2^\alpha c_{d,\alpha}.
 \label{eq:radial-upper-envelope}
\end{equation}
\end{lemma}

\begin{proof}
Appendix~\ref{app:foundational-insertion-inputs} proves that
\((g_{d,\alpha}(r)-c_{d,\alpha})/r^2\) is increasing and derives the endpoint
estimate by polar integration over a tangent sub-ball.
\end{proof}

Independence and
the geometric series also give
\begin{equation}
 \E\norm{x_\gamma}^2
 =s_d\frac{1-\gamma}{1+\gamma}.
 \label{eq:state-second-moment}
\end{equation}

\section{Stationary optimization under uniform random input}
\label{sec:stationary-optimization}

This section determines the optimal fixed memory parameter under stationary
uniform input.  It establishes the transition at power one, solves the
quadratic and quartic cases, extends the moment calculation to every even
power, and analyzes the local change at the star parameter.  The corresponding
numerical comparisons appear in Section~\ref{sec:numerical-stationary-worst-case}.

\subsection{Transition at power one}
\label{sec:stationary-phase}

This subsection determines whether the stationary objective is minimized at the
star parameter \(\gamma=1\) or at an interior memory parameter.  The answer
changes at \(\alpha=1\), and the proof supplies the endpoint estimates used in
the later finite-size analysis.  Figure~\ref{fig:uniform-adversarial-tradeoff}
and Table~\ref{tab:robustness-cost} later compare this stationary choice with
the worst-case choice.

Define the stationary objective on the closed parameter interval by
\begin{equation}
 \Phi_{d,\alpha}(\gamma)
 =h_\alpha(\gamma)M_{d,\alpha}(\gamma),
 \qquad 0\leq\gamma<1,
 \qquad
 \Phi_{d,\alpha}(1)=c_{d,\alpha}.
 \label{eq:stationary-objective}
\end{equation}

\begin{theorem}
\label{thm:stationary-transition}
If \(0<\alpha\leq1\), the continuous extension
\(\Phi_{d,\alpha}\) has its unique minimum at \(\gamma=1\).
Consequently, the stationary objective on \([0,1)\) has an unattained
infimum.  If \(\alpha>1\), the global minimum is attained and every minimizer
belongs to
\[
 \frac12<\gamma<1.
\]
\end{theorem}

\begin{proof}
\noindent\emph{Boundary regime \(0<\alpha\leq1\).}
For \(0<\alpha<1\), one has
\(h_\alpha(\gamma)>\gamma+(1-\gamma)=1\) whenever
\(0<\gamma<1\).  Moreover,
\(M_{d,\alpha}(\gamma)>c_{d,\alpha}\) for every \(\gamma<1\).
For \(\alpha=1\), the first factor is equal to one and the strict inequality
for the second factor remains.  This proves the first assertion.

\smallskip
\noindent\emph{Interior regime \(\alpha>1\).}
Let \(\alpha>1\).  Both \(h_\alpha\) and \(M_{d,\alpha}\) decrease on
 \([0,1/2]\), and the first decrease is strict.  Hence \(\Phi_{d,\alpha}\)
 decreases strictly on that interval.  It remains to exhibit a point of
 \((1/2,1)\) whose value is below \(c_{d,\alpha}\).

\smallskip
\noindent\emph{The case \(1<\alpha<2\).}
Suppose first that \(1<\alpha<2\), and let \(\delta=1-\gamma\).  Direct
differentiation of the potential at the origin gives
\[
 D^2G_{d,\alpha}(0)=\alpha I_d.
\]
Moreover, Eq.\eqref{eq:state-second-moment} gives
\[
 \E\norm{x_{1-\delta}}^2
 =s_d\frac{\delta}{2-\delta}
 =\frac{s_d}{2}\delta+O(\delta^2).
\]
The fourth moment of this weighted sum is \(O(\delta^2)\).  A second-order
expansion of \(G_{d,\alpha}\) at the origin can therefore be averaged to obtain
\[
 M_{d,\alpha}(1-\delta)
 =c_{d,\alpha}+\frac{\alpha s_d}{4}\delta+o(\delta).
\]
Since
\[
 h_\alpha(1-\delta)
 =1-\alpha\delta+\delta^\alpha+O(\delta^2),
\]
and \(\delta^\alpha=o(\delta)\), it follows that
 \[
 \Phi_{d,\alpha}(1-\delta)-c_{d,\alpha}
 =\alpha\left(\frac{s_d}{4}-c_{d,\alpha}\right)\delta+o(\delta).
 \]
Here \(c_{d,\alpha}>s_d\) because \(\alpha<2\), so this difference is
negative for all sufficiently small positive \(\delta\).

\smallskip
\noindent\emph{The case \(\alpha=2\).}
For \(\alpha=2\), Eq.\eqref{eq:second-moment} gives
\[
 \Phi_{d,2}(1/2)
 =\frac12M_{d,2}(1/2)
 =\frac{2s_d}{3}<s_d=c_{d,2}.
\]

\smallskip
\noindent\emph{The case \(\alpha>2\).}
Finally, let \(\alpha>2\).  The stationary state at \(\gamma=1/2\) is
\[
 x_{1/2}=\sum_{j=0}^{\infty}2^{-j-1}p_j.
\]
By Eq.\eqref{eq:moment-through-potential}, the insertion moment at this
parameter is the expected radial potential evaluated at \(\norm{x_{1/2}}\).
Eqs.\eqref{eq:radial-upper-envelope} and
\eqref{eq:state-second-moment} therefore give
\begin{align*}
 M_{d,\alpha}(1/2)
 &=\E g_{d,\alpha}(\norm{x_{1/2}})\\
 &\leq c_{d,\alpha}
 +(g_{d,\alpha}(1)-c_{d,\alpha})\E\norm{x_{1/2}}^2\\
 &\leq c_{d,\alpha}\left(1+(2^\alpha-1)\frac{s_d}{3}\right).
\end{align*}
Multiplication by \(h_\alpha(1/2)=2^{1-\alpha}\) yields
\[
 \frac{\Phi_{d,\alpha}(1/2)}{c_{d,\alpha}}
 \leq
 2^{1-\alpha}\left(1-\frac{s_d}{3}\right)
 +\frac{2s_d}{3}<1.
 \]
The last inequality uses \(2^{1-\alpha}<1/2\) and \(s_d<1\).
The continuous extension attains its minimum on \([0,1]\).  Strict decrease excludes
\(\gamma=1/2\), and the strict comparison with the endpoint excludes
\(\gamma=1\).
\end{proof}

\begin{remark}
The theorem asserts existence and localization for \(\alpha>1\).  Uniqueness
for a general power remains open.
\end{remark}

\subsection{Exact optimization for squared edge length}
\label{sec:quadratic}

The quadratic case admits a complete stationary and finite-size analysis.
This subsection gives the exact objective, its unique minimizer, and the first
finite-size correction.  The dimension-independent stationary formula is the
common curve in Figure~\ref{fig:uniform-adversarial-tradeoff}(b)
and supplies the exact quadratic check in
Figure~\ref{fig:high-dimensional-first-order}.

By Eq.\eqref{eq:second-moment},
\begin{equation}
 \Phi_{d,2}(\gamma)
 =2s_d\frac{2\gamma^2-2\gamma+1}{1+\gamma}.
 \label{eq:quadratic-stationary}
\end{equation}

\begin{theorem}
\label{thm:quadratic-stationary-optimum}
The stationary quadratic objective has its unique minimum at
\begin{equation}
 \gamma
 =\frac{\sqrt{10}-2}{2},
 \label{eq:quadratic-minimizer}
\end{equation}
and its minimum is
\[
 4s_d(\sqrt{10}-3).
\]
\end{theorem}

\begin{proof}
After multiplication by its positive denominator, the derivative of
Eq.\eqref{eq:quadratic-stationary} has the sign of
\(2\gamma^2+4\gamma-3\).  This polynomial is strictly increasing on
\([0,1]\), is negative at zero, and is positive at one.  Its unique root is
Eq.\eqref{eq:quadratic-minimizer}; substitution gives the stated minimum.
\end{proof}

\begin{proposition}
\label{prop:quadratic-finite}
For the quadratic finite-size problem, let \(\gamma_N\) be the unique
minimizer.  Then
\begin{equation}
 \gamma_N
 =
 \frac{\sqrt{10}-2}{2}
 +\frac{11\sqrt{10}-70}{90N}
 +O(N^{-2}),
 \label{eq:quadratic-finite-minimizer}
\end{equation}
and
\begin{equation}
 \min_\gamma F_{d,2,N}(\gamma)
 =
 4s_d(\sqrt{10}-3)N
 +\frac{4s_d}{15}(7\sqrt{10}-20)
 +O(N^{-1}).
 \label{eq:quadratic-finite-value}
\end{equation}
\end{proposition}

\begin{proof}
Direct covariance summation gives the exact formula
\begin{equation}
 F_{d,2,N}(\gamma)
 =2s_dQ_2(\gamma)
 \left[N+\gamma\frac{1-\gamma^{2N}}{1-\gamma^2}\right],
 \qquad
 Q_2(\gamma)=\frac{\gamma^2+(1-\gamma)^2}{1+\gamma}.
 \label{eq:quadratic-exact-finite}
\end{equation}
In a fixed neighborhood of
\(\gamma_0=(\sqrt{10}-2)/2\), the term containing \(\gamma^{2N}\) is
exponentially small.  The remaining expression is
\(2s_d[NQ_2(\gamma)+b(\gamma)]\), where
\[
 b(\gamma)=
 \frac{\gamma(\gamma^2+(1-\gamma)^2)}
 {(1-\gamma)(1+\gamma)^2}.
\]
At \(\gamma_0\),
\[
 Q_2''(\gamma_0)=\frac{4\sqrt{10}}5,
 \qquad
 b'(\gamma_0)=-\frac{44}{45}+\frac{28\sqrt{10}}{45}.
\]
The implicit expansion of \(NQ_2'(\gamma)+b'(\gamma)=0\) gives
Eq.\eqref{eq:quadratic-finite-minimizer}.  Taylor expansion of the value gives
Eq.\eqref{eq:quadratic-finite-value}.
\end{proof}

\subsection{Fourth and higher even powers}
\label{sec:even-powers}

This subsection extends the exact-moment analysis beyond squared edge length.  At
power four the optimizer is characterized by an explicit polynomial, while a
cumulant recursion produces rational objectives for every higher even power.
The fourth-moment formula is used directly in
Figure~\ref{fig:high-dimensional-first-order}.

Combining Eq.\eqref{eq:reduction} and Eq.\eqref{eq:fourth-moment} gives
\begin{equation}
 \Phi_{d,4}(\gamma)
 =
 \frac{4d}{(d+2)(d+4)}
 \frac{
 \bigl(\gamma^4+(1-\gamma)^4\bigr)
 \bigl((d+3)(1+\gamma^2)-2\gamma^3\bigr)}
 {(1+\gamma)^2(1+\gamma^2)}.
 \label{eq:fourth-objective}
\end{equation}

\begin{proposition}
\label{prop:fourth-power}
The function in Eq.\eqref{eq:fourth-objective} is strictly convex on
\([0,1]\).  Its unique minimizer is the unique root in \((0,1)\) of
\begin{align}
 P_d(\gamma)={}&
 -3d-9
 +(8d+24)\gamma
 -(12d+39)\gamma^2 \notag\\
 &+(18d+69)\gamma^3
 -(13d+62)\gamma^4
 +(12d+51)\gamma^5 \notag\\
 &-(2d+22)\gamma^6
 +(2d+6)\gamma^7
 +(2d+4)\gamma^8
 -6\gamma^9.
 \label{eq:fourth-polynomial}
\end{align}
For \(d=1\), this root satisfies
\[
 \frac{8}{15}<\gamma<\frac{15}{28},
\]
and, for \(d\geq2\), it satisfies
\[
 \frac{10}{19}<\gamma<\frac{8}{15}.
\]
\end{proposition}

\begin{proof}
Differentiate Eq.\eqref{eq:fourth-objective}.  Its first derivative is a
positive factor times \(P_d\).  After multiplication by the positive
denominator of the second derivative, the numerator has the Bernstein
representation
\[
 \sum_{k=0}^{11}b_k(d)
 \binom{11}{k}\gamma^k(1-\gamma)^{11-k},
\]
whose complete coefficient vector is recorded in
Appendix~\ref{app:fourth-convexity-certificate}.  Every coefficient is positive
for \(d\geq1\), so
\(\Phi_{d,4}''(\gamma)>0\) on \([0,1]\).  The endpoint signs
\(P_d(0)<0<P_d(1)\) yield the unique root.  Direct exact evaluation of
\(P_d\) at the rational endpoints in the theorem gives the stated brackets.
\end{proof}

\paragraph{A recursion for all even moments.}

Let \(y_\gamma=p-x_\gamma\).  If \(\xi\) is one coordinate of a uniform point
in \(\B\), denote its scalar cumulants by \(\kappa_{2r}\).  For \(r\geq1\),
\[
 \E \xi^{2r}=\prod_{j=0}^{r-1}\frac{2j+1}{d+2j+2},
\]
and the coordinate moments of \(y_\gamma\), starting from
\(\E(y_\gamma)_1^0=1\), satisfy
\begin{equation}
 \E(y_\gamma)_1^{2m}
 =
 \sum_{r=1}^{m}
 \binom{2m-1}{2r-1}
 \kappa_{2r}
 \left(1+\frac{(1-\gamma)^{2r}}{1-\gamma^{2r}}\right)
 \E(y_\gamma)_1^{2m-2r}.
 \label{eq:cumulant-recursion}
\end{equation}

\begin{proposition}
For every \(m\geq1\),
\begin{equation}
 M_{d,2m}(\gamma)
 =
 \left(\prod_{j=0}^{m-1}\frac{d+2j}{2j+1}\right)
 \E(y_\gamma)_1^{2m}.
 \label{eq:even-moment}
\end{equation}
Consequently, \(M_{d,2m}\) and \(\Phi_{d,2m}\) are rational functions of
\(\gamma\).
\end{proposition}

\begin{proof}
The vector \(y_\gamma\) is an isotropic weighted sum of independent uniform
ball points.  Its coefficient power sum at order \(2r\) is
\(1+(1-\gamma)^{2r}/(1-\gamma^{2r})\), as used in
Eq.\eqref{eq:cumulant-recursion}.  Cumulants add under independence
and scale homogeneously, which proves that recursion.
Rotational invariance gives
\[
 \E\norm{y_\gamma}^{2m}
 =\left(\prod_{j=0}^{m-1}\frac{d+2j}{2j+1}\right)
 \E(y_\gamma)_1^{2m},
\]
which proves Eq.\eqref{eq:even-moment}.  Each cumulant is rational in \(d\),
and the coefficient power sums are rational in \(\gamma\), proving the last
assertion.
\end{proof}

\begin{remark}
Eq.\eqref{eq:cumulant-recursion} gives a rational objective for each even
power.  Uniqueness of its minimizer remains open in general.
\end{remark}

\subsection{Local behavior near the star parameter}
\label{sec:endpoint}

This subsection examines the stationary objective as \(\gamma\) approaches the
star value \(1\).  It identifies the threshold \(\alpha=3d+8\) at which the
local slope changes sign and describes the nearby critical branch and barrier.
Figure~\ref{fig:endpoint-bifurcation} later shows the branch, its curvature,
and the barrier-height asymptotics for \(d=2\).

Let \(\delta=1-\gamma\).  The endpoint expansion has the form
\begin{equation}
 \Phi_{d,\alpha}(1-\delta)-c_{d,\alpha}
 =
 \alpha\left(\frac{s_d}{4}-c_{d,\alpha}\right)\delta
 +o(\delta).
 \label{eq:endpoint-linear}
\end{equation}

\begin{theorem}
Let \(\alpha^\star=3d+8\).
The linear coefficient in Eq.\eqref{eq:endpoint-linear} changes sign at
\(\alpha=\alpha^\star\).  At the threshold,
\[
 \Phi_{d,\alpha^\star}(1-\delta)-c_{d,\alpha^\star}
 =
 -\frac{7d\alpha^\star}{16(d+2)}\delta^2
 +o(\delta^2).
\]
There is a stationary branch issuing from the endpoint such that
\begin{equation}
 \delta(\alpha)
 =\frac{\alpha-\alpha^\star}{14(d+2)}
 +O\bigl((\alpha-\alpha^\star)^2\bigr).
 \label{eq:endpoint-bifurcating-maximum}
\end{equation}
Along this branch,
\begin{equation}
 \Phi_{d,\alpha}(1-\delta(\alpha))-c_{d,\alpha}
 =\frac{d\alpha^\star}{448(d+2)^3}(\alpha-\alpha^\star)^2
 +O\bigl((\alpha-\alpha^\star)^3\bigr).
 \label{eq:endpoint-barrier-height}
\end{equation}
Its curvature satisfies
\begin{equation}
 \Phi_{d,\alpha}''(1-\delta(\alpha))
 =-\frac{7d\alpha^\star}{8(d+2)}+O(\alpha-\alpha^\star).
 \label{eq:endpoint-branch-curvature}
\end{equation}
For \(\alpha>\alpha^\star\) sufficiently close to \(\alpha^\star\), the point
\(1-\delta(\alpha)\) is a strict
local maximum.  For every \(\alpha>\alpha^\star\), the endpoint of the continuous
extension is a strict
one-sided local minimum, while the global minimum is attained in
\((1/2,1)\).
\end{theorem}

\begin{proof}
\noindent\emph{Local expansion.}
Let
\[
 D(\alpha,\delta)
 =\Phi_{d,\alpha}(1-\delta)-c_{d,\alpha}.
\]
The radial potential has the local expansion
\[
 g_{d,\alpha}(r)=c_{d,\alpha}+\frac\alpha2r^2
 +\frac{\alpha(\alpha-2)(d+\alpha-2)}{8(d+2)}r^4+O(r^6).
\]
Together with the exact even moments of \(x_{1-\delta}\) through order six
obtained from Eq.\eqref{eq:cumulant-recursion}, this gives, uniformly for
\(\alpha\) near \(\alpha^\star\),
\[
 D(\alpha,\delta)
 =a(\alpha)\delta+b(\alpha)\delta^2+O(\delta^3),
 \qquad
 a(\alpha)=\alpha\left(\frac{s_d}{4}-c_{d,\alpha}\right),
\]
where \(b\) is continuously differentiable and
\[
 b(\alpha^\star)=-\frac{7d\alpha^\star}{16(d+2)}.
\]
The same moment representation permits two differentiations with respect to
\(\delta\), and
\[
 \partial_\delta D(\alpha^\star,0)=0,
 \qquad
 \partial_\delta^2 D(\alpha^\star,0)=2b(\alpha^\star)<0.
\]
\smallskip
\noindent\emph{The implicit branch.}
The implicit-function theorem gives a unique local solution
\(\delta=\delta(\alpha)\) of \(\partial_\delta D(\alpha,\delta)=0\).  Since
\[
 a'(\alpha^\star)=\frac{d\alpha^\star}{16(d+2)^2},
\]
implicit differentiation yields
\[
 \delta'(\alpha^\star)
 =-\frac{a'(\alpha^\star)}{2b(\alpha^\star)}
 =\frac1{14(d+2)},
\]
which proves Eq.\eqref{eq:endpoint-bifurcating-maximum}.

\smallskip
\noindent\emph{Branch value and curvature.}
Substitution of this expansion into \(D\) gives
\[
 D(\alpha,\delta(\alpha))
 =-\frac{a'(\alpha^\star)^2}{4b(\alpha^\star)}
 (\alpha-\alpha^\star)^2+O((\alpha-\alpha^\star)^3),
\]
which is Eq.\eqref{eq:endpoint-barrier-height}.  Moreover,
\(\partial_{\delta\delta}D(\alpha,\delta(\alpha))
=2b(\alpha^\star)+O(\alpha-\alpha^\star)\), which proves
Eq.\eqref{eq:endpoint-branch-curvature}.  Thus this derivative is negative near
\(\alpha^\star\), so \(1-\delta(\alpha)\) is a strict local maximum for
\(\alpha>\alpha^\star\) sufficiently close to \(\alpha^\star\).

\smallskip
\noindent\emph{Local behavior at \(\gamma=1\).}
The sign of \(a(\alpha)\) gives the one-sided endpoint
minimum for \(\alpha>\alpha^\star\), and
Theorem~\ref{thm:stationary-transition} gives an interior value below the
endpoint value.
\end{proof}

\begin{remark}
For \(\alpha>3d+8\), the stationary objective is neither convex nor
single-well unimodal on the complete interval.  A general uniqueness result
requires a different argument.
\end{remark}

\section{Worst-case optimization}
\label{sec:worst-case}

This section determines the optimal fixed memory parameter under arbitrary
input sequences.  It first gives the exact solution in the range
\(0<\alpha\leq3\), then establishes the asymptotic optimized value as the
power grows.  The corresponding computations are collected in
Section~\ref{sec:numerical-stationary-worst-case}.

For a fixed initial state \(x_0\in\B\), let
\begin{equation}
 \mathcal A_\alpha(\gamma)
 =\limsup_{n\to\infty}\frac1n
 \sup_{p_1,\ldots,p_n\in\B}
 \sum_{i=1}^n\norm{p_i-x_{i-1}}^\alpha.
 \label{eq:article-two-adversarial-value}
\end{equation}
The supremum is taken before the limit superior.

\begin{lemma}
\label{lem:initial-state-independence}
For \(0\leq\gamma<1\), the value in
Eq.\eqref{eq:article-two-adversarial-value} is independent of the initial
state in \(\B\).
\end{lemma}

\begin{proof}
Drive states \(x_i\) and \(y_i\) by the same input sequence.  Subtracting the
recursions gives
\[
 x_i-y_i=\gamma^i(x_0-y_0).
\]
The reverse triangle inequality therefore bounds the difference between the
two insertion lengths at step \(i\) by
\(\gamma^{i-1}\norm{x_0-y_0}\).  For \(a,b\in[0,2]\),
\[
 |a^\alpha-b^\alpha|
 \leq
 \begin{cases}
  |a-b|^\alpha, &0<\alpha\leq1,\\
  \alpha2^{\alpha-1}|a-b|, &\alpha\geq1.
 \end{cases}
\]
The first inequality follows from subadditivity, and the second follows from
the mean value theorem.  Summing the resulting geometric bound shows that the
two accumulated costs differ by a constant independent of \(n\) and of the
input sequence.  Division by \(n\), followed by the supremum and limit
superior, proves the assertion.
\end{proof}

\subsection{Exact optimization through power three}
\label{sec:adversarial}

This subsection optimizes the fixed-parameter worst-case cost in the range
where the adversarial insertion cost is known exactly.  It gives the boundary
behavior for \(0<\alpha\leq1\) and a closed optimal parameter for
\(1<\alpha\leq3\).  Its comparison with the stationary optimizer appears in
Figure~\ref{fig:uniform-adversarial-tradeoff} and
Table~\ref{tab:robustness-cost}.

The insertion-cost identity in the next theorem was first obtained in
\cite{CastroGammaInsertion}.  We include its complete proof in
Appendix~\ref{app:foundational-insertion-inputs}, since it is the decisive
input to the edge-power optimization.

\begin{theorem}
\label{thm:exact-adversarial-insertion}
For \(0\leq\gamma<1\) and \(0<\alpha\leq3\),
\begin{equation}
 \mathcal A_\alpha(\gamma)
 =\left(\frac2{1+\gamma}\right)^\alpha.
 \label{eq:exact-adversarial-insertion}
\end{equation}
Antipodal alternation attains this asymptotic value.
\end{theorem}

\begin{proof}[Proof outline]
The alternating diameter input gives the lower bound.  A quadratic potential
proves the matching upper bound through power two, while a quartic potential
proves it at power three.  Interpolation between the last two estimates covers
the remaining powers.  Appendix~\ref{app:foundational-insertion-inputs} gives
the potentials, their factorizations, and the telescoping argument.
\end{proof}

It follows from Eq.\eqref{eq:reduction} that the edge-power objective is
\begin{equation}
 f_\alpha(\gamma)
 =h_\alpha(\gamma)
 \left(\frac{2}{1+\gamma}\right)^\alpha.
 \label{eq:adversarial-objective}
\end{equation}
Minimizing this edge-power objective gives the following parameter and value.

\begin{theorem}
Let \(0<\alpha\leq3\).  If \(0<\alpha\leq1\), then
\[
 \inf_{0\leq\gamma<1}f_\alpha(\gamma)=1,
\]
and the infimum is not attained.  If \(1<\alpha\leq3\), the unique global
minimum is attained at
\begin{equation}
 \gamma
 =\frac{2^{1/(\alpha-1)}}
 {1+2^{1/(\alpha-1)}},
 \label{eq:adversarial-minimizer}
\end{equation}
and the minimum is
\begin{equation}
 \left(1+2^{-\alpha/(\alpha-1)}\right)^{1-\alpha}.
 \label{eq:adversarial-minimum}
\end{equation}
\end{theorem}

\begin{proof}
Let \(t=\gamma/(1-\gamma)\).  Eq.\eqref{eq:adversarial-objective} becomes
\[
 2^\alpha\frac{1+t^\alpha}{(1+2t)^\alpha}.
\]
Its logarithmic derivative has the sign of
\[
 t^{\alpha-1}-2.
\]
The monotonicity of this expression gives the two regimes and
Eq.\eqref{eq:adversarial-minimum}.
\end{proof}

\subsection{Large-power asymptotics}
\label{sec:high-power}

This subsection moves beyond the exact range \(\alpha\leq3\) and determines the
leading behavior of the optimized worst-case value as \(\alpha\) grows.
The maximizing input is unknown in this range.  Periodic block inputs serve as
explicit adversarial witnesses and therefore provide lower bounds.  A
separation argument gives the matching upper bound.  Together they determine
the high-power asymptotic constant.  Figure~\ref{fig:high-power-adversarial}
later shows the finite-power gap between these bounds and the block size that
gives the strongest periodic certificate.

These inputs have a direct geometric role in the exponential recursion.  A
run of copies of \(u\) moves the recursive state toward \(u\).  The subsequent
switch to \(-u\) then produces long attachment edges, whose contributions are
emphasized by a large power.  The block length controls the tradeoff between
approaching each endpoint and switching frequently.  Choosing \(m\) of order
\(\log\alpha\) supplies the lower-bound scale needed below.

Let
\[
 v_\alpha
 =
 \inf_{0\leq\gamma<1}
 h_\alpha(\gamma)\mathcal A_\alpha(\gamma),
\]
with \(\mathcal A_\alpha\) defined in
Eq.\eqref{eq:article-two-adversarial-value}.

For \(0<\alpha\leq3\), antipodal alternation attains the supremum in
Eq.\eqref{eq:article-two-adversarial-value} for every fixed \(\gamma\).  This
is the \(m=1\) case below.  For high powers, the cases \(m\geq2\) are
computable lower-bound witnesses.  Their optimality among all input sequences
is not established.  Their insertion-cost formula also appears in
\cite{CastroGammaInsertion} and is derived here because
multiplication by the two-edge factor \(h_\alpha(\gamma)\) supplies the lower
bounds used in the high-power analysis.

\begin{proposition}
Fix a unit vector \(u\) and \(m\geq1\).  The \(m\)-block input is the
periodic sequence formed by \(m\) copies of \(u\), followed by \(m\) copies
of \(-u\).  Its mean insertion cost is
\begin{equation}
 \frac{2^\alpha}{m}
 \frac{1-\gamma^{\alpha m}}
 {(1-\gamma^\alpha)(1+\gamma^m)^\alpha}.
 \label{eq:block-insertion-cost}
\end{equation}
Consequently, for every \(m\geq1\),
\begin{equation}
 h_\alpha(\gamma)\mathcal A_\alpha(\gamma)
 \geq\frac{s_m(\gamma)}m,
 \qquad
 s_m(\gamma)
 =\frac{(2\gamma)^\alpha+[2(1-\gamma)]^\alpha}
 {(1+\gamma^m)^\alpha}.
 \label{eq:block-energy-lower}
\end{equation}
For \(m=1\), this is the antipodal-alternation bound
\begin{equation}
 h_\alpha(\gamma)\mathcal A_\alpha(\gamma)
 \geq f_\alpha(\gamma),
 \label{eq:alternation-energy-lower}
\end{equation}
where \(f_\alpha\) is defined in Eq.\eqref{eq:adversarial-objective}.
\end{proposition}

\begin{proof}
Represent the state as a scalar multiple of \(u\), and let \(z_0\) be its
value immediately before a block of \(m\) copies of \(u\).  Periodicity gives
\[
 z_0=-\frac{1-\gamma^m}{1+\gamma^m}.
\]
During the first half of the period, the successive insertion lengths are
\[
 \frac{2}{1+\gamma^m},
 \frac{2\gamma}{1+\gamma^m},
 \ldots,
 \frac{2\gamma^{m-1}}{1+\gamma^m}.
\]
The second half has the same list.  Summation proves
Eq.\eqref{eq:block-insertion-cost}.  Since
\[
 \frac{1-\gamma^{\alpha m}}{1-\gamma^\alpha}\geq1,
\]
multiplication by \(h_\alpha(\gamma)\) gives
Eq.\eqref{eq:block-energy-lower}.  The case \(m=1\) gives
Eq.\eqref{eq:alternation-energy-lower}.
\end{proof}

The matching upper bound uses the geometry of the recursion at
\(\gamma=1/2\).  Let \(\ell_i=\norm{p_i-x_{i-1}}\).

\begin{lemma}
Fix \(0<\varepsilon\leq1/8\), and call \(i\) a large index when
\(\ell_i/2>1-\varepsilon\).  Set
\begin{equation}
 g_\varepsilon
 =1+\log_2\frac{1-\sqrt{2\varepsilon-\varepsilon^2}}
 {2\varepsilon}.
 \label{eq:large-index-separation}
\end{equation}
Any two large indices are separated by at least \(g_\varepsilon\).
Their upper asymptotic density is therefore at most
\(g_\varepsilon^{-1}\).
\end{lemma}

\begin{proof}
If \(i\) is large, the reverse triangle inequality gives
\[
 \norm{x_{i-1}}>1-2\varepsilon.
\]
The parallelogram identity and \(x_i=(x_{i-1}+p_i)/2\) give
\[
 \norm{x_i}^2+\frac{\ell_i^2}{4}
 =\frac{\norm{x_{i-1}}^2+\norm{p_i}^2}{2}
 \leq1.
\]
Hence \(\norm{x_i}<\sqrt{2\varepsilon-\varepsilon^2}\).  The recursion also implies
\[
 1-\norm{x_{j+1}}
 \geq\frac{1-\norm{x_j}}2.
\]
If \(i+s\) is the next large index, recovery from
\(\norm{x_i}<\sqrt{2\varepsilon-\varepsilon^2}\) to
\(\norm{x_{i+s-1}}>1-2\varepsilon\) requires
\[
 2^{-(s-1)}(1-\sqrt{2\varepsilon-\varepsilon^2})<2\varepsilon.
\]
This gives the separation in Eq.\eqref{eq:large-index-separation}, except
possibly when equality holds at the threshold.  This boundary case does not
affect the density bound obtained by counting gaps.
\end{proof}

At \(\gamma=1/2\), the contribution to the edge-power cost at step \(i\) is
\(2(\ell_i/2)^\alpha\).
The preceding lemma gives, for every \(0<\varepsilon\leq1/8\),
\begin{equation}
 v_\alpha
 \leq
 2\left((1-\varepsilon)^\alpha+\frac1{g_\varepsilon}\right).
 \label{eq:high-power-upper}
\end{equation}

\begin{theorem}
\label{thm:high-power-adversarial}
The optimized high-power adversarial value satisfies
\begin{equation}
 \lim_{\alpha\to\infty}
 (\log\alpha)v_\alpha
 =2\log2.
 \label{eq:exact-high-power-constant}
\end{equation}
\end{theorem}

\begin{proof}[Proof outline]
The upper bound follows by inserting a shrinking separation threshold into Eq.\eqref{eq:high-power-upper}.  For the matching lower bound, a periodic block of length of order \(\log\alpha\) is combined with the alternating input.  The parameter interval is divided according to which construction supplies the required lower bound.  Appendix~\ref{app:high-power-proof} gives the complete argument.
\end{proof}

\section{Numerical stationary and worst-case comparisons}
\label{sec:numerical-stationary-worst-case}

This section collects the computations associated with the stationary results
of Section~\ref{sec:stationary-optimization} and the worst-case results of
Section~\ref{sec:worst-case}.
The first subsection compares parameter choices under the two input models.
The second illustrates the local star threshold and the large-power limit.

\subsection{Balancing uniform and worst-case performance}

The stationary and worst-case criteria in Sections~\ref{sec:stationary-phase}
and~\ref{sec:adversarial} select different memory parameters.  The following
relative losses put the two criteria on the same scale for
\(1<\alpha\leq3\):
\[
 R_{\rm unif}(\gamma)
 =\frac{\Phi_{d,\alpha}(\gamma)}
 {\min_{\eta}\Phi_{d,\alpha}(\eta)}-1,
 \qquad
 R_{\rm adv}(\gamma)
 =\frac{f_\alpha(\gamma)}
 {\min_{\eta}f_\alpha(\eta)}-1.
\]
The balanced parameter is
\[
 \gamma_{\rm bal}
 \in\operatorname*{argmin}_{0\leq\gamma\leq1}
 \max\{R_{\rm unif}(\gamma),R_{\rm adv}(\gamma)\}.
\]
It quantifies the price of using one fixed parameter when the input model is
uncertain.

Figure~\ref{fig:uniform-adversarial-tradeoff} shows the complete integer range
\(1\leq d\leq100\), while Table~\ref{tab:robustness-cost} gives selected
quantitative cases.

\begin{figure}[H]
 \centering
 \includegraphics[width=\textwidth]{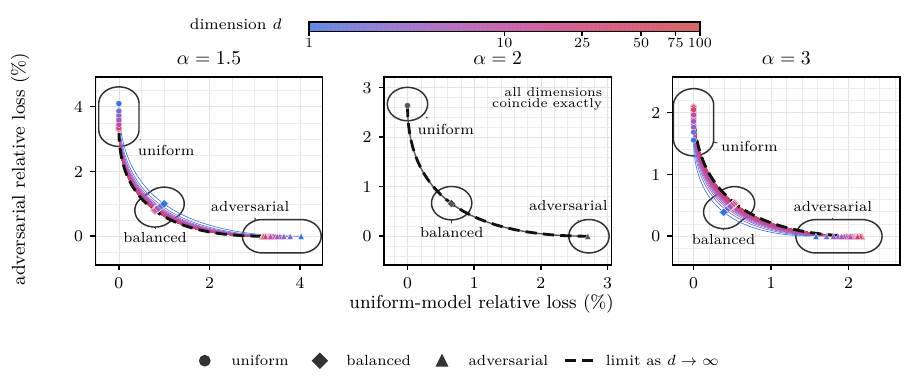}
 \caption{Compromise between performance under uniform random input and
 worst-case performance. Each curve traces the uniform-model loss horizontally
 and the adversarial loss vertically as \(\gamma\) varies; the balanced diamond
 identifies the smallest achievable maximum of the two losses. Circles and
 triangles mark the uniform-optimal and adversarial-optimal endpoints. The
 solid curves represent \(d=1,\ldots,100\), colored by
 \(t_d=\log(d)/\log(100)\).
 The black dashed curve is the limit as \(d\to\infty\) established in
 Section~\ref{sec:high-dimensional}.  For \(\alpha=2\), all dimensions
 coincide by Eq.\eqref{eq:quadratic-stationary}.}
 \label{fig:uniform-adversarial-tradeoff}
\end{figure}

\begin{table*}[t]
\centering
\caption{Price of using one fixed parameter when the input model is uncertain. The last column is the quantity to compare: it gives the smallest achievable maximum of the uniform and adversarial relative losses. The three parameter columns give the uniform-optimal, adversarial-optimal, and balanced choices.}
\label{tab:robustness-cost}
\small
\setlength{\tabcolsep}{4.2pt}
\begin{tabular}{ccccc rrr}
\toprule
$d$ & $\alpha$ & $\gamma_{\mathrm{unif}}$ & $\gamma_{\mathrm{adv}}$ & $\gamma_{\mathrm{bal}}$ & \shortstack{Uniform loss at\\$\gamma_{\mathrm{adv}}$} & \shortstack{Adversarial loss at\\$\gamma_{\mathrm{unif}}$} & \shortstack{Maximum loss at\\$\gamma_{\mathrm{bal}}$} \\
\midrule
2 & 1.25 & 0.74080 & 0.94118 & 0.85147 & 3.527\% & 4.548\% & 1.003\% \\
2 & 1.5 & 0.64247 & 0.80000 & 0.72058 & 3.784\% & 3.861\% & 0.941\% \\
2 & 2 & 0.58114 & 0.66667 & 0.62278 & 2.705\% & 2.633\% & 0.658\% \\
2 & 2.5 & 0.55821 & 0.61351 & 0.58521 & 2.085\% & 2.033\% & 0.509\% \\
2 & 3 & 0.54600 & 0.58579 & 0.56549 & 1.718\% & 1.682\% & 0.421\% \\
\midrule
10 & 1.25 & 0.76339 & 0.94118 & 0.86171 & 2.800\% & 3.628\% & 0.798\% \\
10 & 1.5 & 0.65080 & 0.80000 & 0.72499 & 3.348\% & 3.446\% & 0.838\% \\
10 & 2 & 0.58114 & 0.66667 & 0.62278 & 2.705\% & 2.633\% & 0.658\% \\
10 & 2.5 & 0.55593 & 0.61351 & 0.58398 & 2.284\% & 2.213\% & 0.556\% \\
10 & 3 & 0.54294 & 0.58579 & 0.56386 & 2.024\% & 1.962\% & 0.493\% \\
\bottomrule
\end{tabular}
\end{table*}

The balanced choice keeps the larger loss close to one percent throughout the
reported range and below one percent for all rows except
\((d,\alpha)=(2,1.25)\), where it is \(1.003\%\).  The stationary expectations
were evaluated by deterministic geometric quadrature at three nested
resolutions.  The quadratic rows were also checked against the exact formula.
Across the nonquadratic cases, the largest relative change between the two
finest resolutions was \(9.2\times10^{-5}\).  The corresponding conservative
uncertainty in the reported losses is below \(0.019\) percentage points.

\subsection{The local star threshold and the large-power limit}

The first computation illustrates the local transition established in
Section~\ref{sec:endpoint}.  Figure~\ref{fig:endpoint-bifurcation} shows the
critical branch and barrier as the power crosses the threshold.

Immediately above the threshold, a narrow barrier separates the one-sided
endpoint minimum from the interior descent.  The branch and barrier ratios
converge to the coefficients in the theorem.  The computed branch curvature
also approaches \(-6.125\), as predicted by
Eq.\eqref{eq:endpoint-branch-curvature}, and identifies the emerging critical
point as a strict local maximum.

\begin{figure}[!t]
 \centering
 \includegraphics[width=\textwidth]{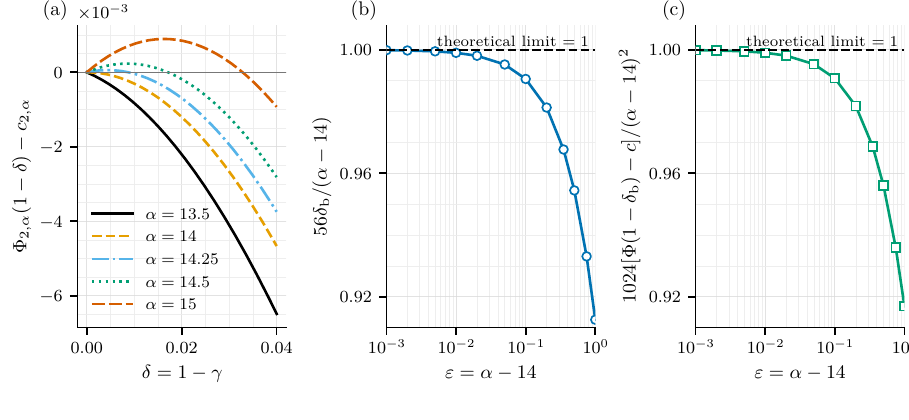}
 \caption{Birth of the interior stationary branch for \(d=2\), where
 \(\alpha^\star=14\).  (a) shows how the objective gap changes near
 \(\delta=1-\gamma=0\) across the threshold.  (b) tests the normalized
 branch coefficient \(56\delta(\alpha)/(\alpha-14)\), and (c) plots the
 normalized barrier coefficient
 \(1024[\Phi_{2,\alpha}(1-\delta(\alpha))-c_{2,\alpha}]/(\alpha-14)^2\).
 Both (b) and (c) use \(\varepsilon=\alpha-14\) on a logarithmic scale. Their
 convergence to the common reference value \(1\) is the visual test of the two
 asymptotic coefficients as \(\varepsilon\to0^+\), in agreement with
 Eq.\eqref{eq:endpoint-bifurcating-maximum} and
 Eq.\eqref{eq:endpoint-barrier-height}.}
 \label{fig:endpoint-bifurcation}
\end{figure}

The second computation concerns the convergence in
Theorem~\ref{thm:high-power-adversarial}.  The finite-power band in
Figure~\ref{fig:high-power-adversarial}(a) remains wide even at very large
powers.  At \(\alpha=10^8\), the computed theorem bounds for the normalized
value are \(0.818105\) and \(1.136188\).  Thus the figure illustrates both the
asymptotic agreement proved by Eq.\eqref{eq:exact-high-power-constant} and its
slow finite-power approach.  The oscillations of the best block size
\(m^*\) within this witness family in
Figure~\ref{fig:high-power-adversarial}(b) come from the integer optimization
over periodic lower-bound witnesses.

\begin{figure}[!t]
 \centering
 \includegraphics[width=\textwidth]{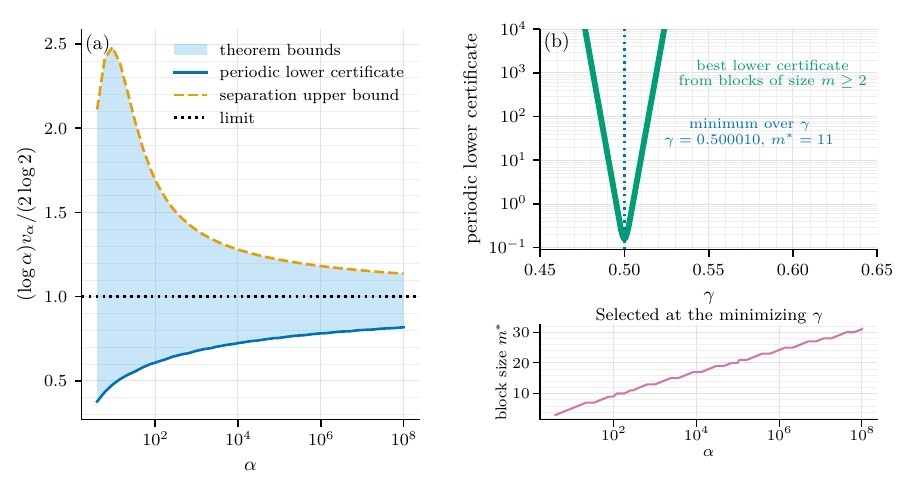}
 \caption{High-power adversarial bounds and their slow convergence.  An
 \(m\)-block input periodically repeats \(m\) copies of \(u\) followed by
 \(m\) copies of \(-u\) and serves as an explicit lower-bound witness.  (a)
 shows the proved lower and upper
 constructions approaching \(2\log2/\log\alpha\), with their gap measuring
 finite-power uncertainty.  At \(\alpha=256\), the upper graph in (b) shows
 the strongest certificate obtained from nonalternating block inputs.  Its
 ordinate is precisely
\(\max_{m\in\mathbb N,\,m\geq2}s_m(\gamma)/m\).  The minimum over \(\gamma\)
occurs near \(\gamma=0.500010\), where the best lower-bound witness within
this block family has size \(m^*=11\).
 The case \(m=1\) is antipodal alternation.  Its bound is below the maximum
 supplied by blocks with \(m\geq2\) throughout the displayed
 \(\gamma\)-interval at \(\alpha=256\), so it never determines the plotted
maximum.  This maximum remains a lower bound on the unknown adversarial
value.  The lower part records the block size \(m^*\) that maximizes this
restricted certificate as \(\alpha\) varies.}
 \label{fig:high-power-adversarial}
\end{figure}

\section{Finite-size optimization}
\label{sec:finite-regimes}

This section determines how the finite-size optimizer approaches its
stationary target.  The analysis covers the boundary regime \(0<\alpha<1\),
fixed interior powers \(\alpha>1\), and the joint window
\(\alpha-1=O(1/\log N)\).  A geometric constant controls the initialization
term in the boundary and joint-window regimes.  The joint-window analysis is
organized into subcritical, critical, and supercritical cases, followed by a
comparison of the first corrections.

\subsection{Fixed powers below one}

For \(\alpha=1\), the edge-power objective equals the insertion cost.  This
critical finite-size case is included in the joint-window analysis below: the
specialization \(\varepsilon_N=0\) of
Theorem~\ref{thm:joint-subcritical} gives its optimizer and minimum-value
asymptotics.

For a unit vector \(e\), recall that
\[
 g_{d,\alpha}(r)=G_{d,\alpha}(re)=\E\norm{p-re}^{\alpha}
\]
and let \(\omega_j\) denote the volume of the unit ball in \(\R^j\), with
\(\omega_0=1\).  For \(R>0\) and \(-R\leq z\leq R\), define the cap volume
\begin{equation}
 \mathcal C_d(R,z)
 =\omega_{d-1}\int_z^R
 (R^2-u^2)^{(d-1)/2}\,du.
 \label{eq:cap-volume}
\end{equation}

The required constant is obtained from the distance between a fixed point and
a uniform point of the ball.  The event \(\norm{p-re}\leq t\) is the lens
\(B(0,1)\cap B(re,t)\).  Its standard decomposition into the two caps cut by
the common boundary hyperplane makes the potential and its radial derivative
explicit.

\begin{proposition}
Fix \(0<r\leq1\), and let \(T_r=\norm{p-re}\).  Its distribution function is
\begin{equation}
 \Pr(T_r\leq t)
 =
 \begin{cases}
  0, & t<0,\\
  t^d, & 0\leq t\leq1-r,\\
  \omega_d^{-1}
  \bigl(\mathcal C_d(1,a(r,t))
  +\mathcal C_d(t,b(r,t))\bigr),
  & 1-r<t<1+r,\\
  1, & t\geq1+r,
 \end{cases}
 \label{eq:conditional-distance-cdf}
\end{equation}
where
\begin{equation}
 a(r,t)=\frac{r^2+1-t^2}{2r},
 \qquad
 b(r,t)=\frac{r^2+t^2-1}{2r}.
 \label{eq:cap-plane-coordinates}
\end{equation}
Consequently,
\begin{equation}
 g_{d,\alpha}(r)
 =\alpha\int_0^{1+r}
 t^{\alpha-1}\bigl(1-\Pr(T_r\leq t)\bigr)\,dt.
 \label{eq:geometric-tail-moment}
\end{equation}

For
\begin{equation}
 \rho(r,t)
 =\frac{\sqrt{((1+t)^2-r^2)(r^2-(1-t)^2)}}{2r},
 \label{eq:intersection-radius}
\end{equation}
the radial derivative is
\begin{equation}
 g_{d,\alpha}'(r)
 =\alpha\frac{\omega_{d-1}}{\omega_d}
 \int_{1-r}^{1+r}
 t^{\alpha-1}\rho(r,t)^{d-1}\,dt.
 \label{eq:geometric-potential-derivative}
\end{equation}
\end{proposition}

\begin{proof}
The event \(T_r\leq t\) is the lens
\[
 B(0,1)\cap B(re,t).
\]
For \(t\leq1-r\), the second ball is contained in the first.  For
\(t\geq1+r\), it contains the first.  In the partial-overlap range, the two
boundary spheres meet in a hyperplane whose coordinates relative to the two
centers are \(a(r,t)\) and \(b(r,t)\).  The lens is the union of the two caps
in Eq.\eqref{eq:conditional-distance-cdf}.

Eq.\eqref{eq:geometric-tail-moment} follows from the tail integral for the
bounded nonnegative random variable \(T_r^\alpha\).  To differentiate it, let
\[
 I_d(r,t)=\operatorname{vol}\bigl(B(0,1)\cap B(re,t)\bigr).
\]
In the partial-overlap range, the common spherical cross-section is a
\((d-1)\)-ball of radius \(\rho(r,t)\).  More explicitly,
Eq.\eqref{eq:conditional-distance-cdf} gives
\[
 I_d(r,t)=\mathcal C_d(1,a(r,t))+\mathcal C_d(t,b(r,t)).
\]
Since
\[
 \partial_z\mathcal C_d(R,z)
 =-\omega_{d-1}(R^2-z^2)^{(d-1)/2},
\]
and
\[
 1-a(r,t)^2=t^2-b(r,t)^2=\rho(r,t)^2,
 \qquad \partial_ra(r,t)+\partial_rb(r,t)=1,
\]
differentiation of the two cap volumes gives
\[
 \frac{\partial I_d}{\partial r}(r,t)
 =-\omega_{d-1}\rho(r,t)^{d-1}.
\]
At \(t=1-r\) and \(t=1+r\), the radius \(\rho(r,t)\) vanishes and the
adjacent formulas for the lens volume agree.  Thus the moving-boundary terms
in the piecewise tail integral cancel.  Differentiating the remaining
integrand gives
Eq.\eqref{eq:geometric-potential-derivative}.
\end{proof}

Define
\begin{equation}
 H_{d,\alpha}
 =2\int_0^1\frac{1-r^d}{r}
 \bigl(g_{d,\alpha}(r)-c_{d,\alpha}\bigr)\,dr.
 \label{eq:H-constant}
\end{equation}
Since \(g_{d,\alpha}(0)=c_{d,\alpha}\), integration in the opposite order and
Eq.\eqref{eq:geometric-potential-derivative} give
\begin{equation}
 H_{d,\alpha}
 =2\alpha\frac{\omega_{d-1}}{\omega_d}
 \int_0^1\int_{1-r}^{1+r}
 \left(-\log r-\frac{1-r^d}{d}\right)
 t^{\alpha-1}\rho(r,t)^{d-1}\,dt\,dr.
 \label{eq:H-geometric-integral}
\end{equation}
This representation is positive and applies to every real \(\alpha>0\).
For every even power \(\alpha=2m\), expansion of
\(\norm{p-re}^{2m}\) reduces the same constant to a finite sum of radial
moments.
Thus \(H_{d,\alpha}\) is the leading accumulated initialization coefficient
in the boundary scaling limits below.

\begin{theorem}
\label{thm:sublinear-boundary}
Let \(d\geq1\) and \(0<\alpha<1\).  For every global minimizer \(\gamma_N\)
of \(F_{d,\alpha,N}\),
\begin{equation}
 N^{1/(\alpha+1)}(1-\gamma_N)
 \longrightarrow
 \left(
  \frac{H_{d,\alpha}}{2\alpha c_{d,\alpha}}
 \right)^{1/(\alpha+1)}.
 \label{eq:sublinear-boundary-location}
\end{equation}
Moreover,
\begin{equation}
 \min_{0\leq\gamma\leq1}F_{d,\alpha,N}(\gamma)
 =Nc_{d,\alpha}
 +(\alpha+1)c_{d,\alpha}^{1/(\alpha+1)}
 \left(\frac{H_{d,\alpha}}{2\alpha}\right)^{\alpha/(\alpha+1)}
 N^{1/(\alpha+1)}
 +o\bigl(N^{1/(\alpha+1)}\bigr),
 \label{eq:sublinear-boundary-value}
\end{equation}
\end{theorem}

\begin{proof}[Proof outline]
Letting \(\delta=1-\gamma\), the objective is expanded uniformly on the scale \(\delta=N^{-1/(\alpha+1)}y\).  The resulting coercive profile has a unique minimizer, and localization excludes smaller and larger scales.  Appendix~\ref{app:finite-fixed-proofs} supplies the uniform estimates and the compactness argument.
\end{proof}

\subsection{Fixed powers above one}

For \(\alpha>0\), let
\[
 M_{d,\alpha,t}(\gamma)=\E\norm{p_{t+1}-x_t}^{\alpha}.
\]
The difference
\(M_{d,\alpha,t}(\gamma)-M_{d,\alpha}(\gamma)\) is the initialization bias
remaining at time \(t\).
For \(\alpha>1\), define the transient correction by
\begin{equation}
 \Psi_{d,\alpha}(\gamma)
 =h_\alpha(\gamma)\sum_{t=0}^{\infty}
 \bigl(M_{d,\alpha,t}(\gamma)-M_{d,\alpha}(\gamma)\bigr).
 \label{eq:transient-correction}
\end{equation}
The series and its first two derivatives converge uniformly on every compact
subset of \((0,1)\).  This regularity is sufficient for the following result.

For \(\alpha>1\), let
\[
 \mathcal M=\operatorname*{argmin}_{0\leq\gamma\leq1}
 \Phi_{d,\alpha}(\gamma),
 \qquad
 \mathcal S=\operatorname*{argmin}_{\gamma\in\mathcal M}
 \Psi_{d,\alpha}(\gamma).
\]
Hence \(\mathcal S\) is the subset of stationary minimizers selected by the
order-one transient correction.  Thus the finite minimization has a
two-scale structure: the order-\(N\) term restricts the candidates to
\(\mathcal M\), and the order-one term selects \(\mathcal S\).  At a unique
nondegenerate stationary minimizer, the same expansion determines the
order-\(N^{-1}\) displacement.

\begin{theorem}
\label{thm:finite-perturbation}
Let \(d\geq1\) and \(\alpha>1\).
Then, on every compact interval \(K\) contained in \((0,1)\),
\begin{equation}
 F_{d,\alpha,N}
 =N\Phi_{d,\alpha}+\Psi_{d,\alpha}+O_K(\rho_K^N)
 \quad\text{in }C^2(K)
 \label{eq:finite-perturbation-identity}
\end{equation}
for some \(\rho_K<1\).  Every global minimizer \(\gamma_N\) of
\(F_{d,\alpha,N}\) satisfies
\[
 \operatorname{dist}(\gamma_N,\mathcal S)\longrightarrow0,
\]
and
\begin{equation}
 \min_\gamma F_{d,\alpha,N}(\gamma)
 =N\min_\gamma\Phi_{d,\alpha}(\gamma)
 +\min_{\gamma\in\mathcal M}\Psi_{d,\alpha}(\gamma)+o(1).
 \label{eq:finite-selection-value}
\end{equation}
\end{theorem}

\begin{corollary}
\label{cor:finite-nondegenerate}
If \(\mathcal M=\{\gamma_*\}\) and
\(\Phi_{d,\alpha}''(\gamma_*)>0\), the finite global minimizer is
unique for all sufficiently large \(N\), and
\begin{equation}
 \gamma_N
 =\gamma_*
 -\frac{\Psi_{d,\alpha}'(\gamma_*)}
 {N\Phi_{d,\alpha}''(\gamma_*)}
 +o(N^{-1}),
 \label{eq:finite-nondegenerate-location}
\end{equation}
\begin{equation}
 \min_\gamma F_{d,\alpha,N}(\gamma)
 =N\Phi_{d,\alpha}(\gamma_*)
 +\Psi_{d,\alpha}(\gamma_*)
 -\frac{\Psi_{d,\alpha}'(\gamma_*)^2}
 {2N\Phi_{d,\alpha}''(\gamma_*)}
 +o(N^{-1}).
 \label{eq:finite-nondegenerate-value}
\end{equation}
\end{corollary}

\begin{proof}[Proof outline for Theorem~\ref{thm:finite-perturbation} and
Corollary~\ref{cor:finite-nondegenerate}]
Geometric contraction of the state recursion gives
Eq.\eqref{eq:finite-perturbation-identity}.  The leading term confines finite
minimizers to \(\mathcal M\), and the order-one term selects \(\mathcal S\).
A nondegenerate stationary minimum yields the displacement and value
corrections.  Appendix~\ref{app:finite-fixed-proofs} contains the uniform
estimates and the full selection argument.
\end{proof}

\begin{remark}
The preceding theorem uses regularity on an interior compact set containing
all finite minimizers for large \(N\).  At \(\gamma=0\), the factor
\(h_\alpha\) may limit the regularity of the complete objective when
\(\alpha\) is not an integer.  The finite critical case \(\alpha=1\) is
covered by the insertion-cost analysis and is outside this perturbation
theorem.
\end{remark}

\subsection{Transition window near \texorpdfstring{$\alpha=1$}{alpha=1}}
\label{sec:joint-finite-window}

The fixed-power regimes meet on the scale
\((\alpha-1)\log N=O(1)\).  Three optimizer scales occur as the limit of
\((\alpha_N-1)\log N\) crosses the threshold identified below: a square-root
scale, a Lambert-corrected critical scale, and the stationary scale.  Within
the last regime, a second transition determines whether the transient or the
stationary correction controls the refined location.

We first establish the stationary endpoint expansion, the transient bound, and
its differentiated form used in these regimes.  Let \(\delta=1-\gamma\).
The proof architecture separates four tasks.  Lemma~\ref{lem:uniform-stationary-endpoint}
gives the stationary endpoint expansion, Lemma~\ref{lem:uniform-transient}
gives the transient value, and Lemma~\ref{lem:uniform-differentiated-transient}
together with Corollary~\ref{cor:differentiated-transient-windows} controls its
first two derivatives on moving intervals.  Lemma~\ref{lem:joint-global-localization}
then transfers the local profiles to global minimizers.  The three regimes
differ only in the scale on which these inputs are balanced.

For \(\alpha\) near one, set
\begin{align}
 q_{d,\alpha}
 &=1-\frac{s_d}{4c_{d,\alpha}}
 =\frac{3d+8-\alpha}{4(d+2)},
 \label{eq:joint-D-q}\\
 \mu_{d,\alpha}
 &=\frac{\alpha s_d}{8}
 +\frac{\alpha(\alpha-2)(\alpha+d-2)s_d}{32(d+2)},
 \label{eq:joint-m2-A}\\
 \nu_{d,\alpha}
 &=\mu_{d,\alpha}-\frac{\alpha^2s_d}{4}
 +\frac{\alpha(\alpha-1)c_{d,\alpha}}2.
 \label{eq:joint-K}
\end{align}
Here \(q_{d,\alpha}\) is the magnitude of the linear endpoint coefficient,
\(\mu_{d,\alpha}\) is the quadratic coefficient of \(M_{d,\alpha}\), and
\(\nu_{d,\alpha}\) is the quadratic coefficient of \(\Phi_{d,\alpha}\).

\begin{lemma}
\label{lem:uniform-stationary-endpoint}
Fix \(d\geq1\) and \(0<\eta<1/2\).  There are
\(\delta_0=\delta_0(d,\eta)>0\) and \(C=C(d,\eta)<\infty\) such that,
uniformly for \(\alpha\in[1-\eta,1+\eta]\) and
\(0<\delta\leq\delta_0\),
\begin{equation}
 M_{d,\alpha}(1-\delta)
 =c_{d,\alpha}+\frac{\alpha s_d}{4}\delta
 +\mu_{d,\alpha}\delta^2+R_{d,\alpha}(\delta),
 \label{eq:joint-M-expansion}
\end{equation}
where
\begin{equation}
 \left|\partial_\delta^kR_{d,\alpha}(\delta)\right|
 \leq C\delta^{3-k},
 \qquad k=0,1,2.
 \label{eq:joint-M-remainder}
\end{equation}
Consequently,
\begin{align}
 \Phi_{d,\alpha}(1-\delta)-c_{d,\alpha}
 ={}&c_{d,\alpha}\delta^\alpha
 -\alpha c_{d,\alpha} q_{d,\alpha}\delta
 +\frac{\alpha s_d}{4}\delta^{\alpha+1}
 +\mu_{d,\alpha}\delta^{\alpha+2}\notag\\
 &+\nu_{d,\alpha}\delta^2+E_{d,\alpha}(\delta),
 \label{eq:joint-Phi-expansion}
\end{align}
with
\begin{equation}
 \left|\partial_\delta^kE_{d,\alpha}(\delta)\right|
 \leq C\delta^{3-k},
 \qquad k=0,1,2.
 \label{eq:joint-Phi-remainder}
\end{equation}
\end{lemma}
We retain \(\mu_{d,\alpha}\delta^{\alpha+2}\) to state the expansion
uniformly across \(\alpha=1\).  In the one-sided window \(\alpha\geq1\)
used below, it may be included in the remainder.

\begin{proof}[Proof outline]
The radial potential is expanded at the center and combined with exact second and fourth moments of the stationary state.  Concentration controls the Taylor remainder and its first two parameter derivatives uniformly near \(\alpha=1\).  The complete estimates appear in Appendix~\ref{app:finite-endpoint-proofs}.
\end{proof}

\begin{theorem}
\label{thm:near-one-stationary-uniqueness}
For every fixed \(d\geq1\), there is \(\varepsilon_d>0\) such that the
stationary objective has a unique global minimizer
\(\gamma_{d,\alpha}\) whenever \(1<\alpha<1+\varepsilon_d\).  If
\(\varepsilon=\alpha-1\) and
\begin{equation}
 \delta_{d,\alpha}
 =q_{d,\alpha}^{1/\varepsilon},
 \label{eq:near-one-stationary-scale}
\end{equation}
then
\begin{equation}
 1-\gamma_{d,\alpha}
 =\delta_{d,\alpha}
 \left(1+O_d\left(\frac{\delta_{d,\alpha}}{\varepsilon}\right)\right)
 \label{eq:near-one-stationary-location}
\end{equation}
as \(\alpha\to1^+\).  Moreover,
\begin{equation}
 \Phi_{d,\alpha}''(\gamma_{d,\alpha})
 =\frac{\alpha c_{d,\alpha} q_{d,\alpha}\varepsilon}
 {\delta_{d,\alpha}}(1+o(1)).
 \label{eq:near-one-stationary-curvature}
\end{equation}
\end{theorem}

\begin{proof}[Proof outline]
Lemma~\ref{lem:uniform-stationary-endpoint} localizes every minimizer near the
endpoint and identifies the relevant scale.  The rescaled derivative has
exactly one zero, at which the objective has positive curvature, while uniform
separation excludes the remainder of the interval.
Appendix~\ref{app:finite-endpoint-proofs} contains the full localization
argument.
\end{proof}

For
\[
 B_{d,\alpha,N}(\gamma)
 =\sum_{t=0}^{N-1}
 \bigl(M_{d,\alpha,t}(\gamma)-M_{d,\alpha}(\gamma)\bigr),
\]
let \(B_{d,\alpha}\) denote the corresponding infinite sum.  For
\(\alpha>1\),
\[
 \Psi_{d,\alpha}(\gamma)
 =h_\alpha(\gamma)B_{d,\alpha}(\gamma).
\]

\begin{lemma}
\label{lem:uniform-transient}
For fixed \(d\) and \(0<\eta<1/2\), there is
\(c_{\mathrm{tr}}=c_{\mathrm{tr}}(d,\eta)\) such that,
uniformly for \(\alpha\in[1-\eta,1+\eta]\) and
\(0<\delta\leq1/2\),
\begin{equation}
 \left|\delta B_{d,\alpha,N}(1-\delta)
 -\frac{H_{d,\alpha}}2\right|
 \leq c_{\mathrm{tr}}\left[
 \sqrt\delta\left(1+\log\frac1\delta\right)
 {}+e^{-N\delta/2}\right].
 \label{eq:joint-transient-bound}
\end{equation}
\end{lemma}

\begin{proof}[Proof outline]
The finite and stationary states are built from the same input sequence.  Splitting the time sum before and after the mixing scale separates the leading initialization term from an exponentially small tail.  Appendix~\ref{app:finite-endpoint-proofs} records the uniform bounds.
\end{proof}

For \(\delta=1-\gamma\), define the finite transient contribution
\begin{equation}
 \mathcal T_{d,\alpha,N}(\delta)
 =h_\alpha(1-\delta)B_{d,\alpha,N}(1-\delta),
 \qquad
 \ell_\delta=1+\log\frac1\delta.
 \label{eq:differentiated-transient-definition}
\end{equation}

\begin{lemma}
\label{lem:uniform-differentiated-transient}
Fix \(d\geq1\) and \(0<\eta<1/2\).  There is
\(C=C(d,\eta)<\infty\) such that, uniformly for
\[
 \alpha\in[1,1+\eta],
 \qquad N\geq1,
 \qquad 0<\delta\leq\frac12,
\]
\begin{align}
 \left|\mathcal T_{d,\alpha,N}(\delta)
 -\frac{H_{d,\alpha}}{2\delta}\right|
 &\leq C\ell_\delta
 +C\delta^{-1}e^{-N\delta/4},
 \label{eq:differentiated-transient-zero}\\
 \left|\mathcal T_{d,\alpha,N}'(\delta)
 +\frac{H_{d,\alpha}}{2\delta^2}\right|
 &\leq C\delta^{-3/2}
 +C\delta^{-2}(1+N\delta)e^{-N\delta/4},
 \label{eq:differentiated-transient-one}\\
 \left|\mathcal T_{d,\alpha,N}''(\delta)
 -\frac{H_{d,\alpha}}{\delta^3}\right|
 &\leq C\delta^{-5/2}\ell_\delta
 +C\delta^{-3}(1+N\delta)^3e^{-N\delta/4}.
 \label{eq:differentiated-transient-two}
\end{align}
\end{lemma}

\begin{proof}[Proof outline]
Differentiating the coupled finite-state representation isolates the leading term \(H_{d,\alpha}/(2\delta)\).  Moment estimates and the geometric tail control the first two derivatives uniformly.  The detailed calculation is given in Appendix~\ref{app:finite-endpoint-proofs}.
\end{proof}

\begin{corollary}
\label{cor:differentiated-transient-windows}
Let \(\alpha_N\in[1,1+\eta]\), and let \(I_N\subset(0,1/2]\) satisfy
\[
 \sup_{\delta\in I_N}\delta\longrightarrow0,
 \qquad
 N\inf_{\delta\in I_N}\delta\longrightarrow\infty.
\]
For \(k=0,1,2\),
\begin{equation}
 \sup_{\delta\in I_N}\delta^{k+1}
 \left|\partial_\delta^k\mathcal T_{d,\alpha_N,N}(\delta)
 -(-1)^k k!\frac{H_{d,\alpha_N}}{2\delta^{k+1}}\right|
 \longrightarrow0.
 \label{eq:differentiated-transient-windows}
\end{equation}
\end{corollary}

For the remainder of this subsection, let
\(\alpha_N=1+\varepsilon_N\), with \(\varepsilon_N>0\), and abbreviate
\[
 F_N=F_{d,\alpha_N,N},\qquad
 c_N=c_{d,\alpha_N},\qquad
 q_N=q_{d,\alpha_N},\qquad
 H_N=H_{d,\alpha_N}.
\]
These four abbreviations are used only in the moving-window analysis below.

Define
\begin{equation}
 a_d(\lambda)
 =c_{d,1}e^{-\lambda/2}+\frac{s_d}{4}-c_{d,1},
 \qquad
 \lambda^\star(d)=2\log\frac{4(d+2)}{3d+7}.
 \label{eq:joint-threshold}
\end{equation}
The stationary endpoint equation suggests the scale
\begin{equation}
 \delta_N=q_N^{1/\varepsilon_N}.
 \label{eq:joint-critical-scales}
\end{equation}

\begin{theorem}
\label{thm:joint-subcritical}
Suppose
\[
 \varepsilon_N\log N\longrightarrow\lambda,
\]
and assume \(\lambda<\lambda^\star(d)\).  Every global minimizer \(\gamma_N\)
of \(F_N\) satisfies
\begin{equation}
 \sqrt N(1-\gamma_N)
 \longrightarrow
 \sqrt{\frac{H_{d,1}}{2a_d(\lambda)}}.
 \label{eq:joint-location}
\end{equation}
Moreover,
\begin{equation}
 \min_\gamma F_N(\gamma)
 =Nc_N
 +\sqrt N\sqrt{2a_d(\lambda)H_{d,1}}
 +o(\sqrt N).
 \label{eq:joint-value}
\end{equation}
\end{theorem}

\begin{proof}
\noindent\emph{Limiting profile.}
For \(\gamma<1\),
\begin{equation}
 F_{d,\alpha,N}(\gamma)
 =N\Phi_{d,\alpha}(\gamma)
 +h_\alpha(\gamma)B_{d,\alpha,N}(\gamma).
 \label{eq:joint-decomposition}
\end{equation}
Let \(\delta=y/\sqrt N\).  Lemmas
\ref{lem:uniform-stationary-endpoint} and \ref{lem:uniform-transient} give,
locally uniformly for \(y>0\),
\begin{equation}
 \frac{F_N(1-y/\sqrt N)-Nc_N}
 {\sqrt N}
 \longrightarrow
 a_d(\lambda)y+\frac{H_{d,1}}{2y}.
 \label{eq:joint-profile}
\end{equation}
The identity \(c_{d,1}=d/(d+1)\) shows that
\(a_d(\lambda)>0\) exactly when \(\lambda<\lambda^\star(d)\).  In that range,
the limiting function has the unique minimizer
\[
 y_*=\sqrt{\frac{H_{d,1}}{2a_d(\lambda)}}.
\]

\smallskip
\noindent\emph{Global localization.}
A fixed trial value on this scale has excess \(O(\sqrt N)\).  For a
minimizing sequence, put \(r_N=1-\gamma_N\) and \(y_N=\sqrt N r_N\).
Parameters
bounded away from \(1\), as well as the finite endpoint \(\gamma=1\), have a
positive gap of order \(N\).  Thus \(r_N\to0\).  If
\(Nr_N\) stays bounded, the transient contribution has positive order
\(N\), so \(Nr_N\to\infty\).  If \(y_N\to0\), the normalized transient contribution diverges
as \(H_{d,1}/(2y_N)\), while the negative part allowed by
Eq.\eqref{eq:joint-Phi-expansion} is \(O(y_N)\).

If \(y_N\to\infty\), divide Eq.\eqref{eq:joint-Phi-expansion} by
\(r_N\).  For \(y_N\geq1\), its right-hand side has a positive lower
limit because \(\lambda<\lambda^\star(d)\).  The transient term is eventually
positive by Lemma~\ref{lem:uniform-transient}, and the normalized excess is
at least \(\kappa y_N\).  Hence every minimizing sequence stays in a compact
subset of \((0,\infty)\).  The locally uniform convergence in
Eq.\eqref{eq:joint-profile} proves Eqs.\eqref{eq:joint-location} and
\eqref{eq:joint-value}.
\end{proof}

\begin{remark}
For \(\lambda=0\), Eq.\eqref{eq:joint-threshold} gives
\(a_d(0)=s_d/4\), and Theorem~\ref{thm:joint-subcritical} recovers the
finite \(\alpha=1\) formulas.
\end{remark}

At \(\lambda=\lambda^\star(d)\), the subcritical linear coefficient
vanishes.  Moreover, Eq.\eqref{eq:joint-D-q} gives
\(q_{d,1}=e^{-\lambda^\star(d)/2}\), so the stationary root suggests the
scale \(\delta_N\) in Eq.\eqref{eq:joint-critical-scales}.

The next lemma ensures that the local critical and supercritical profiles
capture every global minimizer.

\begin{lemma}
\label{lem:joint-global-localization}
Suppose \(\alpha_N\to1\) and there is a trial parameter whose excess over
\(Nc_N\) is \(o(N)\).  Every global minimizer satisfies
\[
 1-\gamma_N\longrightarrow0,
 \qquad N(1-\gamma_N)\longrightarrow\infty.
\]
\end{lemma}

\begin{proof}[Proof outline]
Stationary separation away from the endpoint, together with the uniform transient bound, confines every global minimizer to the endpoint scale selected by a trial parameter.  Appendix~\ref{app:finite-window-proofs} gives the quantitative comparison.
\end{proof}

Let \(W_0\) denote the principal real branch of the Lambert function and
define
\begin{equation}
 w_N=W_0\left(
 \frac{H_N}{c_Nq_NN\varepsilon_N\delta_N^2}
 \right),
 \qquad
 \Delta_N
 =\delta_N e^{w_N/2}
 =\left(
 \frac{H_N}
 {c_Nq_NN\varepsilon_Nw_N}
 \right)^{1/2}.
 \label{eq:joint-critical-W-scale}
\end{equation}
The quantity \(w_N\) is the Lambert correction, and \(\Delta_N\) is the
corresponding critical scale.

\begin{theorem}
\label{thm:joint-critical-crossover}
If
\begin{equation}
 \varepsilon_N\log N\longrightarrow\lambda^\star(d),
 \label{eq:joint-critical-assumption}
\end{equation}
then every global minimizer satisfies
\begin{equation}
 \frac{1-\gamma_N}{\Delta_N}\longrightarrow1.
 \label{eq:joint-critical-location}
\end{equation}
Furthermore,
\begin{align}
 \min_\gamma F_N(\gamma)-Nc_N
 ={}&\frac{H_N}{\Delta_N}
 \frac{w_N-1}{w_N}\notag\\
 &+o\left[
 \frac{H_N}{\Delta_N}
 \left(1+\frac1{w_N}\right)
 \right].
 \label{eq:joint-critical-value}
\end{align}
\end{theorem}

\begin{proof}[Proof outline]
We first use \(z=\delta/\delta_N\) to identify the critical scale and then
\(u=\delta/\Delta_N\) to center the profile at its minimizer.  The stationary
and transient expansions reduce the objective to a one-variable profile.  Its
balance equation gives the Lambert expression, and
Lemma~\ref{lem:joint-global-localization} transfers the local minimum to every
global minimizer.  Appendix~\ref{app:finite-window-proofs} contains the uniform
error analysis.
\end{proof}

\begin{corollary}
\label{cor:joint-critical-canonical}
If
\[
 \varepsilon_N\log N-\lambda^\star(d)
 =o\left(\varepsilon_N\log\frac1{\varepsilon_N}\right),
\]
then
\begin{equation}
 1-\gamma_N
 \sim
 \left(\frac{H_{d,1}}{c_{d,1}q_{d,1}\lambda^\star(d)}\right)^{1/2}
 \left(\frac{\log N}{N\log\log N}\right)^{1/2}.
 \label{eq:joint-canonical-location}
\end{equation}
This is the elementary leading scale.  Under the stronger centering condition
below, retaining \(W_0\) resolves the slowly varying Lambert factor.
If the stronger condition
\(\varepsilon_N\log N-\lambda^\star(d)=o(\varepsilon_N)\) holds, then
\begin{equation}
 1-\gamma_N
 \sim\left[
 \frac{H_{d,1}\log N}
 {c_{d,1}q_{d,1}\lambda^\star(d)N
 W_0\left(
 \dfrac{H_{d,1}e^{2/(3d+7)}\log N}
 {c_{d,1}q_{d,1}\lambda^\star(d)}
 \right)}
 \right]^{1/2}.
 \label{eq:joint-canonical-W-location}
\end{equation}
\end{corollary}

Above \(\lambda^\star(d)\), the stationary root determines the leading
optimizer scale, while the transient term contributes to the refinements.
Assume
\begin{equation}
 \varepsilon_N\log N\longrightarrow\lambda>\lambda^\star(d).
 \label{eq:joint-supercritical-assumption}
\end{equation}

The two quadratic coefficients in the refined formulas enter only through
\begin{equation}
 \chi_{d,\alpha}
 =\frac{\alpha s_d}{4}q_{d,\alpha}+\nu_{d,\alpha},
 \qquad
 \chi_{d,1}=\frac{d(d+7)}{32(d+2)^2}.
 \label{eq:joint-chi-star}
\end{equation}

\begin{theorem}
\phantomsection
\label{thm:joint-supercritical}
\noindent\textup{\textbf{Leading scale.}}
Every global minimizer satisfies
\begin{equation}
 \frac{1-\gamma_N}{\delta_N}\longrightarrow1
 \label{eq:joint-supercritical-location}
\end{equation}
and
\begin{equation}
 \min_\gamma F_N(\gamma)
 =Nc_N
 -N\varepsilon_Nc_Nq_N\delta_N(1+o(1)).
 \label{eq:joint-supercritical-leading-value}
\end{equation}
\noindent\textup{\textbf{Second-order value.}}
The minimum has the refined expansion
\begin{align}
 \min_\gamma F_N(\gamma)
 ={}&Nc_N
 -N\varepsilon_Nc_Nq_N\delta_N
 +N\chi_{d,\alpha_N}\delta_N^2\notag\\
 &+\frac{H_N}{2\delta_N}
 +o\left(N\delta_N^2+\delta_N^{-1}\right).
 \label{eq:joint-supercritical-second-order}
\end{align}
\smallskip

\noindent\textup{\textbf{Refined location.}}
The location has the refinement
\begin{align}
 \frac{1-\gamma_N}{\delta_N}
 ={}&1
 -\frac{2\chi_{d,\alpha_N}
 +\varepsilon_N\frac{\alpha_N s_d}{4}q_N}
 {\alpha_Nc_Nq_N}
 \frac{\delta_N}{\varepsilon_N}\notag\\
 &+\frac{H_N}
 {2\alpha_Nc_Nq_N}
 \frac1{N\varepsilon_N\delta_N^2}
 +o\left(
 \frac{\delta_N}{\varepsilon_N}+\frac1{N\varepsilon_N\delta_N^2}
 \right).
 \label{eq:joint-supercritical-location-refined}
\end{align}
\end{theorem}

\begin{remark}
The leading term
\(-N\varepsilon_Nc_Nq_N\delta_N\) in
Eq.\eqref{eq:joint-supercritical-second-order} is kept exact.  A finite
expansion in powers of \(1/\log N\) can have an omitted error larger than
both displayed corrections.
\end{remark}

The following corollary identifies where the two explicit corrections in
Theorem~\ref{thm:joint-supercritical} exchange order.

\begin{corollary}
\label{cor:joint-supercritical-crossover}
If
\((\varepsilon_N\log N-\lambda)\log N\to0\), then
\begin{align}
 \delta_N
 &\sim e^{-1/(3d+7)}N^{-\lambda^\star(d)/(2\lambda)},
 \label{eq:joint-supercritical-delta-centered}\\
 \frac{\delta_N}{\varepsilon_N}
 &\sim\frac{e^{-1/(3d+7)}}{\lambda}(\log N)
 N^{-\lambda^\star(d)/(2\lambda)},
 \label{eq:joint-supercritical-stationary-scale}\\
 \frac1{N\varepsilon_N\delta_N^2}
 &\sim\frac{e^{2/(3d+7)}}{\lambda}(\log N)
 N^{-(1-\lambda^\star(d)/\lambda)}.
 \label{eq:joint-supercritical-transient-scale}
\end{align}
The transient correction dominates when
\(\lambda^\star(d)<\lambda<3\lambda^\star(d)/2\), and the
stationary correction dominates when \(\lambda>3\lambda^\star(d)/2\).
The same comparison applies to the two explicit corrections in
Eq.\eqref{eq:joint-supercritical-second-order}.
At the crossover value \(\lambda=3\lambda^\star(d)/2\),
\begin{equation}
 \begin{aligned}
 \frac{1-\gamma_N}{\delta_N}
 ={}&1+\frac{2}{3\lambda^\star(d)c_{d,1}q_{d,1}}\\
 &\quad\times\left(
 -2\chi_{d,1}e^{-1/(3d+7)}
 +\frac{\chi_{d,1}}2e^{2/(3d+7)}
 \right)(\log N)N^{-1/3}\\
 &\quad+o\bigl((\log N)N^{-1/3}\bigr),
 \end{aligned}
 \label{eq:joint-supercritical-crossover-location}
\end{equation}
\end{corollary}

\begin{remark}
\label{rem:finite-correction-balance}
For the canonical path
\(\varepsilon_N=r\lambda^\star(d)/\log N\), let \(r_N^{\rm bal}\) denote
the solution near \(3/2\) of \(N\delta_N^3=1\), which is the exact equality
of the two displayed correction scales.  Then
\begin{equation}
 r_N^{\rm bal}
 =\frac32+\frac{9}{2(3d+7)\log N}
 +O\bigl((\log N)^{-2}\bigr).
 \label{eq:finite-correction-balance-location}
\end{equation}
Indeed, Eq.\eqref{eq:joint-D-q} and the definition of \(\delta_N\) give
\[
 0=r_N^{\rm bal}\lambda^\star(d)
 +3\log q_{d,1+r_N^{\rm bal}\lambda^\star(d)/\log N},
\]
and expansion at \(r=3/2\) gives
Eq.\eqref{eq:finite-correction-balance-location}.
\end{remark}

\begin{proof}[Proof outline for Theorem~\ref{thm:joint-supercritical} and
Corollary~\ref{cor:joint-supercritical-crossover}]
The stationary minimum supplies the leading scale.  Balancing its local
curvature against the derivative of the transient correction yields the two
correction scales and the value expansion.  The complete comparison, including
uniform localization, is in Appendix~\ref{app:finite-window-proofs}.
\end{proof}

Theorems~\ref{thm:joint-subcritical}, \ref{thm:joint-critical-crossover}, and
\ref{thm:joint-supercritical} determine the leading optimizer scale in the
three regimes.  Corollary~\ref{cor:joint-supercritical-crossover} then compares
the two first corrections in the supercritical regime.  Figure~\ref{fig:joint-window-guide}
collects these results in their logical order.

\begin{figure}[H]
\centering
\includegraphics[width=\textwidth]{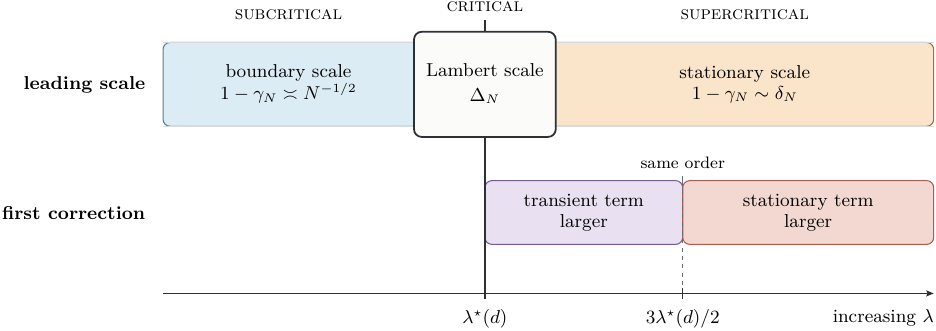}
\caption{Two nested transitions in the joint finite-size window.  The upper
band identifies the leading scale of \(1-\gamma_N\): the boundary scale below
\(\lambda^\star(d)\), the Lambert-corrected scale \(\Delta_N\) at equality, and
the stationary scale \(\delta_N\) above it.  The lower band resolves the
supercritical regime: the transient correction is larger below
\(3\lambda^\star(d)/2\), and the stationary correction is larger above it.}
\label{fig:joint-window-guide}
\end{figure}

The differentiated transient estimate also determines the local shape of the
complete finite objective on both moving scales.  Recall that
\(\Delta_N=\delta_N e^{w_N/2}\).  Define the normalized objective profile and
its stationary and transient comparison profiles by
\begin{align}
 \mathcal P_N(z)
 &=\frac{\Delta_N[F_N(1-\Delta_Nz)-F_N(1-\Delta_N)]}
 {H_N(1+w_N^{-1})},
 \label{eq:moving-normalized-profile}\\
 A_N(z)&=\frac{w_N}{2(w_N+1)}(z+z^{-1}-2),
 \qquad
 B_N(z)=\frac{z\log z-z+1}{w_N+1}.
 \label{eq:moving-comparison-profiles}
\end{align}

\begin{theorem}
\label{thm:moving-differentiated-profile}
Assume either Eq.\eqref{eq:joint-critical-assumption} or
Eq.\eqref{eq:joint-supercritical-assumption}.  Then \(\Delta_N\to0\),
\(N\Delta_N\to\infty\), and, for every \(0<\ell<u<\infty\),
\begin{equation}
 \left\|\mathcal P_N-A_N-B_N\right\|_{C^2([\ell,u])}
 \longrightarrow0.
 \label{eq:moving-C2-convergence}
\end{equation}
\end{theorem}

\begin{proof}[Proof outline]
Eq.\eqref{eq:moving-C2-convergence} gives uniform convergence of the normalized
objective, its slope, and its curvature on every fixed multiplicative
neighborhood of \(\Delta_N\).  The explicit comparison profile \(A_N+B_N\)
is strictly convex and has its unique minimizer at \(z=1\).
Appendix~\ref{app:finite-window-proofs} gives the differentiated convergence
proof.
\end{proof}

\begin{corollary}
\label{cor:moving-global-uniqueness}
Under either set of assumptions in
Theorem~\ref{thm:moving-differentiated-profile}, the finite objective has a
unique global minimizer for all sufficiently large \(N\), and
\[
 1-\gamma_N\sim\Delta_N.
\]
In every fixed supercritical regime, \(\Delta_N\sim\delta_N\).
\end{corollary}

\begin{proof}
The first derivative of the explicit comparison profile is
\[
 A_N'(z)+B_N'(z)
 =\frac{w_N}{2(w_N+1)}(1-z^{-2})
 +\frac{\log z}{w_N+1}
\]
and its second derivative is
\begin{equation}
 A_N''(z)+B_N''(z)
 =\frac{w_Nz^{-3}+z^{-1}}{w_N+1}>0.
 \label{eq:moving-profile-curvature}
\end{equation}
For fixed \(0<\ell<1<u\), the normalized curvature in
Eq.\eqref{eq:moving-profile-curvature} is bounded below uniformly, and the
normalized derivative has opposite signs at \(\ell\) and \(u\).  The
\(C^2\) convergence proves strict convexity and the existence of exactly one
critical point in \([\ell\Delta_N,u\Delta_N]\).  Its ratio to \(\Delta_N\) tends to one.
Theorems~\ref{thm:joint-critical-crossover} and
\ref{thm:joint-supercritical} place every global minimizer in such a fixed
multiplicative interval.  All global minimizers therefore coincide with that
critical point.
\end{proof}

\section{High-dimensional optimization}
\label{sec:high-dimensional}

This section studies how the stationary geometric averages simplify as the
dimension grows.  Uniform and differentiated expansions yield the limiting
objective, its optimizer, and the first correction of order \(d^{-1}\).
The limiting tradeoff curve is included in
Figure~\ref{fig:uniform-adversarial-tradeoff}, and
Figure~\ref{fig:high-dimensional-first-order} later checks the first correction
and its differentiated remainder.

\subsection{Limiting objective and optimizer}

Let
\begin{equation}
 \sigma^2(\gamma)=\frac{2}{1+\gamma}.
 \label{eq:high-dimensional-parameters}
\end{equation}

\begin{theorem}
For every fixed \(\alpha>0\), uniformly in \(0\leq\gamma\leq1\),
\begin{equation}
 M_{d,\alpha}(\gamma)
 =\sigma(\gamma)^\alpha+O_\alpha(d^{-1}).
 \label{eq:high-dimensional-uniform-limit}
\end{equation}
Consequently, \(\Phi_{d,\alpha}\) converges uniformly to
\begin{equation}
 \Phi_{\infty,\alpha}(\gamma)
 =h_\alpha(\gamma)
 \left(\frac{2}{1+\gamma}\right)^{\alpha/2}.
 \label{eq:high-dimensional-objective}
\end{equation}
The limiting objective has a unique minimizer.  It equals \(1\) when
\(0<\alpha\leq1\).  When \(\alpha>1\), denote it by
\(\gamma_{\infty,\alpha}\).  It belongs to \((1/2,1)\), and the ratio
\(t=\gamma_{\infty,\alpha}/(1-\gamma_{\infty,\alpha})\) is characterized by
\begin{equation}
 t^{\alpha-1}(2+3t)=3+4t.
 \label{eq:high-dimensional-optimizer}
\end{equation}
Every selection
\[
 \gamma_{d,\alpha}
 \in\operatorname*{argmin}_{0\leq\gamma\leq1}
 \Phi_{d,\alpha}(\gamma)
\]
satisfies \(\gamma_{d,\alpha}\to\gamma_{\infty,\alpha}\).
\end{theorem}

\begin{proof}
Let \(r=\norm{p-x_\gamma}^2\).
The second and fourth moments of a uniform ball point, together with the
coefficient power sums of the stationary state, give
\[
 \E r=\frac d{d+2}\sigma^2(\gamma)
\]
and
\[
 \E\bigl(r-\sigma^2(\gamma)\bigr)^2
 =\frac{2\sigma^4(\gamma)}{d+2}
 -\frac{2d}{(d+2)(d+4)}
 \left(1+\frac{(1-\gamma)^3}{1+\gamma+\gamma^2+\gamma^3}\right).
\]
In particular, the second centered moment is at most \(6/d\), uniformly in
\(\gamma\).  Since \(\sigma^2(\gamma)\in[1,2]\) and
\(r\in[0,4]\), a Taylor bound for
\(u\mapsto u^{\alpha/2}\) around \(\sigma^2(\gamma)\) gives
Eq.\eqref{eq:high-dimensional-uniform-limit}.

It remains to identify the limiting minimizer.  Let
\(t=\gamma/(1-\gamma)\).  Apart from its positive constant factor,
Eq.\eqref{eq:high-dimensional-objective} becomes
\[
 \frac{1+t^\alpha}{[(1+t)(1+2t)]^{\alpha/2}}.
\]
Its logarithmic derivative has the sign of
\[
 t^{\alpha-1}(2+3t)-(3+4t).
\]
For \(0<\alpha\leq1\), both factors in
Eq.\eqref{eq:high-dimensional-objective} are minimized at \(\gamma=1\).
For \(\alpha>1\), the last expression is negative for \(t\leq1\).
For \(t>1\), the function \(t^{\alpha-1}\) increases and
\((3+4t)/(2+3t)\) decreases, so there is exactly one zero.  This proves
Eq.\eqref{eq:high-dimensional-optimizer} and uniqueness.  Uniform objective
convergence and separation of the unique minimum imply convergence of the
complete argmin set.
\end{proof}

\subsection{First-order correction}

For the first-order refinement, define
\begin{equation}
\mathcal C_\alpha(\gamma)
 =-\alpha\sigma(\gamma)^\alpha
 +\frac{\alpha(\alpha-2)}4\sigma(\gamma)^{\alpha-4}
 \left(\sigma(\gamma)^4-1-
 \frac{(1-\gamma)^3}{1+\gamma+\gamma^2+\gamma^3}\right).
\label{eq:high-dimensional-coefficient}
\end{equation}

The coefficient \(\mathcal C_\alpha(\gamma)\) is the first finite-dimensional
correction to the limiting stationary insertion cost.

\begin{theorem}
\label{thm:high-dimensional-first-order}
For fixed \(\alpha>0\), uniformly in \(\gamma\in[0,1]\),
\begin{equation}
 M_{d,\alpha}(\gamma)
 =\sigma(\gamma)^\alpha
 +\frac{\mathcal C_\alpha(\gamma)}d
 +O_\alpha(d^{-2}).
 \label{eq:high-dimensional-first-order}
\end{equation}
For fixed \(\alpha>1\) and every compact interval \(I\) contained in
\((0,1)\),
\begin{equation}
 \left\|
 M_{d,\alpha}-\sigma^\alpha-\frac{\mathcal C_\alpha}{d}
 \right\|_{C^2(I)}
 =O_{\alpha,I}(d^{-2}),
 \label{eq:high-dimensional-differentiated}
\end{equation}
where
\[
 \|f\|_{C^2(I)}
 =\max_{0\leq k\leq2}\sup_{\gamma\in I}|f^{(k)}(\gamma)|.
\]
\end{theorem}

\begin{proof}[Proof outline]
The proof has three steps: exact moment identities identify the order
\(d^{-1}\) term, a uniform estimate bounds the remainder by \(O(d^{-2})\), and
interior derivative estimates give the \(C^2\) conclusion.
Appendix~\ref{app:high-dimensional-proof} contains the complete calculation.
\end{proof}

\begin{corollary}
For every fixed \(\alpha>1\), the stationary objective has a unique global
minimizer \(\gamma_{d,\alpha}\) for all sufficiently large \(d\).  Moreover,
\begin{equation}
 \gamma_{d,\alpha}
 =\gamma_{\infty,\alpha}
 -\frac{(h_\alpha\mathcal C_\alpha)'(\gamma_{\infty,\alpha})}
 {d\Phi_{\infty,\alpha}''(\gamma_{\infty,\alpha})}
 +O_\alpha(d^{-2}),
 \label{eq:high-dimensional-location-rate}
\end{equation}
and
\begin{equation}
 \min_\gamma\Phi_{d,\alpha}(\gamma)
 =\Phi_{\infty,\alpha}(\gamma_{\infty,\alpha})
 +\frac{h_\alpha(\gamma_{\infty,\alpha})
 \mathcal C_\alpha(\gamma_{\infty,\alpha})}{d}
 +O_\alpha(d^{-2}).
 \label{eq:high-dimensional-minimum}
\end{equation}
\end{corollary}

\begin{proof}
Uniform convergence places every global minimizer in a compact neighborhood
of \(\gamma_{\infty,\alpha}\).  At the root \(t>1\) of
Eq.\eqref{eq:high-dimensional-optimizer}, the derivative of the difference
between its two sides is
\[
 (\alpha-1)\frac{3+4t}{t}
 +\frac1{2+3t}>0.
\]
Thus \(\Phi_{\infty,\alpha}''(\gamma_{\infty,\alpha})>0\).
Eq.\eqref{eq:high-dimensional-differentiated} gives positive curvature on
that neighborhood for all sufficiently large \(d\), which proves
uniqueness.  Expanding the first-order condition yields
Eq.\eqref{eq:high-dimensional-location-rate}.  Substitution into the value
expansion gives Eq.\eqref{eq:high-dimensional-minimum}.
\end{proof}

\begin{proposition}
The limiting optimizer satisfies
\[
 (1-\gamma_{\infty,\alpha})^{\alpha-1}
 \longrightarrow\frac34
 \qquad\text{as }\alpha\to1^+,
\]
and
\[
 \gamma_{\infty,\alpha}
 =\frac12+\frac{\log(7/5)}{4(\alpha-1)}
 +O((\alpha-1)^{-2})
 \qquad\text{as }\alpha\to\infty.
\]
\end{proposition}

\begin{proof}
Let \(\varepsilon=\alpha-1\).  Eq.\eqref{eq:high-dimensional-optimizer}
is equivalent to
\[
 t^\varepsilon=\frac{3+4t}{2+3t}.
\]
As \(\varepsilon\to0^+\), this equation forces \(t\to\infty\)
and hence \(t^\varepsilon\to4/3\).  Since
\(1-\gamma_{\infty,\alpha}=1/(1+t)\), the first limit follows.
As \(\varepsilon\to\infty\), put \(t=1+u\).  Taking logarithms and
expanding at \(u=0\) gives
\[
 \varepsilon u=\log(7/5)+O(\varepsilon^{-1}).
\]
Substitution into \(\gamma=t/(1+t)\) gives the stated expansion.
\end{proof}

\begin{remark}
All high-dimensional estimates keep \(\alpha\) fixed.  The differentiated
estimate is uniform on compact intervals contained in \((0,1)\).
\end{remark}

\section{Numerical evaluation and parameter selection}
\label{sec:numerical-scaling}

This section checks the asymptotic results of Sections~\ref{sec:finite-regimes}
and~\ref{sec:high-dimensional}.  The first subsection follows the transition
through power one and verifies the local shape of the finite objective.  The
second tests the first high-dimensional correction and its differentiated
remainder.

\subsection{Finite-size transitions near power one}

Only three diagnostic quantities are needed: the scaled parameter \(r\), the
size-pair exponent \(\widehat\beta_N\), and the regime-specific location scale
\(S_N\).

For \(d=2\), put \(r=\lambda/\lambda^\star(2)\).  We minimized the finite
objective along
\[
 \alpha_N=1+\frac{r\lambda^\star(2)}{\log N}
\]
at 28 values of \(r\) in \([0,2]\), including a grid of spacing \(0.025\) on
\([0.9,1.1]\), and at \(N=2^{10},2^{12},\ldots,2^{32}\).  The size-pair
exponent is
\[
 \widehat\beta_N(r)
 =-\frac{\log(1-\gamma_{4N})-\log(1-\gamma_N)}{\log4}.
\]
Here both minimizers lie on the joint path above, at their respective values of
\(N\).  The location scale is
\begin{equation}
 S_N(r)=
 \begin{cases}
  \displaystyle
  \sqrt{\frac{H_{2,1}}{2a_2(r\lambda^\star(2))}}\,N^{-1/2},&0\leq r<1,\\
  \Delta_N,&r=1,\\
  \delta_N,&1<r\leq2.
 \end{cases}
 \label{eq:finite-diagnostic-reference-scale}
\end{equation}

Figure~\ref{fig:finite-size-diagnostics} places the leading-scale transition
and the correction crossover in the same coordinates.  The regime-specific
normalization exposes the nonuniform convergence near \(r=1\), while the
signed comparison in Figure~\ref{fig:finite-size-diagnostics}(c) shows which
explicit correction controls the refined supercritical location on either
side of \(r=3/2\).
The value \(r=3/2\) is an asymptotic crossover.  Higher-order terms displace
the finite-size balance points and the intersections of the effective-exponent
curves slightly to its right.

The objective was evaluated at three nested radial, angular, and state
resolutions.  All \(336\) parameter pairs passed the promotion tests.  The
largest relative optimizer change between the two finest resolutions was
\(9.20\times10^{-4}\), and the corresponding maximum relative objective
change was \(1.74\times10^{-8}\).  Recomputing the displayed effective
exponents with the two finest resolutions changed them by at most
\(4.85\times10^{-4}\).  For \(N=2^{10}\) the complete finite sum was used.  At
larger \(N\), a stationary decomposition and its Poisson equation evaluate the
transient sum, with an exponentially small terminal term on the displayed
scales.  At \(N=4096\), direct summation and the accelerated evaluation were
compared at \(r=0\) and \(1\).  Both gave the same minimizer at the stored
precision, and the largest relative difference between their objective values
was \(3.78\times10^{-14}\).

\begin{figure}[!t]
 \centering
 \includegraphics[width=\textwidth]{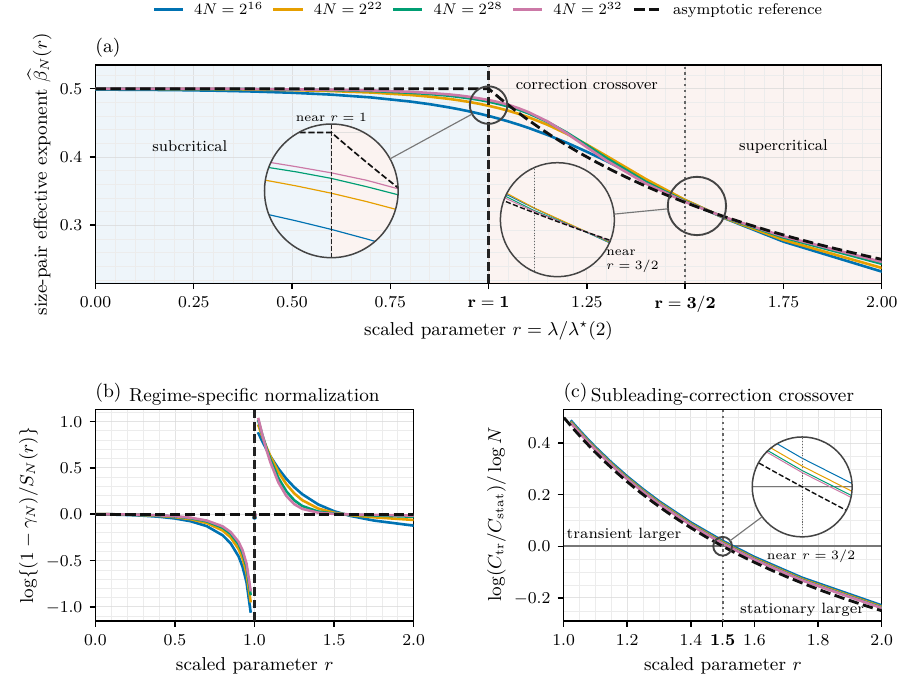}
 \caption{Finite-size approach to the scaling transitions in \(d=2\).
 (a) compares the size-pair exponent with its limit, \(1/2\) for
 \(r\leq1\) and \(1/(2r)\) for \(r>1\).  The first circular lens enlarges the
 marked region near \(r=1\).  The dotted line at \(r=3/2\) marks the
 asymptotic correction crossover, and the second lens shows the finite-size
 curve crossings slightly to its right.  Because they compare pairs of finite
 sizes, their crossings need not occur exactly at \(r=3/2\); see
 Remark~\ref{rem:finite-correction-balance}.  (b) uses the reference scale
 \(S_N(r)\) in
 Eq.\eqref{eq:finite-diagnostic-reference-scale}; convergence to the zero line
 means \((1-\gamma_N)/S_N(r)\to1\).  Its three disconnected pieces correspond
 to the three normalizations, with isolated markers at \(r=1\).  (c)
 compares the transient and stationary corrections through
 \(\log(C_{\rm tr}/C_{\rm stat})/\log N\), where
 \(C_{\rm tr}=1/(N\varepsilon_N\delta_N^2)\) and
 \(C_{\rm stat}=\delta_N/\varepsilon_N\).  The dashed limiting curve vanishes
 at \(r=3/2\); the zeros of the colored finite-size curves approach this value
 and are enlarged in the circular lens.}
 \label{fig:finite-size-diagnostics}
\end{figure}

The location asymptotics are complemented by the local shape of the finite
objective.  Theorem~\ref{thm:moving-differentiated-profile} states that
\(\mathcal P_N-A_N-B_N\) converges to zero with its first two derivatives on
every fixed interval containing \(z=1\).

Figure~\ref{fig:local-objective-geometry} compares the finite objective and its
first two derivatives with the stationary and transient profiles in this
result.

For Figure~\ref{fig:local-objective-geometry}, each profile was evaluated on
129 Chebyshev--Lobatto points and audited independently on 65 points at a
second spatial resolution.  A degree-18 Chebyshev projection was used only to
differentiate the computed profile.  Across the nine displayed parameter-size
pairs, the maximum differences between the two resolutions were
\(4.01\times10^{-4}\), \(1.01\times10^{-3}\), and \(4.50\times10^{-3}\) for
the profile, slope, and curvature, respectively.

\begin{figure}[H]
 \centering
 \includegraphics[width=\textwidth]{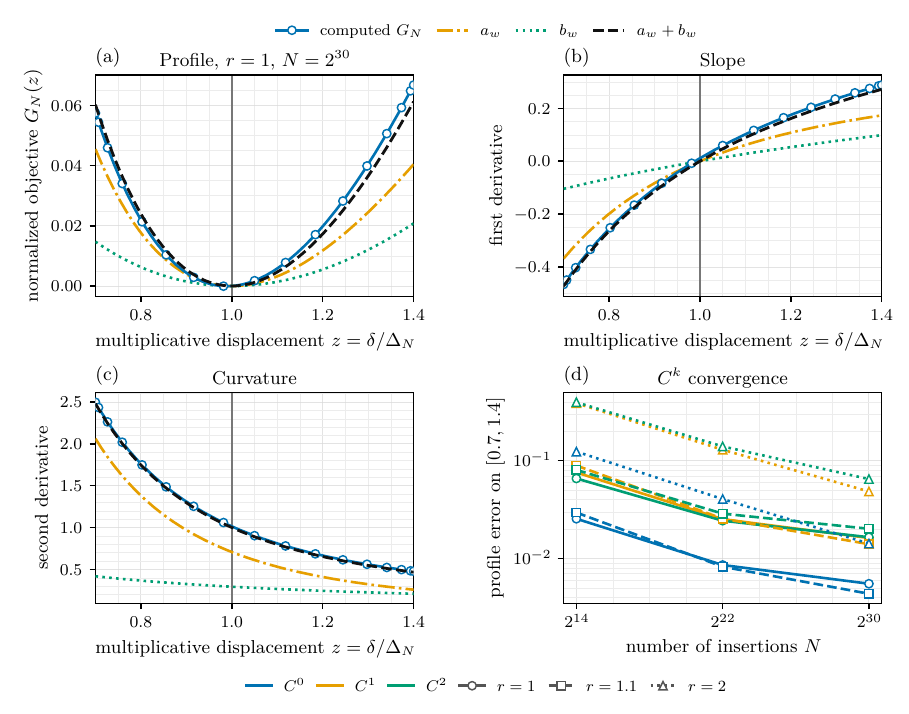}
 \caption{Numerical comparison of the differentiated moving-window profile in
 \(d=2\).  (a)--(c) compare \(\mathcal P_N\), its slope, and its
 curvature with the stationary and transient profiles \(A_N\), \(B_N\), and
 their sum at \(r=1\) and \(N=2^{30}\).  (d) gives the corresponding
 \(C^0\), \(C^1\), and \(C^2\) errors on \(0.7\leq z\leq1.4\) for
 \(r=1,1.1,2\).  Their decrease with \(N\) tests the objective, slope, and
 curvature conclusions of
 Theorem~\ref{thm:moving-differentiated-profile}.}
 \label{fig:local-objective-geometry}
\end{figure}

The constants in the finite-size results can be evaluated directly from the
geometric integral in Eq.\eqref{eq:H-geometric-integral}.  Their numerical
magnitudes are collected in Table~\ref{tab:finite-regime-constants}.  Part~(a)
also records the two thresholds that separate the joint finite regimes and the
coefficient in the critical location.  Part~(b) gives the location and cost
coefficients in Theorem~\ref{thm:sublinear-boundary}.

\begin{table}[H]
\centering
\small
\caption{Numerical constants needed to use the finite-size laws. In part~(a), $\lambda^\star(d)$ and $3\lambda^\star(d)/2$ locate the leading transition and the correction crossover, while the last column gives the critical optimizer scale. In part~(b), the last two columns quantify the optimizer location and the cost correction for $0<\alpha<1$. The displayed digits come from the highest quadrature order; the reproducibility files retain the resolution differences.}
\label{tab:finite-regime-constants}
\textbf{(a) Dimension-dependent constants at $\alpha=1$.}\par\smallskip
\begin{tabular}{rcccc}
\toprule
$d$ & $H_{d,1}$ & $\lambda^\star(d)$ & $3\lambda^\star(d)/2$ & critical location coefficient \\
\midrule
1 & 0.166667 & 0.364643 & 0.546965 & 1.047360 \\
2 & 0.244533 & 0.415279 & 0.622918 & 1.042637 \\
3 & 0.289286 & 0.446287 & 0.669431 & 1.039395 \\
5 & 0.338504 & 0.482324 & 0.723486 & 1.035311 \\
10 & 0.387483 & 0.520566 & 0.780849 & 1.030634 \\
20 & 0.417395 & 0.545288 & 0.817933 & 1.027446 \\
50 & 0.437521 & 0.562585 & 0.843877 & 1.025150 \\
100 & 0.444638 & 0.568839 & 0.853258 & 1.024307 \\
\bottomrule
\end{tabular}
\end{table}

\begin{table}[H]
\centering
\small
\ContinuedFloat
\textbf{(b) Constants for $0<\alpha<1$.}\par\smallskip
\begin{tabular}{rrccc}
\toprule
$d$ & $\alpha$ & $H_{d,\alpha}$ & location coefficient & cost-correction coefficient \\
\midrule
1 & 0.25 & 0.043372 & 0.169090 & 0.641253 \\
1 & 0.50 & 0.085208 & 0.253735 & 0.503721 \\
1 & 0.75 & 0.126132 & 0.334535 & 0.439877 \\
\cmidrule(lr){1-5}
2 & 0.25 & 0.061863 & 0.206485 & 0.748996 \\
2 & 0.50 & 0.122867 & 0.286790 & 0.642633 \\
2 & 0.75 & 0.183585 & 0.361195 & 0.592982 \\
\cmidrule(lr){1-5}
5 & 0.25 & 0.082579 & 0.246189 & 0.838568 \\
5 & 0.50 & 0.166256 & 0.322192 & 0.774026 \\
5 & 0.75 & 0.251431 & 0.390340 & 0.751489 \\
\cmidrule(lr){1-5}
10 & 0.25 & 0.092666 & 0.264812 & 0.874824 \\
10 & 0.50 & 0.187916 & 0.338924 & 0.831674 \\
10 & 0.75 & 0.286075 & 0.404335 & 0.825440 \\
\bottomrule
\end{tabular}
\end{table}

\subsection{The high-dimensional expansion}

The final computation checks the two-term expansion in
Theorem~\ref{thm:high-dimensional-first-order}.  In
Figure~\ref{fig:high-dimensional-first-order}(a) and (b), the scaled
differences for \(\alpha=2\) and \(4\) approach the explicit coefficient
\(\mathcal C_\alpha(\gamma)\) as the dimension increases.  Part~(c) subtracts
this first correction and measures the \(C^2([0.5,0.9])\) norm of the
remainder.  Its alignment with the \(d^{-2}\) reference tests the
differentiated remainder estimate.  All quantities are evaluated from exact
moment formulas and therefore contain no sampling error.

\begin{figure}[H]
 \centering
 \includegraphics[width=\textwidth]{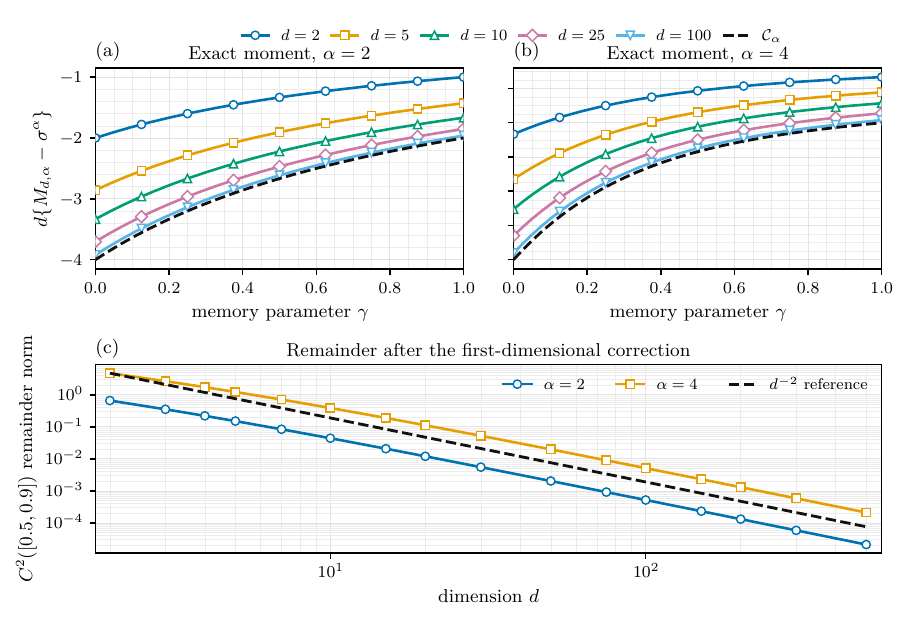}
 \caption{First-order high-dimensional expansion.  (a) and (b) show
 \(d[M_{d,\alpha}(\gamma)-\sigma(\gamma)^\alpha]\) for \(\alpha=2\) and
 \(4\), together with the explicit coefficient
 \(\mathcal C_\alpha(\gamma)\).  (c) shows the differentiated remainder after
 the two terms in Eq.\eqref{eq:high-dimensional-first-order}; the dashed
 reference has order \(d^{-2}\).}
 \label{fig:high-dimensional-first-order}
\end{figure}

\subsection{Implications for parameter selection}

The preceding results provide a parameter-selection map for the
fixed-parameter \(\gamma\)-strategy.  The finite phase diagram identifies
when a long memory compensates for its initialization cost.  The adversarial
theorem gives an exact robust choice through power three, and the high-power
theorem determines the scale required when the longest realized edges
dominate.  The numerical evaluations make the transition and correction
scales directly available for the dimensions reported in
Table~\ref{tab:finite-regime-constants}.

These conclusions apply directly to the edge-power cost of the labelled tree.
Models of a complete system can use this solution as a geometric component
and add activation, capacity, hop, delay, or state-movement charges.  The
recent online and dynamic range-assignment literature
\cite{deBergMarkovicUmboh2023,deBergSadhukhanSpieksma2024} and the
bounded-recourse Steiner literature \cite{GuGuptaKumar2016} supply natural
benchmarks for such extensions.

The numerical data supporting the figures and tables are available at
\url{https://github.com/tashimir/optimal-exponential-memory-edge-power-data}.

\section{Conclusions}

The analysis gives a complete phase diagram for the fixed-parameter
\(\gamma\)-strategy across the uniform, worst-case, finite-size, and
high-dimensional regimes considered here.

Eq.\eqref{eq:reduction} separates attachment quality, measured by the
insertion cost, from the two-edge factor generated by subdivision at each
labelled auxiliary vertex.  This factor produces a stationary transition at
power one and creates
interior minimizers for every power greater than one.  Exact low powers show
that the best stationary and adversarial parameters need not agree.  At high
powers, periodic blocks provide stronger lower-bound certificates than
alternation in part of the parameter interval.  These block lower bounds,
combined with the separation argument at \(\gamma=1/2\), determine the optimized
high-power scale and its leading constant.  The finite-size constant for
\(0<\alpha<1\) is expressed by the positive geometric integral in
Eq.\eqref{eq:H-geometric-integral}.  The joint window near power one has an
explicit subcritical profile and threshold.  At the threshold, the optimizer
is governed by a Lambert profile and has scale
\(\sqrt{\log N/(N\log\log N)}\).  Above the threshold, it approaches the
exact stationary root \(q_{d,\alpha_N}^{1/(\alpha_N-1)}\), and its minimum is
known through the stationary quadratic and transient corrections.  The
differentiated moving profile also proves eventual uniqueness of the finite
optimizer in the critical and supercritical windows.  In high dimension,
differentiated convergence gives eventual uniqueness and the first displacement
of the stationary optimizer.

Several directions remain natural for further study.  They include broader
classes of geometric inputs, state-dependent or time-varying rules, and
models that assign explicit costs to retained state, movement, capacity, or
delay.  On the analytic side, the adversarial problem beyond the cubic range,
the global structure of the stationary minimizer, and limits in which the
dimension, the power, and the number of processed points vary together deserve
further investigation.  These extensions may clarify how widely the
constant-gain rule preserves its balance between distribution-sensitive
performance and controlled behavior under arbitrary input sequences.

\appendix
\section{Foundational insertion-cost inputs}
\label{app:foundational-insertion-inputs}

This section proves the three insertion-cost facts on which the optimization
arguments rely: ordering of the stationary moment, the radial upper envelope,
and the exact adversarial value through power three.

\begin{proof}[Proof of the order and endpoint assertions in
Lemma~\ref{lem:stationary-moment-order}]
A symmetric random vector \(y\) is said to be more peaked than a symmetric
random vector \(z\) when
\[
 \Pr(y\in K)\geq\Pr(z\in K)
\]
for every compact convex set \(K\) symmetric about the origin.  A theorem of
Olkin and Tong~\cite{OlkinTong1988} states that if independent random vectors
have a common symmetric log-concave density and one nonnegative coefficient
vector majorizes another with the same sum, then the sum with the more balanced
coefficient vector is more peaked.

For \(n\geq0\), consider the normalized coefficients
\[
 w_j^{(n)}(\gamma)
 =\frac{(1-\gamma)\gamma^j}{1-\gamma^{n+1}},
 \qquad 0\leq j\leq n.
\]
Their first \(m+1\) entries have sum
\[
 \frac{1-\gamma^{m+1}}{1-\gamma^{n+1}}
 =\frac{\sum_{j=0}^m\gamma^j}{\sum_{j=0}^n\gamma^j},
 \qquad 0\leq m<n.
\]
For \(0<\gamma<1\), its derivative has numerator
\[
 \sum_{i=0}^m\sum_{j=m+1}^n(i-j)\gamma^{i+j-1}<0.
\]
Consequently, if \(0\leq\gamma_1<\gamma_2<1\), then
\(w^{(n)}(\gamma_1)\) majorizes \(w^{(n)}(\gamma_2)\).  The uniform density on
\(\B\) is symmetric and log-concave in the extended-value sense, so the
normalized sum for \(\gamma_2\) is more peaked than the one for
\(\gamma_1\).

Let
\[
 y_{k,n}=\sum_{j=0}^n w_j^{(n)}(\gamma_k)p_j,
 \qquad
 \lambda_{k,n}=1-\gamma_k^{n+1},
 \qquad k\in\{1,2\}.
\]
Since \(\lambda_{2,n}<\lambda_{1,n}\), one has
\(K/\lambda_{1,n}\subset K/\lambda_{2,n}\) for every such set \(K\).
Peakedness and this inclusion give
\[
 \Pr(\lambda_{2,n}y_{2,n}\in K)
 \geq\Pr(\lambda_{1,n}y_{1,n}\in K).
\]
The truncated sums converge almost surely to the stationary states.  Their
laws, including the limiting laws, are absolutely continuous because each sum
contains a nonzero multiple of a uniform point.  Boundaries of the relevant
convex sets are therefore continuity sets, and passage to the limit shows that
\(x_{\gamma_2}\) is more peaked than \(x_{\gamma_1}\).

For \(t\geq0\), let
\[
 H_t(z)=\Pr(\norm{p-z}\leq t)
 =\frac{\operatorname{vol}(\B\cap B(z,t))}
 {\operatorname{vol}(\B)}.
\]
Slicing both balls by lines parallel to \(z\) shows that \(H_t(z)\) depends
only on \(\norm z\) and is nonincreasing in that radius.  Its superlevel sets
are centered balls.  Layer-cake integration and the peakedness comparison give
\[
 \Pr(\norm{p-x_{\gamma_2}}\leq t)
 \geq
 \Pr(\norm{p-x_{\gamma_1}}\leq t).
\]
Integration of the tail probabilities proves that
\(M_{d,\alpha}\) is nonincreasing for every \(\alpha>0\).

The same slicing argument is strict on a nonempty interval of \(t\) whenever
\(z\ne0\).  Integrating the strict tail comparison yields
\[
 \E\norm{p-z}^\alpha>\E\norm p^\alpha=c_{d,\alpha}.
\]
The stationary state is nonzero almost surely, so conditioning on it gives
\(M_{d,\alpha}(\gamma)>c_{d,\alpha}\) for every \(\gamma<1\).  Finally,
Eq.\eqref{eq:state-second-moment} implies \(x_\gamma\to0\) in \(L^2\) as
\(\gamma\to1^-\).  Bounded convergence gives the asserted endpoint limit.
\end{proof}

\begin{proof}[Proof of Lemma~\ref{lem:radial-upper-envelope}]
Fix a unit vector \(e\), and let \(u\) be uniform on the unit sphere in
\(\R^{d+2}\).  The divergence theorem on the unit ball, followed by integration
by parts in the first-coordinate density on the sphere, gives, for
\(0<r<1\),
\begin{equation}
 g_{d,\alpha}'(r)
 =\alpha r\,\E\norm{u-re}^{\alpha-2}.
 \label{eq:appendix-radial-derivative}
\end{equation}
For \(d=1\), the same identity follows directly from
\[
 g_{1,\alpha}'(r)
 =\frac{(1+r)^\alpha-(1-r)^\alpha}{2},
\]
because the first coordinate of a uniform point on the sphere in \(\R^3\) is
uniform on \([-1,1]\).

Let \(m(r)=\E\norm{u-re}^{\alpha-2}\).  The corresponding radial potential in
\(\R^{d+2}\) satisfies
\[
 \bigl(r^{d+1}m'(r)\bigr)'
 =r^{d+1}(\alpha-2)(d+\alpha-2)
 \E\norm{u-re}^{\alpha-4}.
\]
Since \(m'(0)=0\), the function \(m\) is increasing when \(\alpha>2\).
Integrating Eq.\eqref{eq:appendix-radial-derivative} and substituting \(s=rt\)
give
\[
 \frac{g_{d,\alpha}(r)-c_{d,\alpha}}{r^2}
 =\alpha\int_0^1t\,m(rt)\,dt.
\]
This quotient is increasing, and hence
\[
 g_{d,\alpha}(r)-c_{d,\alpha}
 \leq\bigl(g_{d,\alpha}(1)-c_{d,\alpha}\bigr)r^2.
\]

It remains to bound the endpoint.  The change of variables
\(q=(p-e)/2\) maps \(\B\) onto the ball
\(D=B(-e/2,1/2)\).  In polar coordinates about the origin, the radial extent
of \(D\) in direction \(\omega\) is
\(( -\omega\mathbin{\cdot}e)_+\).  Therefore
\[
 \E_D\norm q^\alpha
 =\frac{d}{d+\alpha}
 \frac{\int ( -\omega\mathbin{\cdot}e)_+^{d+\alpha}\,d\omega}
 {\int ( -\omega\mathbin{\cdot}e)_+^d\,d\omega}
 \leq\frac{d}{d+\alpha}=c_{d,\alpha}.
\]
Since \(g_{d,\alpha}(1)=2^\alpha\E_D\norm q^\alpha\), the endpoint estimate
follows.
\end{proof}

\begin{proof}[Proof of Theorem~\ref{thm:exact-adversarial-insertion}]
Put
\[
 a_\gamma=\frac2{1+\gamma}.
\]
For a unit vector \(u\), take the alternating input
\(p_i=(-1)^{i-1}u\).  On its attracting two-periodic orbit, the state before
the insertion of \(u\) is
\(-((1-\gamma)/(1+\gamma))u\), and every insertion length equals
\(a_\gamma\).  Lemma~\ref{lem:initial-state-independence} transfers the same
asymptotic value to every initial state.  Hence
\[
 \mathcal A_\alpha(\gamma)\geq a_\gamma^\alpha.
\]

For the matching upper bound through power two, let
\(y=\gamma x+(1-\gamma)p\).  Direct expansion and maximization over
\(x,p\in\B\) give
\begin{equation}
 \norm{p-x}^2
 +\frac2{1-\gamma^2}\bigl(\norm y^2-\norm x^2\bigr)
 \leq a_\gamma^2.
 \label{eq:appendix-quadratic-potential}
\end{equation}
Indeed, with \(r=\norm x\), \(s=\norm p\), and
\(c=x\mathbin{\cdot}p\), the left side is
\[
 -r^2+\frac{3-\gamma}{1+\gamma}s^2
 -\frac{2(1-\gamma)}{1+\gamma}c.
\]
It is maximized by \(c=-rs\), then by \(s=1\), and finally by
\(r=(1-\gamma)/(1+\gamma)\), where its value is \(a_\gamma^2\).
For \(0<\alpha\leq2\), concavity of \(z^{\alpha/2}\) gives
\[
 z^{\alpha/2}\leq a_\gamma^\alpha
 +\frac\alpha2a_\gamma^{\alpha-2}(z-a_\gamma^2).
\]
Applying this inequality to \(z=\norm{p-x}^2\), using
Eq.\eqref{eq:appendix-quadratic-potential}, and summing over the input sequence
leaves a bounded telescoping term.  Division by \(n\) proves
\(\mathcal A_\alpha(\gamma)\leq a_\gamma^\alpha\) for
\(0<\alpha\leq2\).

For power three, define
\[
 V_\gamma(x)
 =\frac{11+2\gamma-\gamma^2}
 {2(1-\gamma)(1+\gamma)^2}\norm x^2
 +\frac1{4(1-\gamma)}\norm x^4.
\]
We claim that
\begin{equation}
 \norm{p-x}^3+V_\gamma(y)-V_\gamma(x)
 \leq a_\gamma^3.
 \label{eq:appendix-cubic-potential}
\end{equation}
For fixed \(x\), the left side before subtracting \(V_\gamma(x)\) is convex
in \(p\), so its maximum over \(\B\) is attained on the unit sphere.  With
\(r=\norm x\) and \(t=(x/r)\mathbin{\cdot}p\), both the distance term and the
quartic potential term are convex functions of \(t\).  The maximum is therefore
attained at \(t=-1\) or \(t=1\).  Identifying the corresponding diameter with
\([-1,1]\), both cases reduce to \(p=1\) and \(x\in[-1,1]\).  Expansion gives
\begin{align}
 &a_\gamma^3+V_\gamma(x)
 -V_\gamma(\gamma x+1-\gamma)-(1-x)^3\notag\\
 &\quad=
 \frac{\bigl((1+\gamma)x+1-\gamma\bigr)^2}
 {4(1+\gamma)^3}P_\gamma(x),
 \label{eq:appendix-cubic-factorization}
\end{align}
where
\begin{align*}
 P_\gamma(x)={}&
 (\gamma^4+2\gamma^3+2\gamma^2+2\gamma+1)x^2\\
 &+(-2\gamma^4-4\gamma^3+4\gamma+2)x\\
 &+\gamma^4+2\gamma^3-2\gamma^2-6\gamma+5.
\end{align*}
The leading coefficient of \(P_\gamma\) is positive and its discriminant is
\(-16(1-\gamma)^2(1+\gamma)^2<0\).  Thus
Eq.\eqref{eq:appendix-cubic-factorization} is nonnegative, proving
Eq.\eqref{eq:appendix-cubic-potential}.  Summation telescopes the bounded
potential and proves the upper bound at \(\alpha=3\).

Finally, let \(2<\alpha<3\), and choose \(\theta\in(0,1)\) with
\[
 \frac1\alpha=\frac\theta2+\frac{1-\theta}{3}.
\]
H\"older's inequality for the insertion lengths \(\ell_i\) gives
\[
 \left(\frac1n\sum_{i=1}^n\ell_i^\alpha\right)^{1/\alpha}
 \leq
 \left(\frac1n\sum_{i=1}^n\ell_i^2\right)^{\theta/2}
 \left(\frac1n\sum_{i=1}^n\ell_i^3\right)^{(1-\theta)/3}.
\]
The quadratic and cubic potential inequalities bound the two factors by
\(a_\gamma+o(1)\), uniformly over the input sequence.  Hence the left side has
limit superior at most \(a_\gamma\).  The alternating input supplies equality,
which completes the proof.
\end{proof}

\section{Bernstein certificate for the fourth-power objective}
\label{app:fourth-convexity-certificate}

In the Bernstein representation used in the proof of
Proposition~\ref{prop:fourth-power}, the coefficient vector is
\begin{align*}
(&17(d+3),\,3(53d+157)/11,\,2(359d+1071)/55,\\
&4(509d+1546)/165,\,31(43d+133)/110,
 (5701d+17705)/462,\\
&2(1512d+4559)/231,\,(809d+2227)/55,
 4(747d+1690)/165,\\
&4(353d+548)/55,\,48(10d+7)/11,\,8(11d-4)).
\end{align*}
All twelve coefficients are strictly positive for \(d\geq1\), which provides
the convexity certificate invoked in the main text.

\section{Proof for the high-power adversarial limit}
\label{app:high-power-proof}

The proof has two parts.  The separation of almost maximal insertion lengths
gives the upper bound at \(\gamma=1/2\).  For the matching lower bound, an
\(m\)-block input controls parameters near \(1/2\), while antipodal
alternation controls the remaining parameters.  The points \(\gamma_0\) and
\(\gamma_1\) below mark where these two lower certificates are exchanged.

\begin{proof}[Proof of Theorem~\ref{thm:high-power-adversarial}]
\smallskip
\noindent\emph{Upper bound.}
Take \(\varepsilon=2\log\alpha/\alpha\) in
Eq.\eqref{eq:high-power-upper}.  Then
\[
 (1-\varepsilon)^\alpha\leq\alpha^{-2}
\]
and
\[
 g_\varepsilon
 =\frac{\log\alpha}{\log2}+O(\log\log\alpha).
\]
It follows that
\begin{equation}
 \limsup_{\alpha\to\infty}
 (\log\alpha)v_\alpha
 \leq2\log2.
 \label{eq:high-power-limsup}
\end{equation}

\smallskip
\noindent\emph{Choice of the block length.}
For the uniform lower bound, choose
\begin{equation}
 m=\left\lceil\log_2(\alpha\log\alpha)\right\rceil.
 \label{eq:high-power-block-length}
\end{equation}
Thus \(\alpha2^{-m}\leq1/\log\alpha\) and
\begin{equation}
 m=\frac{\log\alpha+\log\log\alpha}{\log2}+O(1).
 \label{eq:high-power-block-asymptotic}
\end{equation}
We prove
\begin{equation}
 \inf_{0\leq\gamma<1}
 \max\left\{\frac{s_m(\gamma)}m,
 f_\alpha(\gamma)\right\}
 \geq\frac{2-o(1)}m.
 \label{eq:combined-periodic-lower}
\end{equation}

\smallskip
\noindent\emph{The region \(0\leq\gamma\leq\gamma_0\).}
\[
 \gamma_0=\frac{1+(\log2)/\alpha}{2}.
\]
For \(0\leq\gamma\leq1/2\), convexity gives
\[
 (2\gamma)^\alpha+[2(1-\gamma)]^\alpha\geq2,
\]
while
\[
 (1+\gamma^m)^\alpha
 \leq(1+2^{-m})^\alpha=1+o(1).
\]
The numerator bound also holds on \([1/2,\gamma_0]\).  On this interval,
\[
 \alpha\gamma^m
 \leq\alpha2^{-m}(1+(\log2)/\alpha)^m=o(1),
\]
and hence \(s_m(\gamma)\geq2-o(1)\).

\smallskip
\noindent\emph{The intermediate region \(\gamma_0\leq\gamma\leq\gamma_1\).}
It remains to consider \(\gamma\geq\gamma_0\).  Define
\[
 \gamma_1=
 \frac{(2/m)^{1/\alpha}}{2-(2/m)^{1/\alpha}}.
\]
For all sufficiently large \(\alpha\), one has
\(\gamma_0<\gamma_1<1\).  On \([\gamma_0,\gamma_1]\), let
\[
 r_m(\gamma)=\frac{2\gamma}{1+\gamma^m}.
\]
The first term in the numerator of
Eq.\eqref{eq:block-energy-lower} gives
\[
 s_m(\gamma)\geq r_m(\gamma)^\alpha.
\]
Moreover,
\[
 \frac{r_m'(\gamma)}{r_m(\gamma)}
 =\frac{1-(m-1)\gamma^m}{\gamma(1+\gamma^m)}.
\]
Thus \(r_m\) first increases and then decreases, so its minimum on
\([\gamma_0,\gamma_1]\) is attained at an endpoint.  At the first endpoint,
\[
 r_m(\gamma_0)^\alpha
 =\exp\left(
 \alpha\log(1+(\log2)/\alpha)
 -\alpha\log(1+\gamma_0^m)
 \right)
 =2+o(1).
\]
At the second endpoint,
\[
 \gamma_1
 =1-\frac{2\log(m/2)}{\alpha}
 +O\left(\frac{\log^2(m/2)}{\alpha^2}\right),
\]
and \(m\log(m/2)/\alpha\to0\).  Therefore
\[
 \alpha\log r_m(\gamma_1)
 =(m-2)\log(m/2)+o(1)\longrightarrow\infty.
\]
This proves \(s_m(\gamma)\geq2-o(1)\) throughout
\([\gamma_0,\gamma_1]\).

\smallskip
\noindent\emph{The alternating-input region \(\gamma\geq\gamma_1\).}
Finally, for \(\gamma\geq\gamma_1\), the increasing function
\(2\gamma/(1+\gamma)\) and Eq.\eqref{eq:alternation-energy-lower} give
\[
 f_\alpha(\gamma)
 \geq\left(\frac{2\gamma}{1+\gamma}\right)^\alpha
 \geq\left(\frac{2\gamma_1}{1+\gamma_1}\right)^\alpha
 =\frac2m.
\]
\smallskip
\noindent\emph{Matching the asymptotic constant.}
This proves Eq.\eqref{eq:combined-periodic-lower}.
Eqs.\eqref{eq:block-energy-lower} and
\eqref{eq:high-power-block-asymptotic} imply
\[
 \liminf_{\alpha\to\infty}
 (\log\alpha)v_\alpha
 \geq2\log2.
\]
Combining this inequality with Eq.\eqref{eq:high-power-limsup} proves
Eq.\eqref{eq:exact-high-power-constant}.
\end{proof}

\section{Proofs for fixed-power finite-size regimes}
\label{app:finite-fixed-proofs}

The first proof treats \(0<\alpha<1\) by separating the stationary boundary
scale, the transient initialization cost, and global localization.  The second
uses the decomposition \(F_N=N\Phi+\Psi+o(1)\): the leading term selects
\(\mathcal M\), and the order-one term selects \(\mathcal S\).

\begin{proof}[Proof of Theorem~\ref{thm:sublinear-boundary}]
\smallskip
\noindent\emph{A coercive envelope.}
Let \(\delta=1-\gamma\) and \(\beta=1/(\alpha+1)\).
We first derive the objective uniformly on the scale
\(\delta=yN^{-\beta}\).  The function
\((g_{d,\alpha}(r)-c_{d,\alpha})/r^2\) extends to a strictly positive
continuous function on \([0,1]\).  Hence there are constants \(m,M>0\)
such that
\begin{equation}
 mr^2\leq g_{d,\alpha}(r)-c_{d,\alpha}\leq Mr^2,
 \qquad 0\leq r\leq1.
 \label{eq:sublinear-quadratic-envelope}
\end{equation}
Writing
\[
 x_t=(1-\delta)^tp_0+\eta_t
\]
and summing the exact second moments gives
\begin{align}
 \sum_{t=0}^{N-1}\E\norm{x_t}^2
 &=s_d\left[
 N\frac{\delta}{2-\delta}
 +\frac{2(1-\delta)}{2-\delta}
 \sum_{t=0}^{N-1}(1-\delta)^{2t}
 \right].
 \label{eq:sublinear-state-sum}
\end{align}
If
\[
 \mathcal E_N^{\rm sub}(\delta)
 =\sum_{t=0}^{N-1}
 \E\bigl[g_{d,\alpha}(\norm{x_t})-c_{d,\alpha}\bigr],
\]
then Eqs.\eqref{eq:sublinear-quadratic-envelope} and
\eqref{eq:sublinear-state-sum} imply
\begin{equation}
 c\left(N\delta+\min\{N,\delta^{-1}\}\right)
 \leq\mathcal E_N^{\rm sub}(\delta)
 \leq C\left(N\delta+\min\{N,\delta^{-1}\}\right)
 \label{eq:sublinear-coercive-envelope}
\end{equation}
for \(0<\delta\leq1/2\).

\smallskip
\noindent\emph{Reduction to the initial ray.}
It remains to identify the leading term of \(\mathcal E_N^{\rm sub}\).  Define the
initial-ray sum
\[
 \mathcal I_N(\delta)=\sum_{t=0}^{N-1}
 \E\bigl[g_{d,\alpha}((1-\delta)^t\norm{p_0})-c_{d,\alpha}\bigr].
\]
For \(d\geq2\), convolution with the indicator of \(\B\) gives a bounded
Hessian for the radial potential.  Since \(\eta_t\) is centered,
\[
 |\mathcal E_N^{\rm sub}(\delta)-\mathcal I_N(\delta)|\leq CN\delta.
\]
For \(d=1\), the derivative of the potential is \(\alpha\)-H\"older, and
the corresponding Taylor remainder gives
\[
 |\mathcal E_N^{\rm sub}(\delta)-\mathcal I_N(\delta)|
 \leq C_\alpha N\delta^{(1+\alpha)/2}.
\]
Both bounds are \(o(N^\beta)\) uniformly for \(y\) in a compact subset of
\((0,\infty)\).

\smallskip
\noindent\emph{The ray-integral limit.}
Let
\[
 \psi(u)=d\int_0^1r^{d-1}
 \bigl(g_{d,\alpha}(ur)-c_{d,\alpha}\bigr)\,dr,
 \qquad \ell=-\log(1-\delta).
\]
The quadratic envelope gives \(0\leq\psi(u)\leq Cu^2\), so the corresponding
Riemann sums converge.  Since the integrand defining \(\psi\) is nonnegative,
Tonelli's theorem \cite[Theorem~2.37(a), p.~67]{Folland1999} permits the order
of integration to be exchanged.  The substitution \(u=e^{-s}r\) then gives
\begin{align*}
 \ell\sum_{t\geq0}\psi(e^{-t\ell})
 &\longrightarrow\int_0^\infty\psi(e^{-s})\,ds\\
 &=d\int_0^1r^{d-1}\int_0^r
 \frac{g_{d,\alpha}(u)-c_{d,\alpha}}{u}\,du\,dr\\
 &=\int_0^1\frac{1-u^d}{u}
 \bigl(g_{d,\alpha}(u)-c_{d,\alpha}\bigr)\,du
 =\frac{H_{d,\alpha}}2.
\end{align*}
The tail after time \(N\) is \(O(\delta^{-1}e^{-2N\delta})\).  Therefore,
uniformly when \(\delta\to0\) and \(N\delta\to\infty\),
\begin{equation}
 \mathcal E_N^{\rm sub}(\delta)=\frac{H_{d,\alpha}}{2\delta}
 +o(\delta^{-1})+o(N^\beta).
 \label{eq:sublinear-ray-asymptotic}
\end{equation}
\smallskip
\noindent\emph{The rescaled limiting profile.}
The definitions of \(F_{d,\alpha,N}\) and \(\mathcal E_N^{\rm sub}\) give the exact
decomposition
\[
 F_{d,\alpha,N}(1-\delta)-Nc_{d,\alpha}
 =Nc_{d,\alpha}\bigl(h_\alpha(1-\delta)-1\bigr)
 +h_\alpha(1-\delta)\mathcal E_N^{\rm sub}(\delta).
\]
Let \(\delta=yN^{-\beta}\), where \(\beta=1/(\alpha+1)\).  Since
\(h_\alpha(1-\delta)-1=\delta^\alpha+O(\delta)\), the two terms on the
right satisfy, uniformly for \(y\) in every compact subset of \((0,\infty)\),
\begin{align*}
 \frac{Nc_{d,\alpha}
 \bigl(h_\alpha(1-\delta)-1\bigr)}{N^\beta}
 &\longrightarrow c_{d,\alpha}y^\alpha,\\
 \frac{h_\alpha(1-\delta)\mathcal E_N^{\rm sub}(\delta)}{N^\beta}
 &\longrightarrow \frac{H_{d,\alpha}}{2y}.
\end{align*}
Adding these limits gives
\[
 \frac{F_{d,\alpha,N}(1-yN^{-\beta})-Nc_{d,\alpha}}{N^\beta}
 \longrightarrow c_{d,\alpha} y^\alpha+\frac{H_{d,\alpha}}{2y}.
\]

\smallskip
\noindent\emph{Localization and optimization.}
The trial value \(\delta=N^{-\beta}\) has excess \(O(N^\beta)\).
Eq.\eqref{eq:sublinear-coercive-envelope}, together with
\(h_\alpha(1-\delta)-1\geq\delta^\alpha/2\) for small \(\delta\), excludes
parameters bounded away from the endpoint, bounded \(N\delta\), and the two
scales \(N^\beta\delta\to0\) and \(N^\beta\delta\to\infty\).  Thus every
minimizer satisfies
\[
 0<a\leq N^\beta(1-\gamma_N)\leq b<\infty.
\]
The limiting function has a unique minimizer.  Differentiating it proves
Eq.\eqref{eq:sublinear-boundary-location}, and evaluating it there proves
Eq.\eqref{eq:sublinear-boundary-value}.
\end{proof}

\begin{proof}[Proof of Theorem~\ref{thm:finite-perturbation} and
Corollary~\ref{cor:finite-nondegenerate}]
\smallskip
\noindent\emph{Coupling with the stationary recursion.}
Let
\[
 b_t(\gamma)=M_{d,\alpha,t}(\gamma)-M_{d,\alpha}(\gamma).
\]
We first establish the expansion in Eq.\eqref{eq:finite-perturbation-identity}.
Extend the input sequence to independent points \((p_t)_{t\in\mathbb Z}\)
uniformly distributed in \(\B\), and let
\[
 \bar x_0=(1-\gamma)\sum_{j=0}^{\infty}\gamma^j p_{-j}.
\]
Thus \(\bar x_0\) has the stationary distribution in
Eq.\eqref{eq:stationary-state}.  Starting from \(x_0=p_0\), define
\(x_t\) and \(\bar x_t\) for \(t\geq1\) by
\[
 x_t=\gamma x_{t-1}+(1-\gamma)p_t,
 \qquad
 \bar x_t=\gamma\bar x_{t-1}+(1-\gamma)p_t.
\]
Subtracting the two recursions gives the exact identity
\begin{equation}
 x_t-\bar x_t=\gamma^t(x_0-\bar x_0).
 \label{eq:finite-stationary-state-difference}
\end{equation}

\smallskip
\noindent\emph{Summability of the transient correction.}
Let \(K=[a,b]\) with \(0<a\leq b<1\), and choose \(r_K\) such that
\(b<r_K<1\).  The geometric series defining \(\bar x_0\), together with
its first two derivatives with respect to \(\gamma\), converges uniformly on
\(K\).  Differentiating Eq.\eqref{eq:finite-stationary-state-difference}
therefore yields
\begin{equation}
 \sup_{\gamma\in K}
 \E\left\|\frac{\partial^j}{\partial\gamma^j}
 (x_t-\bar x_t)\right\|
 \leq C_{K,j}(1+t)^j r_K^t,
 \qquad j=0,1,2.
 \label{eq:finite-stationary-state-derivatives}
\end{equation}
Indeed, each differentiation of \(\gamma^t\) contributes at most one factor
of \(t\), and \(b^t\), \(tb^{t-1}\), and \(t(t-1)b^{t-2}\) are all bounded
by a constant times \((1+t)^2r_K^t\).

Let
\[
 G_{d,\alpha}(x)=\E\norm{p-x}^{\alpha}.
\]
Conditioning on the current state gives
\[
 b_t(\gamma)
 =\E\bigl[G_{d,\alpha}(x_t)-G_{d,\alpha}(\bar x_t)\bigr].
\]
The function \(G_{d,\alpha}\) has two continuous derivatives on \(\B\).
Its Hessian is locally Lipschitz in all the cases considered here except when
\(d=1\) and \(1<\alpha<2\).  In that remaining case,
\[
 G_{1,\alpha}''(x)
 =\frac\alpha2\bigl((1+x)^{\alpha-1}+(1-x)^{\alpha-1}\bigr),
\]
so the Hessian is H\"older continuous of order \(\alpha-1\).  Applying the
chain rule to the last expression for \(b_t\) and using
Eq.\eqref{eq:finite-stationary-state-derivatives} gives, for some \(q_K<1\),
\begin{equation}
 \sup_{\gamma\in K}|b_t^{(j)}(\gamma)|
 \leq C_{K,j}(1+t)^j q_K^t,
 \qquad j=0,1,2.
 \label{eq:finite-perturbation-coupling-bound}
\end{equation}
For the H\"older case, one may take a number \(q_K\) larger than
\(r_K^{\alpha-1}\).  Thus the slower H\"older rate is still exponential.

The bound in Eq.\eqref{eq:finite-perturbation-coupling-bound} makes the series
of derivatives uniformly summable on \(K\), which justifies term-by-term
differentiation.  From the definitions of \(F_{d,\alpha,N}\),
\(\Phi_{d,\alpha}\), and \(\Psi_{d,\alpha}\),
\begin{align*}
 F_{d,\alpha,N}
 &=h_\alpha\sum_{t=0}^{N-1}
   \bigl(M_{d,\alpha}+b_t\bigr),\\
 N\Phi_{d,\alpha}&=Nh_\alpha M_{d,\alpha},\\
 \Psi_{d,\alpha}&=h_\alpha\sum_{t=0}^{\infty}b_t.
\end{align*}
Consequently,
\begin{equation}
 F_{d,\alpha,N}-N\Phi_{d,\alpha}-\Psi_{d,\alpha}
 =-h_\alpha\sum_{t=N}^{\infty}b_t.
 \label{eq:finite-perturbation-tail}
\end{equation}
Choose \(\rho_K\) with \(q_K<\rho_K<1\).  For \(j\leq2\), the
corresponding derivative of the right-hand side is bounded by a constant
times
\[
 \sum_{t=N}^{\infty}(1+t)^2q_K^t=O_K(\rho_K^N).
\]
This proves Eq.\eqref{eq:finite-perturbation-identity} in \(C^2(K)\).

\smallskip
\noindent\emph{Localization on an interior compact interval.}
We next locate the finite minimizers.  Theorem
\ref{thm:stationary-transition} shows that
\(\min_\gamma\Phi_{d,\alpha}(\gamma)<c_{d,\alpha}\) for \(\alpha>1\),
and that every stationary minimizer is in \((1/2,1)\).  Fix
\(\widehat\gamma\in\mathcal M\), and choose \(\eta>0\) such that
\[
 \Phi_{d,\alpha}(\widehat\gamma)=c_{d,\alpha}-4\eta.
\]
The function \(G_{d,\alpha}\) is convex and even for \(\alpha>1\).  Hence
\[
 G_{d,\alpha}(0)
 \leq\frac{G_{d,\alpha}(x)+G_{d,\alpha}(-x)}2
 =G_{d,\alpha}(x).
\]
Since \(G_{d,\alpha}(0)=c_{d,\alpha}\), each term in the finite sum satisfies
\(M_{d,\alpha,t}(\gamma)\geq c_{d,\alpha}\).  Hence
\begin{equation}
 \frac{F_{d,\alpha,N}(\gamma)}{N}
 \geq c_{d,\alpha}h_\alpha(\gamma)
 \qquad (0\leq\gamma\leq1).
 \label{eq:finite-minimizer-endpoint-bound}
\end{equation}
Because \(h_\alpha(0)=h_\alpha(1)=1\), continuity provides
\(\varepsilon>0\) such that the right-hand side of
Eq.\eqref{eq:finite-minimizer-endpoint-bound} is at least
\(c_{d,\alpha}-\eta\) whenever
\(\gamma\in[0,\varepsilon]\cup[1-\varepsilon,1]\).  Decrease
\(\varepsilon\), if necessary, so that \(\mathcal M\) is contained in
\([\varepsilon,1-\varepsilon]\).  On the other hand,
Eq.\eqref{eq:finite-perturbation-identity}, applied at
\(\widehat\gamma\), gives
\[
 \frac{F_{d,\alpha,N}(\widehat\gamma)}{N}
 \leq c_{d,\alpha}-3\eta
\]
for all sufficiently large \(N\).  Therefore every global minimizer belongs
to the fixed compact interval
\[
 K_0=[\varepsilon,1-\varepsilon]
\]
when \(N\) is sufficiently large.

\smallskip
\noindent\emph{Selection among stationary minimizers.}
Let
\[
 \Phi_{\min}=\min_\gamma\Phi_{d,\alpha}(\gamma),
 \qquad
 \mathcal S=\operatorname*{argmin}_{\gamma\in\mathcal M}
 \Psi_{d,\alpha}(\gamma).
\]
Both \(\mathcal M\) and \(\mathcal S\) are nonempty compact subsets of
\(K_0\).  On \(K_0\), let
\[
 R_N=F_{d,\alpha,N}-N\Phi_{d,\alpha}-\Psi_{d,\alpha}.
\]
Eq.\eqref{eq:finite-perturbation-identity} gives
\(\sup_{K_0}|R_N|=o(1)\).  If \(\gamma_N\) minimizes
\(F_{d,\alpha,N}\) and \(\gamma^\dagger\in\mathcal S\), comparison with
\(\gamma^\dagger\)
gives
\begin{equation}
 N\bigl(\Phi_{d,\alpha}(\gamma_N)-\Phi_{\min}\bigr)
 +\Psi_{d,\alpha}(\gamma_N)+R_N(\gamma_N)
 \leq
 \Psi_{d,\alpha}(\gamma^\dagger)+R_N(\gamma^\dagger).
 \label{eq:finite-minimizer-comparison}
\end{equation}
The function \(\Psi_{d,\alpha}\) is bounded on \(K_0\).  Dividing
Eq.\eqref{eq:finite-minimizer-comparison} by \(N\) therefore shows that
\[
 \Phi_{d,\alpha}(\gamma_N)\longrightarrow\Phi_{\min}.
\]
Compactness of \(K_0\) then implies
\(\operatorname{dist}(\gamma_N,\mathcal M)\to0\).

The first term on the left-hand side of
Eq.\eqref{eq:finite-minimizer-comparison} is nonnegative.  Removing it gives
\[
 \Psi_{d,\alpha}(\gamma_N)
\leq\Psi_{d,\alpha}(\gamma^\dagger)+o(1).
\]
Every convergent subsequence of \((\gamma_N)\) has its limit in
\(\mathcal M\), and continuity of \(\Psi_{d,\alpha}\) shows that this limit
belongs to \(\mathcal S\).  It follows that
\(\operatorname{dist}(\gamma_N,\mathcal S)\to0\), which is the asserted
selection rule.  In particular,
\[
 \Psi_{d,\alpha}(\gamma_N)
 \longrightarrow\min_{\gamma\in\mathcal M}\Psi_{d,\alpha}(\gamma).
\]
Evaluating the finite objective at \(\gamma^\dagger\) gives the corresponding upper
bound in Eq.\eqref{eq:finite-selection-value}.  The lower bound follows from
\(\Phi_{d,\alpha}(\gamma_N)\geq\Phi_{\min}\), the preceding convergence of
\(\Psi_{d,\alpha}(\gamma_N)\), and \(R_N(\gamma_N)=o(1)\).  This proves
Eq.\eqref{eq:finite-selection-value}.

\smallskip
\noindent\emph{The nondegenerate case.}
It remains to prove the two refined formulas.  Suppose that
\(\mathcal M=\{\gamma_*\}\) and
\(\Phi_{d,\alpha}''(\gamma_*)>0\).  The selection result gives
\(\gamma_N\to\gamma_*\).  Choose a closed interval \(U\) around
\(\gamma_*\) and a constant \(m>0\) such that
\(\Phi_{d,\alpha}''\geq m\) on \(U\).  The \(C^2\) expansion already proved
gives, uniformly on \(U\),
\[
 F_{d,\alpha,N}''
 =N\Phi_{d,\alpha}''+\Psi_{d,\alpha}''+O(\rho_U^N).
\]
Thus \(F_{d,\alpha,N}''\geq Nm/2\) on \(U\) for all sufficiently large
\(N\).  Every global minimizer then lies in \(U\), and strict convexity on
\(U\) proves its uniqueness.

\smallskip
\noindent\emph{Location and value corrections.}
Let \(e_N=\gamma_N-\gamma_*\).  Since \(\gamma_*\) is an interior
stationary minimizer, \(\Phi_{d,\alpha}'(\gamma_*)=0\), and
\[
 F_{d,\alpha,N}'(\gamma_*)
 =\Psi_{d,\alpha}'(\gamma_*)+O(\rho_U^N)=O(1).
\]
The mean value theorem, applied between \(\gamma_*\) and \(\gamma_N\), now
gives a point \(\xi_N\) between them such that
\[
 e_N
 =-\frac{F_{d,\alpha,N}'(\gamma_*)}
 {F_{d,\alpha,N}''(\xi_N)}=O(N^{-1}).
\]
Expanding the identity
\(F_{d,\alpha,N}'(\gamma_N)=0\) yields
\[
 0
 =N\Phi_{d,\alpha}''(\gamma_*)e_N
  +\Psi_{d,\alpha}'(\gamma_*)+o(1).
\]
Solving for \(e_N\) proves
Eq.\eqref{eq:finite-nondegenerate-location}.

Finally, Taylor expansion at \(\gamma_*\), using
\(e_N=O(N^{-1})\), gives
\begin{align*}
 F_{d,\alpha,N}(\gamma_N)
 &=N\Phi_{d,\alpha}(\gamma_*)+\Psi_{d,\alpha}(\gamma_*)\\
 &\quad+\frac{N}{2}\Phi_{d,\alpha}''(\gamma_*)e_N^2
 +\Psi_{d,\alpha}'(\gamma_*)e_N+o(N^{-1}).
\end{align*}
Substituting
\[
 e_N
 =-\frac{\Psi_{d,\alpha}'(\gamma_*)}
 {N\Phi_{d,\alpha}''(\gamma_*)}+o(N^{-1})
\]
combines the two correction terms into
\(-\Psi_{d,\alpha}'(\gamma_*)^2/
(2N\Phi_{d,\alpha}''(\gamma_*))\).  This proves
Eq.\eqref{eq:finite-nondegenerate-value}.
\end{proof}

\section{Uniform endpoint and transient estimates}
\label{app:finite-endpoint-proofs}

These estimates support the analysis near \(\alpha=1\).  The stationary
endpoint expansion is proved first.  The transient estimate then couples a
finite state to a stationary state driven by the same later inputs, and the
final lemma repeats that coupling with two derivatives.

\begin{proof}[Proof of Lemma~\ref{lem:uniform-stationary-endpoint}]
\smallskip
\noindent\emph{Expansion of the radial potential.}
For \(\norm x<1\), polar integration from \(x\) gives
\begin{equation}
 \frac1{\omega_d}\int_{\B}\norm{p-x}^\alpha\,dp
 =\frac1{\omega_d(d+\alpha)}
 \int_{S^{d-1}}\rho_x(u)^{d+\alpha}\,d\sigma(u),
 \label{eq:joint-radial-potential}
\end{equation}
where
\[
 \rho_x(u)=-\langle x,u\rangle
 +\sqrt{1-\norm x^2+\langle x,u\rangle^2}.
\]
This is the usual radial-volume normalization
\cite{ZhuZhouXu2014}.  On \(\norm x\leq1/2\), the right-hand side has six
uniformly bounded derivatives with respect to \(x\).  Symmetry and direct
differentiation give
\[
 D^2G_{d,\alpha}(0)=\alpha I_d.
\]
For \(d\geq2\), the identity
\[
 \Delta G_{d,\alpha}(x)
 =\alpha(\alpha+d-2)G_{d,\alpha-2}(x)
\]
determines the fourth-order coefficient.  For \(d=1\), direct integration gives
\begin{equation}
 G_{1,\alpha}(x)
 =\frac{(1+x)^{\alpha+1}+(1-x)^{\alpha+1}}
 {2(\alpha+1)},
 \qquad |x|<1,
 \label{eq:joint-d1-potential}
\end{equation}
and its Taylor expansion gives the same coefficient.  In every dimension it
equals
\[
 \frac{\alpha(\alpha-2)(\alpha+d-2)}{8(d+2)}.
\]
Therefore, for
\[
 Q(x)=G_{d,\alpha}(x)-c_{d,\alpha}
 -\frac\alpha2\norm x^2
 -\frac{\alpha(\alpha-2)(\alpha+d-2)}{8(d+2)}\norm x^4,
\]
\begin{equation}
 |Q(x)|\leq C\norm x^6,
 \quad \norm{\nabla Q(x)}\leq C\norm x^5,
 \quad \norm{D^2Q(x)}\leq C\norm x^4
 \label{eq:joint-local-cancellation}
\end{equation}
on this smaller ball.

\smallskip
\noindent\emph{Coefficient and state moment estimates.}
For \(a_j(\delta)=\delta(1-\delta)^j\), let
\[
 s_\delta=\sum_{j\geq0}a_j(\delta)p_j,
\]
where the \(p_j\) are independent uniform points in \(\B\).  A
geometric-series calculation gives
\begin{equation}
 \sum_{j\geq0}a_j^2=\frac\delta{2-\delta},
 \quad
 \sum_{j\geq0}(a_j')^2=\frac2{\delta(2-\delta)^3},
 \quad
 \sum_{j\geq0}(a_j'')^2
 =\frac{8(1-\delta+\delta^2)}
 {\delta^3(2-\delta)^5}.
 \label{eq:joint-coefficient-sums}
\end{equation}
Coordinatewise Hoeffding concentration \cite{Hoeffding1963} then implies,
for every fixed \(q>0\),
\begin{equation}
 \norm{s_\delta}_{L^q}=O(\delta^{1/2}),
 \quad
 \norm{s_\delta'}_{L^q}=O(\delta^{-1/2}),
 \quad
 \norm{s_\delta''}_{L^q}=O(\delta^{-3/2}).
 \label{eq:joint-state-moments}
\end{equation}

\smallskip
\noindent\emph{Control near the boundary.}
For \(d\geq2\) and \(\alpha\in[1-\eta,1+\eta]\), the singular part of the
Hessian kernel is controlled in polar coordinates by
\[
 \int_0^1 r^{d+\alpha-3}\,dr
 \leq \int_0^1 r^{-\eta}\,dr<\infty.
\]
Thus the kernel is uniformly integrable over this range of \(\alpha\), and
\(D^2Q\) is globally bounded on \(\B\).  For \(d=1\), differentiating
Eq.\eqref{eq:joint-d1-potential} gives
\[
 |Q''(x)|
 \leq C_\eta\bigl(1+(1-|x|)^{-\eta}\bigr).
\]
The decomposition \(s_\delta=\delta p_0+z_\delta\) shows that the density of
\(s_\delta\) is at most \((2\delta)^{-1}\).  If
\(w_\delta=1-|s_\delta|\), then, for \(0<q<1\),
\begin{equation}
 \E[w_\delta^{-q};|s_\delta|>1/2]
 \leq C_q\delta^{-q}\Pr(|s_\delta|>1/2)^{1-q}.
 \label{eq:joint-negative-moment}
\end{equation}
The probability on the right is exponentially small in \(1/\delta\).
Taking \(q=2\eta<1\) controls the singular Hessian term.

\smallskip
\noindent\emph{Differentiation under the expectation.}
To justify differentiation, truncate the series for \(s_\delta\) at \(n\).
On compact subsets of \((0,1)\), the truncated states and their first two
derivatives converge pathwise and in every finite \(L^q\).  Choose
\(\varepsilon_0>0\) with \(2\eta(1+\varepsilon_0)<1\).  The density bound
and Eq.\eqref{eq:joint-state-moments} imply uniform integrability of
\[
 Q''(s_{\delta,n})(s_{\delta,n}')^2.
\]
The finite chain rule therefore passes to the limit.  If
\(F(\delta)=\E Q(s_\delta)\), then
\[
 F'=\E[\nabla Q(s_\delta)\cdot s_\delta'],
\]
\[
 F''=\E[(s_\delta')^TD^2Q(s_\delta)s_\delta'
 +\nabla Q(s_\delta)\cdot s_\delta''].
\]
Eqs.\eqref{eq:joint-local-cancellation} and
\eqref{eq:joint-state-moments}, followed by the exponentially small tail
estimate, give
\[
 |F|\leq C\delta^3,
 \qquad |F'|\leq C\delta^2,
 \qquad |F''|\leq C\delta.
\]
\smallskip
\noindent\emph{Completion of the endpoint expansion.}
Exact fourth-moment expansion also gives
\[
 \E\norm{s_\delta}^4=\frac{s_d}{4}\delta^2+O_{C^2}(\delta^3).
\]
Together with
\(\E\norm{s_\delta}^2=s_d\delta/(2-\delta)\), this proves
Eqs.\eqref{eq:joint-M-expansion}--\eqref{eq:joint-M-remainder}.
Multiplication by \((1-\delta)^\alpha+\delta^\alpha\) proves
Eqs.\eqref{eq:joint-Phi-expansion}--\eqref{eq:joint-Phi-remainder}.
\end{proof}

\begin{proof}[Proof of Theorem~\ref{thm:near-one-stationary-uniqueness}]
\smallskip
\noindent\emph{The derivative profile.}
Define the endpoint profile
\[
 P(\delta)
 =\Phi_{d,\alpha}(1-\delta)-c_{d,\alpha}.
\]
Eq.\eqref{eq:joint-Phi-expansion} and
Eq.\eqref{eq:joint-Phi-remainder} imply, for \(\alpha=1+\varepsilon\) in a
fixed sufficiently small one-sided neighborhood of one,
\begin{align}
 P'(\delta)
 &=\alpha c_{d,\alpha}\delta^\varepsilon
 -\alpha c_{d,\alpha} q_{d,\alpha}+u'(\delta),
 \label{eq:near-one-first-derivative}\\
 |u'(\delta)|&\leq C_d\delta,
 &|u''(\delta)|&\leq C_d.
 \label{eq:near-one-remainder-derivatives}
\end{align}
Indeed, the second derivative of
\(\delta^{\alpha+1}\) is bounded for \(\alpha\geq1\), and every remaining
term has the stated bounds.

\smallskip
\noindent\emph{Global localization.}
We first localize global minimizers.  As \(\alpha\to1^+\),
\(h_\alpha\to1\) uniformly on \([0,1]\), while
\[
 \sup_{0\leq r\leq2}|r^\alpha-r|\longrightarrow0.
\]
It follows directly from the expectation defining \(M_{d,\alpha}\), including
its continuous endpoint value, that
\(\Phi_{d,\alpha}\to\Phi_{d,1}\) uniformly on \([0,1]\).
The function \(\Phi_{d,1}\) has the unique minimizer \(\gamma=1\) by
Theorem~\ref{thm:stationary-transition}.  Hence every global minimizer for
\(\alpha\to1^+\) satisfies \(\delta=1-\gamma\to0^+\).

\smallskip
\noindent\emph{Existence and uniqueness of the critical point.}
Choose the one-sided neighborhood so that
\(q_{d,\alpha}\) stays in a compact subinterval of \((0,1)\).  On
\(0<\delta\leq\varepsilon^2\), Eqs.\eqref{eq:near-one-first-derivative} and
\eqref{eq:near-one-remainder-derivatives} give
\[
 P''(\delta)
 =\alpha c_{d,\alpha}\varepsilon\delta^{\varepsilon-1}+O_d(1)>0
\]
for all sufficiently small \(\varepsilon\).  Furthermore,
\[
 \lim_{\delta\to0^+}P'(\delta)
 =-\alpha c_{d,\alpha} q_{d,\alpha}<0,
\]
whereas \((\varepsilon^2)^\varepsilon\to1\) and
\(q_{d,\alpha}\to q_{d,1}<1\), so
\(P'(\varepsilon^2)>0\).  Thus there is exactly one zero of
\(P'\) in this interval.  On
\([\varepsilon^2,\delta_1]\), the leading difference
\(\alpha c_{d,\alpha}(\delta^\varepsilon-q_{d,\alpha})\) is bounded below
by a positive constant when \(\varepsilon\) is small.  Choosing
\(\delta_1>0\) sufficiently small makes the remainder in
Eq.\eqref{eq:near-one-first-derivative} smaller than half that constant.
Consequently, \(P'>0\) on this interval.  Global localization
then shows that the unique zero just found is the unique global minimizer.

\smallskip
\noindent\emph{Location and curvature.}
At \(\delta=\delta_{d,\alpha}\), the two leading terms in
Eq.\eqref{eq:near-one-first-derivative} cancel.  Evaluation at
\(\delta_{d,\alpha}/2\) and
\(2\delta_{d,\alpha}\) places the zero between these two values,
because \(\delta_{d,\alpha}/\varepsilon\to0\).  The mean-value theorem
applied to \(z^\varepsilon\) on \([1/2,2]\) then gives
Eq.\eqref{eq:near-one-stationary-location}.  A final differentiation of
Eq.\eqref{eq:near-one-first-derivative}, together with that location estimate,
gives Eq.\eqref{eq:near-one-stationary-curvature}.
\end{proof}

\begin{proof}[Proof of Lemma~\ref{lem:uniform-transient}]
\smallskip
\noindent\emph{The initial-ray integral.}
Let \(G=G_{d,\alpha}\).  For \(x\in\B\),
\[
 \nabla G(x)
 =\frac{\alpha}{\omega_d}\int_{\B}
 (x-p)\norm{x-p}^{\alpha-2}\,dp.
\]
The absolute value of the integrand near \(p=x\) is
\(\norm{x-p}^{\alpha-1}\), and its local radial integral is proportional to
\[
 \int_0^1r^{d+\alpha-2}\,dr<\infty
\]
because \(\alpha\geq1-\eta>0\).  Since the interval of powers is compact, this formula
gives a constant \(c_{\mathrm{Lip}}=c_{\mathrm{Lip}}(d,\eta)\) such that
\begin{equation}
 |G(x)-G(y)|\leq c_{\mathrm{Lip}}\norm{x-y}
 \qquad (x,y\in\B)
 \label{eq:uniform-potential-lipschitz}
\end{equation}
for every \(\alpha\in[1-\eta,1+\eta]\).

Let \(\gamma=1-\delta\), \(\ell=-\log\gamma\), and
\(a_t=\gamma^t=e^{-t\ell}\).  Define
\[
 \phi_{d,\alpha}(u)
 =\E[G(e^{-u}p_0)-c_{d,\alpha}].
\]
The radial density of \(\norm{p_0}\) is \(d r^{d-1}\,dr\) on \([0,1]\).
Moreover,
Eq.\eqref{eq:uniform-potential-lipschitz} gives
\(\lvert g_{d,\alpha}(s)-c_{d,\alpha}\rvert\leq c_{\mathrm{Lip}}s\), so the following double
integral is absolutely convergent.  Fubini's theorem
\cite[Theorem~2.37(b), p.~67]{Folland1999} and the substitution
\(s=e^{-u}r\) give
\begin{align}
 \int_0^\infty\phi_{d,\alpha}(u)\,du
 &=d\int_0^1r^{d-1}\int_0^\infty
 \bigl(g_{d,\alpha}(e^{-u}r)-c_{d,\alpha}\bigr)\,du\,dr\notag\\
 &=d\int_0^1r^{d-1}\int_0^r
 \frac{g_{d,\alpha}(s)-c_{d,\alpha}}s\,ds\,dr\notag\\
 &=\int_0^1\frac{1-s^d}{s}
 \bigl(g_{d,\alpha}(s)-c_{d,\alpha}\bigr)\,ds
 =\frac{H_{d,\alpha}}2.
 \label{eq:uniform-transient-ray-integral}
\end{align}

\smallskip
\noindent\emph{Discrete approximation of the ray integral.}
The derivative of \(\phi_{d,\alpha}\) satisfies
\[
 |\phi_{d,\alpha}'(u)|
 \leq c_{\mathrm{Lip}}e^{-u}.
\]
Consequently, direct comparison on each interval
\([t\ell,(t+1)\ell]\) gives
\begin{align*}
 \left|\ell\sum_{t=0}^{\infty}\phi_{d,\alpha}(t\ell)
 -\int_0^\infty\phi_{d,\alpha}(u)\,du\right|
 &\leq
 \sum_{t=0}^{\infty}\int_{t\ell}^{(t+1)\ell}
 |\phi_{d,\alpha}(t\ell)-\phi_{d,\alpha}(u)|\,du\\
 &\leq \ell\int_0^\infty|\phi_{d,\alpha}'(u)|\,du
 \leq c_{\mathrm{Lip}}\ell.
\end{align*}
For \(0<\delta\leq1/2\),
\(\ell=-\log(1-\delta)=\delta+O(\delta^2)\).  Combining this fact with
Eq.\eqref{eq:uniform-transient-ray-integral} yields
\begin{equation}
 \left|\delta\sum_{t=0}^{\infty}\phi_{d,\alpha}(t\ell)
 -\frac{H_{d,\alpha}}2\right|\leq C\delta.
 \label{eq:uniform-transient-ray-sum}
\end{equation}

\smallskip
\noindent\emph{Coupling with the stationary state.}
We now compare the transient terms with this sum.  The finite state has the
decomposition
\[
 x_t=a_tp_0+y_t,
 \qquad
 y_t=\delta\sum_{j=1}^{t}\gamma^{t-j}p_j.
\]
Let
\[
 s_\delta=\delta\sum_{j=0}^{\infty}\gamma^j q_j
\]
be formed from an independent sequence of uniform points, and set
\[
 z_t=y_t+a_ts_\delta.
\]
The coefficients in \(z_t\) are
\(\delta,\delta\gamma,\delta\gamma^2,\ldots\).  Thus \(z_t\) has the
stationary distribution in Eq.\eqref{eq:stationary-state}, and
\begin{equation}
 M_{d,\alpha,t}(1-\delta)-M_{d,\alpha}(1-\delta)
 =\E[G(x_t)-G(z_t)].
 \label{eq:uniform-transient-stationary-comparison}
\end{equation}
Independence and centering give
\[
 \E\norm{y_t}^2
 \leq s_d\delta^2\sum_{j=0}^{\infty}\gamma^{2j}
 =s_d\frac{\delta}{2-\delta}\leq\delta,
 \qquad
 \E\norm{z_t}^2=s_d\frac{\delta}{2-\delta}\leq\delta.
\]

\smallskip
\noindent\emph{Early and late times.}
Let \(j_0\) be the least integer for which \(a_{j_0}\leq\sqrt\delta\).  Since
\(\gamma\geq1/2\),
\[
 \frac{\sqrt\delta}{2}<a_{j_0}\leq\sqrt\delta,
 \qquad
 j_0\leq1+\frac{\log(1/\delta)}{2\delta}.
\]
For \(t<j_0\), Eqs.\eqref{eq:uniform-potential-lipschitz} and
\eqref{eq:uniform-transient-stationary-comparison}, followed by
Cauchy--Schwarz, give
\begin{align*}
 &\left|M_{d,\alpha,t}(1-\delta)-M_{d,\alpha}(1-\delta)
 -\phi_{d,\alpha}(t\ell)\right|\\
 &\qquad\leq
 \E|G(a_tp_0+y_t)-G(a_tp_0)|
 +\E|G(z_t)-G(0)|
 \leq2c_{\mathrm{Lip}}\sqrt\delta.
\end{align*}
For \(t\geq j_0\), the identity
\[
 x_t-z_t=a_t(p_0-s_\delta)
\]
and Eq.\eqref{eq:uniform-potential-lipschitz} give
\[
 |M_{d,\alpha,t}(1-\delta)-M_{d,\alpha}(1-\delta)|\leq2c_{\mathrm{Lip}}a_t,
 \qquad
 |\phi_{d,\alpha}(t\ell)|\leq c_{\mathrm{Lip}}a_t.
\]

The shared-input coupling therefore leaves only the explicit factor
\(a_t(p_0-s_\delta)\).  The threshold \(j_0\) separates the initial-ray
estimate, valid while \(a_t>\sqrt\delta\), from the summable geometric tail.

It follows that
\begin{align}
 &\delta\left|
 B_{d,\alpha}(1-\delta)
 -\sum_{t=0}^{\infty}\phi_{d,\alpha}(t\ell)\right|\notag\\
 &\qquad\leq
 2c_{\mathrm{Lip}}\delta j_0\sqrt\delta
 {}+3c_{\mathrm{Lip}}\delta\sum_{t=j_0}^{\infty}a_t
 \leq C\sqrt\delta\left(1+\log\frac1\delta\right).
 \label{eq:uniform-transient-infinite-comparison}
\end{align}

\smallskip
\noindent\emph{Finite truncation.}
Finally, the same stationary comparison, now used for every \(t\geq N\),
gives
\begin{align*}
 \delta|B_{d,\alpha,N}(1-\delta)-B_{d,\alpha}(1-\delta)|
 &\leq2c_{\mathrm{Lip}}\delta\sum_{t=N}^{\infty}a_t\\
 &=2c_{\mathrm{Lip}}a_N
 \leq2c_{\mathrm{Lip}}e^{-N\delta}
 \leq2c_{\mathrm{Lip}}e^{-N\delta/2}.
\end{align*}
Combining this bound with Eqs.\eqref{eq:uniform-transient-ray-sum} and
\eqref{eq:uniform-transient-infinite-comparison}, and increasing
\(c_{\mathrm{tr}}\) if
necessary, proves Eq.\eqref{eq:joint-transient-bound}.
\end{proof}

\begin{proof}[Proof of Lemma~\ref{lem:uniform-differentiated-transient}]
\smallskip
\noindent\emph{State and derivative moment bounds.}
Throughout this proof, let \(G=G_{d,\alpha}\).  Let
\(\gamma=1-\delta\), \(a_t=\gamma^t\), and decompose the finite state as
\[
 x_t=a_tp_0+y_t,
 \qquad
 y_t=\delta\sum_{n=0}^{t-1}\gamma^np_{t-n}.
\]
The calculation follows the three-term decomposition
\begin{equation}
 B_{d,\alpha,N}
 =\mathcal R_{d,\alpha}
 +(B_{d,\alpha}-\mathcal R_{d,\alpha})
 +(B_{d,\alpha,N}-B_{d,\alpha}),
 \label{eq:differentiated-transient-proof-map}
\end{equation}
where, with \(\ell=-\log(1-\delta)\),
\[
 \phi_{d,\alpha}(u)=\E[G(e^{-u}p_0)-c_{d,\alpha}],
 \qquad
 \mathcal R_{d,\alpha}(\delta)
 =\sum_{t\geq0}\phi_{d,\alpha}(t\ell).
\]
The first term produces the explicit singular profile.  The second measures
the coupling error, and the third is the finite-time tail.  Each term is
controlled together with its first two derivatives before multiplication by
\(h_\alpha(1-\delta)\).
For \(b_n(\delta)=\delta\gamma^n\), direct summation gives
\begin{align}
 \sum_{n\geq0}b_n^2
 &=\frac{\delta}{2-\delta},
 \label{eq:transient-coefficient-zero}\\
 \sum_{n\geq0}(b_n')^2
 &=\frac{2}{\delta(2-\delta)^3},
 \label{eq:transient-coefficient-one}\\
 \sum_{n\geq0}(b_n'')^2
 &=\frac{8(1-\delta+\delta^2)}
 {\delta^3(2-\delta)^5}.
 \label{eq:transient-coefficient-two}
\end{align}
The points are independent, centered, and bounded.  The standard moment
estimate for sums of independent bounded vectors therefore implies, for every
fixed finite \(q\geq2\),
\begin{equation}
 \norm{y_t}_{L^q},\norm{s_\delta}_{L^q}=O(\delta^{1/2}),
 \quad
 \norm{y_t'}_{L^q},\norm{s_\delta'}_{L^q}=O(\delta^{-1/2}),
 \quad
 \norm{y_t''}_{L^q},\norm{s_\delta''}_{L^q}=O(\delta^{-3/2}),
 \label{eq:transient-state-moments}
\end{equation}
where \(s_\delta=\delta\sum_{n\geq0}\gamma^np_{-n}\) has the stationary
law.  We also use \(\lvert a_t'\rvert\leq2ta_t\) and
\(\lvert a_t''\rvert\leq4t^2a_t\).

\smallskip
\noindent\emph{Regularity of the potential.}
For \(d\geq2\), differentiation of the geometric integral defining \(G\)
gives, uniformly in the indicated interval of powers,
\begin{equation}
 \norm{D^2G(x)}\leq C,
 \qquad
 \norm{D^2G(x)-D^2G(y)}
 \leq C\omega(\norm{x-y}),
 \label{eq:transient-hessian-modulus}
\end{equation}
where \(\omega(r)=r\log(e/r)\) for \(0<r\leq1\), extended by a constant
for larger \(r\).  The logarithm is needed at \((d,\alpha)=(2,1)\).
Since \(\nabla G(0)=0\), its gradient is bounded by
\(C\norm x\).  In dimension one,
\[
 G(x)
 =\frac{(1+x)^{\alpha+1}+(1-x)^{\alpha+1}}
 {2(\alpha+1)}.
\]
Its second derivative is uniformly bounded and has uniformly bounded total
variation.  Extension by constants outside \([-1,1]\) gives
\begin{equation}
 \int_\R\lvert G''(y+r)-G''(y)\rvert\,dy\leq C\lvert r\rvert.
 \label{eq:transient-BV-translation}
\end{equation}
Conditioning on one uniform summand converts
Eq.\eqref{eq:transient-BV-translation} into the averaged translation bound
needed below.

\smallskip
\noindent\emph{Differentiated ray quadrature.}
We first evaluate the initial-radius term in
Eq.\eqref{eq:differentiated-transient-proof-map}.
For the remainder of the proof, let \(\phi=\phi_{d,\alpha}\).
Eq.\eqref{eq:H-constant} gives
\begin{equation}
 \int_0^\infty\phi_{d,\alpha}(u)\,du=\frac{H_{d,\alpha}}2.
 \label{eq:transient-ray-integral}
\end{equation}
The potential estimates imply
\[
 \lvert\phi(u)\rvert+\lvert\phi'(u)\rvert+\lvert\phi''(u)\rvert
 \leq Ce^{-2u}
\]
and
\[
 \lvert\phi''(u+v)-\phi''(u)\rvert
 \leq Cv\log(e/v)(1+u)e^{-2u},
 \qquad 0<v\leq1.
\]
Applying the cellwise quadrature estimate
\[
 \left|\ell\sum_{t\geq0}\psi(t\ell)-\int_0^\infty\psi(u)\,du\right|
 \leq\sum_{t\geq0}\int_{t\ell}^{(t+1)\ell}
 \lvert\psi(t\ell)-\psi(u)\rvert\,du
\]
to \(\phi\), \(u\phi'\), and \(u^2\phi''\), and using
\[
 \int_0^\infty u\phi'(u)\,du=-\frac{H_{d,\alpha}}2,
 \qquad
 \int_0^\infty u^2\phi''(u)\,du=H_{d,\alpha},
\]
yields
\begin{align}
 \left|\mathcal R_{d,\alpha}-\frac{H_{d,\alpha}}{2\delta}\right|
 &\leq C,\notag\\
 \left|\mathcal R_{d,\alpha}'+\frac{H_{d,\alpha}}{2\delta^2}\right|
 &\leq C\delta^{-1},\notag\\
 \left|\mathcal R_{d,\alpha}''-\frac{H_{d,\alpha}}{\delta^3}\right|
 &\leq C\delta^{-2}\ell_\delta.
 \label{eq:transient-ray-quadrature}
\end{align}

\smallskip
\noindent\emph{Coupling and termwise differentiation.}
We next compare the finite and stationary states before truncating the time
sum.  Let \(s_\delta\) be an independent stationary state and put
\(z_t=y_t+a_ts_\delta\).  Then \(z_t\) is stationary and
\[
 b_t=M_{d,\alpha,t}(1-\delta)-M_{d,\alpha}(1-\delta)
 =\E[G(x_t)-G(z_t)].
\]
The difference \(v_t=x_t-z_t=a_t(p_0-s_\delta)\) has derivatives
\[
 v_t'=a_t'(p_0-s_\delta)-a_ts_\delta',
 \qquad
 v_t''=a_t''(p_0-s_\delta)-2a_t's_\delta'-a_ts_\delta''.
\]
Eqs.\eqref{eq:transient-state-moments} and
\eqref{eq:transient-hessian-modulus} give summable geometric majorants on
compact subintervals of \((0,1)\).  Hence the infinite series for
\(B_{d,\alpha}\) can be differentiated twice term by term.

Only now choose the least integer \(j_0\) for which
\(a_{j_0}\leq\sqrt\delta\).  Then
\[
 \frac{\sqrt\delta}{2}<a_{j_0}\leq\sqrt\delta,
 \qquad
 j_0\leq C\delta^{-1}\ell_\delta.
\]
\smallskip
\noindent\emph{Early-time estimates.}
For \(t<j_0\), Taylor expansion around \(a_tp_0\), centering of
\(y_t,y_t',y_t''\), and Eqs.\eqref{eq:transient-state-moments} and
\eqref{eq:transient-hessian-modulus} give the following termwise estimates.
Let
\[
 e_t=b_t-\phi(t\ell).
\]
For a centered vector \(v\) independent of \(a\), the identity
\[
 \E[G(a+v)-G(a)]
 =\int_0^1(1-s)\E\bigl[v^TD^2G(a+sv)v\bigr]\,ds
\]
removes the first-order term.  Applying this identity to the finite and
stationary perturbations, and then differentiating it, yields
\begin{align*}
 |e_t|&\leq C\delta,\\
 |e_t'|&\leq C\bigl(1+ta_t\delta^{1/2}\bigr),\\
 |e_t''|&\leq C\ell_\delta\bigl(
 \delta^{-1}+ta_t\delta^{-1/2}+t^2a_t\delta^{1/2}\bigr).
\end{align*}
The Hessian modulus controls every differentiated Hessian difference.  In
dimension one, Eq.\eqref{eq:transient-BV-translation} gives the same averaged
bound.  Finally,
\[
 \sum_{t\geq0}t^ma_t\leq C_m\delta^{-m-1},
 \qquad
 j_0\leq C\delta^{-1}\ell_\delta,
 \qquad
 \delta^{1/2}\ell_\delta\leq C.
\]
Summing the three termwise estimates therefore gives
\begin{align}
 \sum_{t<j_0}\lvert b_t-\phi(t\ell)\rvert
 &\leq C\ell_\delta,\notag\\
 \sum_{t<j_0}\lvert b_t'-\partial_\delta\phi(t\ell)\rvert
 &\leq C\delta^{-3/2},\notag\\
 \sum_{t<j_0}\lvert b_t''-\partial_\delta^2\phi(t\ell)\rvert
 &\leq C\delta^{-5/2}\ell_\delta.
 \label{eq:transient-early-sum}
\end{align}

\smallskip
\noindent\emph{Late-time estimates.}
For \(t\geq j_0\), the synchronous coupling and the same moment bounds give
\begin{align*}
 \lvert b_t\rvert
 &\leq C(a_t\delta+a_t^2),\\
 \lvert b_t'\rvert
 &\leq C(ta_t^2+a_t\delta^{-1/2}+ta_t\delta+a_t),
\end{align*}
and
\[
 \lvert b_t''\rvert\leq C\bigl[
 \omega(a_t)(t^2a_t^2+\delta^{-1})+t^2a_t^2
 +ta_t\delta^{-1/2}+a_t\delta^{-1}+a_t\delta^{-3/2}
 +t^2a_t\delta+ta_t\bigr].
\]
In dimension one, the first term is bounded by
\(C\min\{1,a_t/\delta\}(t^2a_t^2+\delta^{-1})\).
For \(r>0\) and \(m=0,1,2\), direct summation gives
\[
 \sum_{t\geq j_0}t^ma_t^r
 \leq C_{m,r}a_{j_0}^r\delta^{-m-1}
 \bigl(1+(j_0\delta)^m\bigr).
\]
Together with \(a_{j_0}\leq\sqrt\delta\),
\(j_0\delta\leq C\ell_\delta\), and
\(\omega(a_t)\leq Ca_t^{1/2}(1+t\delta)\), this bounds each displayed tail
term.  Summation gives
\begin{align}
 \left|B_{d,\alpha}(1-\delta)-\mathcal R_{d,\alpha}(\delta)\right|
 &\leq C\ell_\delta,\notag\\
 \left|\partial_\delta B_{d,\alpha}(1-\delta)-\mathcal R_{d,\alpha}'(\delta)\right|
 &\leq C\delta^{-3/2},\notag\\
 \left|\partial_\delta^2B_{d,\alpha}(1-\delta)-\mathcal R_{d,\alpha}''(\delta)\right|
 &\leq C\delta^{-5/2}\ell_\delta.
 \label{eq:transient-infinite-comparison}
\end{align}

\smallskip
\noindent\emph{Finite truncation and the two-edge factor.}
The finite tail begins at \(N\).  For every fixed \(\theta>0\) and integer
\(m\geq0\),
\[
 \sum_{t=N}^\infty t^m\gamma^{\theta t}
 \leq C_{m,\theta}\delta^{-m-1}
 (1+(N\delta)^m)e^{-\theta N\delta/2}.
\]
Applying this estimate to the coupled bounds above produces the three
exponential terms in
Eqs.\eqref{eq:differentiated-transient-zero}--
\eqref{eq:differentiated-transient-two} for \(B_{d,\alpha,N}\).
Finally,
\[
 h_\alpha(1-\delta)=1+O(\delta),
 \qquad h_\alpha'=O(1),
 \qquad h_\alpha''=O(\delta^{-1})
\]
uniformly for \(\alpha\in[1,1+\eta]\).  The product rule shows that the
additional terms are absorbed by the displayed bounds for
\(\mathcal T_{d,\alpha,N}\).
\end{proof}

\section{Proofs for the joint finite-size windows}
\label{app:finite-window-proofs}

The arguments are arranged in dependency order.  Global localization first
restricts every minimizer to endpoint scales.  The critical and supercritical
proofs then compare their explicit one-variable profiles with uniform
remainders.  The final proof supplies the \(C^2\) convergence used for local
uniqueness.  Figure~\ref{fig:joint-window-guide} identifies the scales before
the detailed estimates begin.

\begin{proof}[Proof of Lemma~\ref{lem:joint-global-localization}]
\smallskip
\noindent\emph{Excluding \(N\delta=O(1)\).}
Choose \(0<\eta<1/2\) so that
\(\alpha_N\in[1-\eta,1+\eta]\) for all sufficiently large \(N\).  We first
exclude \(N(1-\gamma_N)=O(1)\).  Let \(\delta=1-\gamma\) and suppose that
\(N\delta\leq K\) for some fixed \(K>0\).  The recursion gives
\[
 x_t=(1-\delta)^tp_0+y_{t,\delta},
 \qquad
 y_{t,\delta}
 =\delta\sum_{j=1}^{t}(1-\delta)^{t-j}p_j.
\]
Independence and centering imply
\[
 \E\norm{y_{t,\delta}}^2
 \leq s_d\frac{\delta}{2-\delta}
 \leq\frac{C_K}{N}.
\]
For all sufficiently large \(N\), \(\delta\leq1/2\), and hence
\[
 (1-\delta)^t
 \geq e^{-2t\delta}
 \geq e^{-2K}
 \qquad (0\leq t<N).
\]
Strict radial increase of \(g_{d,\alpha}\), continuity in \(\alpha\), and
compactness give
\[
 m_K
 =\min_{\alpha\in[1-\eta,1+\eta]}
 \E\left[
 g_{d,\alpha}(e^{-2K}\norm{p_0})-c_{d,\alpha}
 \right]>0.
\]
By Eq.\eqref{eq:uniform-potential-lipschitz} and Cauchy--Schwarz,
\[
 M_{d,\alpha,t}(1-\delta)
 =\E G_{d,\alpha_N}(x_t)
 \geq c_N+m_K-C_KN^{-1/2}
 \geq c_N+\frac{m_K}{2}
\]
for all sufficiently large \(N\), uniformly for \(0\leq t<N\).  If
\(\alpha_N\leq1\), then \(h_{\alpha_N}(1-\delta)\geq1\).  If
\(\alpha_N\geq1\), then
\[
 h_{\alpha_N}(1-\delta)
 \geq(1-\delta)^{\alpha_N}
 \geq1-(1+\eta)\delta.
\]
It follows that
\[
 F_N(1-\delta)-Nc_N
 \geq \frac{m_K}{2}N-O_K(1).
\]
This order-\(N\) gap is incompatible with the assumed trial value of excess
\(o(N)\).  Therefore every global minimizer satisfies
\[
 N(1-\gamma_N)\longrightarrow\infty.
\]

\smallskip
\noindent\emph{Excluding parameters away from the endpoint.}
We now exclude parameters separated from one.  Fix \(\delta_0>0\) and let
\(0\leq\gamma\leq1-\delta_0\).  Extend the input sequence to negative
indices, set
\[
 \bar x_0=(1-\gamma)\sum_{j=0}^{\infty}\gamma^jp_{-j},
\]
and drive \(x_t\) and \(\bar x_t\) with the same points
\(p_1,p_2,\ldots\).  The state \(\bar x_t\) is stationary, and subtraction of
the recursions gives
\[
 x_t-\bar x_t=\gamma^t(p_0-\bar x_0).
\]
In particular,
\(\E G_{d,\alpha_N}(\bar x_t)=M_{d,\alpha_N}(\gamma)\) for every \(t\).
The definitions of the finite and stationary objectives therefore give
\begin{align*}
 F_N(\gamma)
 &=h_{\alpha_N}(\gamma)
   \sum_{t=0}^{N-1}\E G_{d,\alpha_N}(x_t),\\
 N\Phi_{d,\alpha_N}(\gamma)
 &=h_{\alpha_N}(\gamma)
   \sum_{t=0}^{N-1}\E G_{d,\alpha_N}(\bar x_t).
\end{align*}
Subtracting these identities and applying
Eq.\eqref{eq:uniform-potential-lipschitz} term by term yields
\begin{align*}
 \left|F_N(\gamma)-N\Phi_{d,\alpha_N}(\gamma)\right|
 &\leq h_{\alpha_N}(\gamma)c_{\mathrm{Lip}}
 \sum_{t=0}^{N-1}\E\norm{x_t-\bar x_t}\\
 &\leq
 2h_{\alpha_N}(\gamma)c_{\mathrm{Lip}}
 \sum_{t=0}^{\infty}\gamma^t
 \leq C_{\delta_0}.
\end{align*}
The last bound is uniform in \(N\), \(\alpha_N\in[1-\eta,1+\eta]\), and
\(0\leq\gamma\leq1-\delta_0\), because
\(h_{\alpha_N}(\gamma)\) is bounded on this set and
\(\sum_{t\geq0}\gamma^t\leq\delta_0^{-1}\).  This proves directly that the
finite objective differs from its stationary approximation by at most a
constant depending only on \(d\), \(\eta\), and \(\delta_0\).

Theorem~\ref{thm:stationary-transition} gives
\[
 m_{\delta_0}
 =\min_{0\leq\gamma\leq1-\delta_0}
 \bigl(\Phi_{d,1}(\gamma)-c_{d,1}\bigr)>0.
\]
Moreover, \(h_\alpha\to1\) uniformly on \([0,1]\) and
\(\sup_{0\leq r\leq2}|r^\alpha-r|\to0\) as \(\alpha\to1\).  Therefore
\(\Phi_{d,\alpha_N}\to\Phi_{d,1}\) uniformly on \([0,1]\), and
\(c_N\to c_{d,1}\).  For all sufficiently large \(N\),
\[
 \Phi_{d,\alpha_N}(\gamma)-c_N
 \geq\frac{m_{\delta_0}}2
 \qquad (0\leq\gamma\leq1-\delta_0).
\]
Combining the last two estimates gives
\[
 F_N(\gamma)-Nc_N
 \geq\frac{m_{\delta_0}}2N-C_{\delta_0}.
\]
The trial value again excludes this order-\(N\) gap.  Since
\(\delta_0>0\) was arbitrary, \(1-\gamma_N\to0\).
\end{proof}

\begin{proof}[Proof of Theorem~\ref{thm:joint-critical-crossover}]
\smallskip
\noindent\emph{The rescaled profile.}
Let \(\delta=\delta_Nz\).  Lemmas
\ref{lem:uniform-stationary-endpoint} and
\ref{lem:uniform-transient} reduce the excess, after multiplication by
\(\delta_N\), to
\begin{equation}
 P_N(z)=c_Nq_NN\varepsilon_N\delta_N^2(z\log z-z)
 +\frac{H_N}{2z}
 \label{eq:joint-critical-profile}
\end{equation}
This model profile is strictly convex.  Its critical equation is
\[
 c_Nq_NN\varepsilon_N\delta_N^2\log z
 =\frac{H_N}{2z^2},
\]
and its unique minimizer is
\begin{equation}
 z=e^{w_N/2}
 =\left(
 \frac{H_N}
 {c_Nq_NN\varepsilon_N\delta_N^2w_N}
 \right)^{1/2}.
 \label{eq:joint-critical-model-minimizer}
\end{equation}
Thus the two coordinates used below are related by
\[
 \delta=\delta_Nz=\Delta_Nu,
 \qquad z=e^{w_N/2}u,
 \qquad \Delta_N=\delta_Ne^{w_N/2}.
\]
The remainder estimate now shows that the exact finite objective inherits
this scale and minimizer.

\smallskip
\noindent\emph{Remainder control.}
The remainder is
\[
 R_N(z)
 =\delta_N\bigl[
 F_N(1-\delta_Nz)-Nc_N\bigr]-P_N(z).
\]
The Taylor formula
\[
 z^{\varepsilon_N}
 =1+\varepsilon_N\log z
 +O\bigl(\varepsilon_N^2\log^2z\bigr)
\]
and Eqs.\eqref{eq:joint-Phi-expansion}--
\eqref{eq:joint-Phi-remainder} show that, for fixed
\(0<c<C<\infty\), the stationary part of
\(R_N(e^{w_N/2}u)-R_N(e^{w_N/2})\) is bounded by
\[
 C\left[
 c_Nq_NN\varepsilon_N\delta_N^2\varepsilon_Ne^{w_N/2}(1+w_N^2)
 +N\delta_N^3e^{w_N}
 +N\delta_N^4e^{3w_N/2}
 \right]
\]
uniformly for \(c\leq u\leq C\).  The differentiated transient estimate adds
\(C\delta_N\ell_{\Delta_N}+o(H_N/e^{w_N/2})\).
Under Eq.\eqref{eq:joint-critical-assumption},
\[
 \delta_N=N^{-1/2+o(1)},\qquad
 w_N=O(\log\log N),\qquad
 \Delta_N=N^{-1/2+o(1)}.
\]
Since \(w_Ne^{w_N}\) equals the argument of \(W_0\) in
Eq.\eqref{eq:joint-critical-W-scale},
\begin{equation}
 \frac12c_Nq_NN\varepsilon_N\delta_N^2w_N
 =\frac{H_N}{2e^{w_N}}.
 \label{eq:joint-critical-Lambert-identity}
\end{equation}
Substitution of Eq.\eqref{eq:joint-critical-Lambert-identity} in the first
remainder term, together with the preceding bounds, gives
\begin{equation}
 \sup_{c\leq u\leq C}
 \frac{e^{w_N/2}
 \lvert R_N(e^{w_N/2}u)-R_N(e^{w_N/2})\rvert}
 {H_N(1+w_N^{-1})}
 \longrightarrow0.
 \label{eq:joint-critical-normalized-remainder}
\end{equation}

\smallskip
\noindent\emph{Normalized profile gap.}
Consequently, for \(u>0\),
\begin{align}
 P_N(e^{w_N/2}u)-P_N(e^{w_N/2})
 =\frac{H_N}{2e^{w_N/2}}
 \left[
 u+u^{-1}-2
 +\frac2{w_N}(u\log u-u+1)
 \right].
 \label{eq:joint-critical-profile-gap}
\end{align}
Both functions in brackets are nonnegative and vanish only at \(u=1\).
\smallskip
\noindent\emph{Multiplicative localization.}
The exact stationary leading terms and the positive transient term, before
the Taylor expansion in \(z^{\varepsilon_N}\), give
\[
 \delta_N\bigl[
 F_N(1-\Delta_Nu)
 -F_N(1-\Delta_N)\bigr]
 \geq\frac{H_N}{4e^{w_N/2}}
 (u+u^{-1})-\frac{C H_N}{e^{w_N/2}}
\]
whenever \(u\leq M^{-1}\) or \(u\geq M\), for \(M\) fixed and then
sufficiently large.  This excludes \(u\to0\) and \(u\to\infty\).
On \(M^{-1}\leq u\leq M\), if \(\lvert\log u\rvert\geq\tau>0\), then
\[
 u+u^{-1}-2
 =2\bigl(\cosh(\log u)-1\bigr)
 \geq2(\cosh\tau-1).
\]
By Eq.\eqref{eq:joint-critical-normalized-remainder}, the remainder is
negligible on this compact interval.  The gap in
Eq.\eqref{eq:joint-critical-profile-gap} excludes every fixed multiplicative
separation from one.  Lemma~\ref{lem:joint-global-localization} handles the
remaining endpoint scales.
Hence
\((1-\gamma_N)/(\delta_N e^{w_N/2})\to1\), which proves
Eq.\eqref{eq:joint-critical-location}.

\smallskip
\noindent\emph{The minimum value.}
At \(e^{w_N/2}\),
\[
 P_N(e^{w_N/2})
 =\frac{H_N}{e^{w_N/2}}\frac{w_N-1}{w_N}.
\]
Division by \(\delta_N\) proves
Eq.\eqref{eq:joint-critical-value}.
\end{proof}

\begin{proof}[Proof of Theorem~\ref{thm:joint-supercritical} and
Corollary~\ref{cor:joint-supercritical-crossover}]
\noindent\emph{Localization and leading scale.}
For \(\delta=\delta_Nz\), divide the leading stationary excess by
\(N\varepsilon_Nc_Nq_N\delta_N\).  It converges locally
uniformly to \(z\log z-z\).  Let
\[
 \vartheta(z)=z\log z-z+1.
\]
Then \(\vartheta\geq0\), with equality only at one, and
\(\vartheta''(z)=1/z\).  Hence, on every fixed interval
\([a,b]\subset(0,\infty)\) containing one,
\begin{equation}
 \vartheta(z)\geq\frac{(z-1)^2}{2b}.
 \label{eq:joint-supercritical-profile-coercivity}
\end{equation}
Moreover, \(\vartheta\) has a positive minimum on
\((0,a]\cup[b,\infty)\) when \(a\) is sufficiently small and \(b\) is
sufficiently large.

The two correction scales are smaller than the leading scale because
\[
 \frac{N\delta_N^2+\delta_N^{-1}}
 {N\varepsilon_Nc_Nq_N\delta_N}
 =O\left(\frac{\delta_N}{\varepsilon_N}
 +\frac1{N\varepsilon_N\delta_N^2}\right)=o(1).
\]
Lemma~\ref{lem:joint-global-localization},
Eq.\eqref{eq:joint-supercritical-profile-coercivity}, and the uniform
endpoint expansion therefore place every
minimizer in one fixed interval \(a\leq z\leq b\).  On this interval,
Eq.\eqref{eq:joint-supercritical-profile-coercivity} forces \(z\to1\).
This proves
Eqs.\eqref{eq:joint-supercritical-location} and
\eqref{eq:joint-supercritical-leading-value}.

\smallskip
\noindent\emph{Second-order value.}
On that neighborhood, Eq.\eqref{eq:joint-Phi-expansion} has a uniform
\(o(N\delta_N^2)\) remainder and
 Lemma~\ref{lem:uniform-transient} has a uniform
 \(o(\delta_N^{-1})\) remainder.  The leading profile \(\vartheta\) is minimized
at \(z=1\).  It remains to verify that evaluation at one preserves the
 second-order value.  Choose \(c>0\) sufficiently small.  The differentiated
endpoint expansion and
Eqs.\eqref{eq:differentiated-transient-one}--
\eqref{eq:differentiated-transient-two} give
\[
 \partial_\delta^2F_N(1-\delta)
 \geq\frac{cN\varepsilon_N}{\delta_N}
\]
throughout the fixed multiplicative neighborhood, as well as
\[
 \left|
 \left.\partial_\delta F_N(1-\delta)
 \right|_{\delta=\delta_N}
 \right|
 \leq C\bigl(N\delta_N+\delta_N^{-2}\bigr)
 +o\bigl(N\delta_N+\delta_N^{-2}\bigr).
\]
Strong convexity and the first-order condition imply
\[
 0\leq
 F_N(1-\delta_N)
 -F_N(\gamma_N)
 \leq
 \frac{\delta_N}{2cN\varepsilon_N}
 \left|
 \left.\partial_\delta F_N(1-\delta)
 \right|_{\delta=\delta_N}
 \right|^2.
\]
The right-hand side is bounded by
\[
 C\left(
 \frac{N\delta_N^3}{\varepsilon_N}
 +\frac1{N\varepsilon_N\delta_N^3}
 +\frac1{\varepsilon_N}
 \right).
\]
The three terms are \(o(N\delta_N^2+\delta_N^{-1})\), respectively,
because \(\delta_N/\varepsilon_N\to0\), \(N\varepsilon_N\delta_N^2\to\infty\), and
\(\varepsilon_N^{-1}=o(\delta_N^{-1})\).  Evaluation at one proves
Eq.\eqref{eq:joint-supercritical-second-order}.  The centered formulas follow
from
\[
 \log q_{d,1+\varepsilon}
 =-\frac{\lambda^\star(d)}2-1/(3d+7)\varepsilon+O(\varepsilon^2).
\]

\smallskip
\noindent\emph{Refined location.}
For the refined location, normalize the exact finite objective by \(N\) and
let
\[
 J_N(\delta)=\frac1N F_N(1-\delta).
\]
On every fixed multiplicative neighborhood
\(a\delta_N\leq\delta\leq b\delta_N\),
Eqs.\eqref{eq:joint-Phi-expansion},
\eqref{eq:joint-Phi-remainder},
\eqref{eq:differentiated-transient-one}, and
\eqref{eq:differentiated-transient-two} give, uniformly for
\(a\leq z\leq b\),
\begin{align}
 J_N'(\delta_N)
 ={}&\left(2\chi_{d,\alpha_N}
 +\varepsilon_N\frac{\alpha_N s_d}{4}q_N\right)\delta_N
 -\frac{H_N}{2N\delta_N^2}\notag\\
 &+o\left(\delta_N+\frac1{N\delta_N^2}\right),
 \label{eq:joint-supercritical-slope}\\
 J_N''(\delta_Nz)
 ={}&\frac{\alpha_Nc_Nq_N\varepsilon_N}{\delta_N}
 z^{\varepsilon_N-1}
 +o\left(\frac{\varepsilon_N}{\delta_N}\right).
 \label{eq:joint-supercritical-stiffness}
\end{align}
The algebraic errors from the first and second transient derivatives are
smaller by factors \(\delta_N^{1/2}\) and
\(\delta_N^{1/2}\ell_{\delta_N}\), respectively.  The truncation errors
are exponentially small because \(N\delta_N\to\infty\).

Let \(\rho=(1-\gamma_N)/\delta_N-1\).  The qualitative location
already proved and Eq.\eqref{eq:joint-supercritical-stiffness}, followed by
the mean-value theorem, first give
\[
 \rho=O\left(
 \frac{\delta_N}{\varepsilon_N}+\frac1{N\varepsilon_N\delta_N^2}
 \right).
\]
Expanding the first-order condition on this scale gives
\begin{align*}
 0={}&\alpha_Nc_Nq_N\varepsilon_N\rho
 +\left(2\chi_{d,\alpha_N}
 +\varepsilon_N\frac{\alpha_N s_d}{4}q_N\right)\delta_N\\
 &-\frac{H_N}{2N\delta_N^2}
 +o\left(\delta_N+\frac1{N\delta_N^2}\right).
\end{align*}
Division by \(\alpha_Nc_Nq_N\varepsilon_N\) proves
Eq.\eqref{eq:joint-supercritical-location-refined}.

\smallskip
\noindent\emph{Centered corrections and crossover.}
Finally,
\[
 \frac{\delta_N/\varepsilon_N}
 {1/(N\varepsilon_N\delta_N^2)}
 =N\delta_N^3.
\]
Eqs.\eqref{eq:joint-supercritical-delta-centered}--
\eqref{eq:joint-supercritical-transient-scale} show that the two corrections
have the same order when \(\lambda=3\lambda^\star(d)/2\).  Substitution at
this value gives Eq.\eqref{eq:joint-supercritical-crossover-location}.
\end{proof}

\begin{proof}[Proof of Theorem~\ref{thm:moving-differentiated-profile}]
\smallskip
\noindent\emph{Exact balance identity.}
The definitions of \(w_N\) and \(\Delta_N\) give the exact identity
\begin{equation}
 N\varepsilon_Nc_Nq_N\Delta_N^2w_N=H_N.
 \label{eq:moving-Lambert-identity}
\end{equation}
\smallskip
\noindent\emph{Moving-window expansion and uniform remainder.}
The endpoint expansion in
Lemma~\ref{lem:uniform-stationary-endpoint} and
Corollary~\ref{cor:differentiated-transient-windows} give, uniformly with two
derivatives when \(\delta=\Delta_Nz\) and \(z\in[\ell,u]\),
\begin{align}
 F_N(1-\delta)-Nc_N
 ={}&Nc_Nq_N\delta
 \left[\left(\frac{\delta}{\delta_N}\right)^{\varepsilon_N}
 -(1+\varepsilon_N)\right]
 +\frac{H_N}{2\delta}
 +\mathcal R_N^{\rm mov}(\delta),
 \label{eq:moving-objective-expansion}
\end{align}
where the centered, rescaled remainder is
\begin{equation}
 \frac1{1+w_N^{-1}}
 \left\|\Delta_N[\mathcal R_N^{\rm mov}(\Delta_N\,\cdot)
 -\mathcal R_N^{\rm mov}(\Delta_N)]
 \right\|_{C^2([\ell,u])}=o(1).
 \label{eq:moving-objective-remainder}
\end{equation}
Indeed, the stationary remainder contributes at most
\[
 O\left(
 \frac{\Delta_N^{1-\eta}}
 {\varepsilon_N(1+w_N)}
 \right).
\]
\smallskip
\noindent\emph{Critical and supercritical scales.}
In the critical regime,
\[
 \delta_N=N^{-1/2+o(1)},\qquad
 w_N=O(\log\log N),\qquad
 \Delta_N=N^{-1/2+o(1)},\qquad
 \varepsilon_N^{-1}=O(\log N).
\]
Consequently,
\[
 \frac{\Delta_N^{1-\eta}}{\varepsilon_N(1+w_N)}
 \leq\frac{\Delta_N^{1-\eta}}{\varepsilon_N}
 =N^{-(1-\eta)/2+o(1)}\log N=o(1).
\]
In the supercritical regime, let
\(p=\lambda^\star(d)/(2\lambda)\in(0,1/2)\).  Then
\[
 \delta_N=N^{-p+o(1)},\qquad
 w_N\longrightarrow0,\qquad
 \Delta_N=N^{-p+o(1)},\qquad
 \varepsilon_N^{-1}=O(\log N),
\]
and hence
\[
 \frac{\Delta_N^{1-\eta}}{\varepsilon_N(1+w_N)}
 =N^{-p(1-\eta)+o(1)}\log N=o(1).
\]
Eqs.\eqref{eq:differentiated-transient-zero}--
\eqref{eq:differentiated-transient-two} contribute
\(O(\Delta_N^{1/2}\ell_{\Delta_N})+o(1)=o(1)\) in both regimes.  These estimates also give
\(\varepsilon_N(1+w_N)\to0\), as required below.

\smallskip
\noindent\emph{Identification of the profile.}
Uniformly in \(C^2([\ell,u])\),
\[
 e^{\varepsilon_N(w_N/2+\log z)}
 =1+\varepsilon_N(w_N/2+\log z)
 +O\bigl(\varepsilon_N^2(1+w_N)^2\bigr).
\]
Substitution into Eq.\eqref{eq:moving-objective-expansion}, followed by
Eq.\eqref{eq:moving-Lambert-identity}, gives the centered leading expression
\[
 \frac{H_N}2(z+z^{-1}-2)
 +\frac{H_N}{w_N}(z\log z-z+1),
\]
which proves Eq.\eqref{eq:moving-C2-convergence}.
\end{proof}

\section{Proof of the first-order high-dimensional expansion}
\label{app:high-dimensional-proof}

The leading high-dimensional limit remains proved in the main text.  The
refinement below expands the radial power around its deterministic limit,
computes the first centered moments, and then controls the Taylor remainder
uniformly together with two derivatives.

\begin{proof}[Proof of Theorem~\ref{thm:high-dimensional-first-order}]
\noindent\emph{Moment identities.}
Let \(\beta=\alpha/2\), and represent the stationary difference as
\[
 p-x_\gamma=\sum_{j=-1}^{\infty}b_jp_j,
 \qquad
 b_{-1}=1,
 \qquad
 b_j=-(1-\gamma)\gamma^j\quad(j\geq0),
\]
where \(p_{-1}=p\).  For coefficient sequences \(a,c\), let
\[
 y_a=\sum_i a_ip_i,
 \qquad
 g_{ac}=y_a\mathbin{\cdot}y_c.
\]
The main scalar quantities are
\[
 r=\norm{p-x_\gamma}^2,
 \qquad \sigma^2=\sigma^2(\gamma)=\frac2{1+\gamma},
 \qquad z=r-\sigma^2.
\]
Accordingly, the target expansion has the form
\[
 \E r^\beta
 =\sigma^\alpha
 +\beta\sigma^{\alpha-2}\E z
 +\frac{\beta(\beta-1)}2\sigma^{\alpha-4}\E z^2
 +O_\alpha(d^{-2}).
\]
The Gram identities determine the two displayed centered moments; the rest of
the proof certifies the remainder and its derivatives.
The second and fourth coordinate moments of a uniform point in \(\B\) give
\begin{align}
 \operatorname{Cov}(g_{ac},g_{uv})
 ={}&\frac d{(d+2)^2}
 \bigl(
 \langle a,u\rangle\langle c,v\rangle
 +\langle a,v\rangle\langle c,u\rangle
 \bigr) \notag\\
 &-\frac{2d}{(d+2)(d+4)}
 \sum_i a_ic_iu_iv_i.
 \label{eq:gram-covariance}
\end{align}
This identity follows first for finite sums by expanding the four scalar
products.  The infinite-series version follows by bounded convergence.

Let
\[
 \sigma_4=\sum_{j=-1}^{\infty}b_j^4
 =1+\frac{(1-\gamma)^4}{1-\gamma^4},
 \qquad
 \sigma_6=\sum_{j=-1}^{\infty}b_j^6
 =1+\frac{(1-\gamma)^6}{1-\gamma^6}.
\]
Eq.\eqref{eq:gram-covariance} yields
\begin{equation}
 \mathbb Ez=-\frac{2\sigma^2}{d+2},
 \qquad
 \mathbb Ez^2
 =\frac{2\sigma^4}{d+2}
 -\frac{2d\sigma_4}{(d+2)(d+4)}.
 \label{eq:high-dimensional-first-moments}
\end{equation}
The first three even cumulants of one coordinate of a uniform point in
\(\B\) are
\[
 \kappa_2=\frac1{d+2},
 \qquad
 \kappa_4=-\frac6{(d+2)^2(d+4)},
 \qquad
 \kappa_6=\frac{240}{(d+2)^3(d+4)(d+6)}.
\]
Cumulants add under the independent weighted sum \(y_b\).  Rotational
invariance then reconstructs its third radial moment as
\begin{equation}
 \mathbb Er^3
 =\frac d{(d+2)^2}
 \left[
  (d+4)\sigma^6-6\sigma^2\sigma_4+\frac{16\sigma_6}{d+6}
 \right].
 \label{eq:high-dimensional-third-radial-moment}
\end{equation}
Combining this identity with Eq.\eqref{eq:high-dimensional-first-moments} and
\(\mathbb Ez^3=\mathbb Er^3-3\sigma^2\mathbb Er^2
+3\sigma^4\mathbb Er-\sigma^6\) gives
\begin{align}
 \mathbb Ez^3
 ={}&-\frac4{(d+2)^2(d+4)(d+6)}
 \bigl[
 d^2(\sigma^6+3\sigma^2\sigma_4-4\sigma_6) \notag\\
 &\hspace{35mm}
 +d(10\sigma^6+18\sigma^2\sigma_4-16\sigma_6)+24\sigma^6
 \bigr].
 \label{eq:high-dimensional-third-moment}
\end{align}
Consequently, uniformly on every compact interval \(I\) contained in
\((0,1)\),
\begin{align}
 \mathbb Ez
 &=-\frac{2\sigma^2}{d}+O_{C^2(I)}(d^{-2}),\notag\\
 \mathbb Ez^2
 &=\frac{2}{d}\left(\sigma^4-1-
 \frac{(1-\gamma)^3}{1+\gamma+\gamma^2+\gamma^3}\right)
 +O_{C^2(I)}(d^{-2}),\notag\\
 \mathbb Ez^3
 &=O_{C^2(I)}(d^{-2}).
 \label{eq:high-dimensional-centered-moments}
\end{align}

\smallskip
\noindent\emph{Uniform remainder estimate.}
We record the bounds needed for the Taylor remainder.  The projection of
the uniform measure on the sphere \(S^{d+1}\) onto its first \(d\)
coordinates is uniform on \(\B\).  The spherical spectral gap therefore
implies the Poincar\'e inequality
\begin{equation}
 \operatorname{Var}F
 \leq\frac1{d+1}\mathbb E\sum_i\norm{\nabla_iF}^2
 \label{eq:ball-product-poincare}
\end{equation}
for functions of independent uniform ball points
\cite{BonnefontJoulinMa2016}.  For the Gram form,
\[
 \nabla_ig_{ac}=a_iy_c+c_iy_a,
 \qquad
 \sum_i\norm{\nabla_ig_{ac}}^2\leq k(a,c),
\]
where
\[
 k(a,c)
 =2\|a\|_2^2\|c\|_1^2+2\|c\|_2^2\|a\|_1^2.
\]
Let \(\widetilde g_{ac}=g_{ac}-\mathbb Eg_{ac}\).  The first application of
Eq.\eqref{eq:ball-product-poincare} gives
\[
 \mathbb E\widetilde g_{ac}^2\leq\frac{k(a,c)}{d+1}.
\]
Applying it to \(\widetilde g_{ac}^2\), and using
\(\nabla_i(\widetilde g_{ac}^2)=2\widetilde g_{ac}\nabla_ig_{ac}\), gives
\begin{align*}
 \operatorname{Var}(\widetilde g_{ac}^2)
 &\leq\frac4{d+1}
 \mathbb E\left[
 \widetilde g_{ac}^2\sum_i\norm{\nabla_ig_{ac}}^2
 \right]\\
 &\leq\frac{4k(a,c)}{d+1}\mathbb E\widetilde g_{ac}^2
 \leq\frac{4k(a,c)^2}{(d+1)^2}.
\end{align*}
Since
\(\mathbb E\widetilde g_{ac}^4=\operatorname{Var}(\widetilde g_{ac}^2)
+(\mathbb E\widetilde g_{ac}^2)^2\), the two applications give
\[
 \mathbb E\widetilde g_{ac}^4
 \leq\frac{5k(a,c)^2}{(d+1)^2}.
\]
For the infinite coefficient sequences, apply these estimates to finite
truncations and pass to the limit using the uniformly bounded
\(\ell^1\)-tails.
The coefficient sequences \(b,b',b''\) have uniformly bounded
\(\ell^1\)- and \(\ell^2\)-norms on \(I\).  Thus, with
\(z_1=z'\) and \(z_2=z''\),
\begin{equation}
 \mathbb Ez^4+\mathbb Ez_1^4+\mathbb Ez_2^4=O_I(d^{-2}).
 \label{eq:high-dimensional-fourth-moments}
\end{equation}
H\"older's inequality also bounds
\(\mathbb E|z|^3|z_1|\), \(\mathbb E|z|^3|z_2|\), and
\(\mathbb Ez^2z_1^2\) by \(O_I(d^{-2})\).

\smallskip
\noindent\emph{Negative moments near \(r=0\).}
The possible singularity at \(r=0\) is controlled geometrically.
Conditionally on \(x_\gamma=x\),
\[
 \Pr(r\leq t\mid x)
 =\frac{\operatorname{vol}(\B\cap B(x,\sqrt t))}{\omega_d}
 \leq t^{d/2},
 \qquad 0\leq t\leq1.
\]
Hence, for \(0<\kappa<d/2\) and \(0<a<1\),
\begin{equation}
 \mathbb E[r^{-\kappa};r\leq a]
 \leq\frac d{d-2\kappa}a^{d/2-\kappa}.
 \label{eq:high-dimensional-negative-moment}
\end{equation}
This estimate includes \(d=1\) and \(1<\alpha<2\), where the largest
required exponent is \(1-\alpha/2<1/2\).

\smallskip
\noindent\emph{Differentiated estimate.}
On \(I\), the stationary series and its first two derivatives converge
uniformly.  Pathwise differentiation gives
\[
 r'=-2(p-x_\gamma)\mathbin{\cdot}x_\gamma',
 \qquad
 r''=2\norm{x_\gamma'}^2
 -2(p-x_\gamma)\mathbin{\cdot}x_\gamma'',
\]
and therefore
\begin{equation}
 |r'|\leq C_I\sqrt r,
 \qquad
 |r''|\leq C_I(1+\sqrt r).
 \label{eq:high-dimensional-pathwise-derivatives}
\end{equation}
The kernel \(z\mapsto\norm z^\alpha\) has locally integrable second
derivatives when \(d+\alpha>2\).  Convolution with the indicator of \(\B\),
Eq.\eqref{eq:high-dimensional-negative-moment}, and
Eq.\eqref{eq:high-dimensional-pathwise-derivatives} justify the two
derivatives under the expectation.

Expand \(f(r)=r^\beta\) to order three at \(\sigma^2\).  Establish the value and
the first two derivatives of the expected remainder before splitting the
expectation into
\[
 |z|\leq \sigma^2/2,
 \qquad r<\sigma^2/2,
 \qquad r>3\sigma^2/2.
\]
On the first event, Taylor's theorem and
Eq.\eqref{eq:high-dimensional-fourth-moments} bound the remainder and its
first two derivatives by \(O_{\alpha,I}(d^{-2})\).  The same estimate holds
on the upper event by fourth-moment Markov inequality.  On the lower event,
Eqs.\eqref{eq:high-dimensional-negative-moment} and
\eqref{eq:high-dimensional-pathwise-derivatives} give an exponentially
small bound.  No event indicator is differentiated in this argument.
It follows that
\[
 \left\|
 \mathbb Er^\beta-f(\sigma^2)-f'(\sigma^2)\mathbb Ez
 -\frac12f''(\sigma^2)\mathbb Ez^2
 -\frac16f'''(\sigma^2)\mathbb Ez^3
 \right\|_{C^2(I)}
 =O_{\alpha,I}(d^{-2}).
\]
Combining this estimate with
Eq.\eqref{eq:high-dimensional-centered-moments} gives
\[
 \begin{aligned}
 M_{d,\alpha}
 ={}&f(\sigma^2)+\frac1d\left[
  -2\sigma^2f'(\sigma^2)
  +\left(\sigma^4-1-
  \frac{(1-\gamma)^3}{1+\gamma+\gamma^2+\gamma^3}\right)f''(\sigma^2)
 \right]\\
 &+O_{\alpha,I}(d^{-2})\quad\text{in }C^2(I).
 \end{aligned}
\]
whose coefficient is precisely \(\mathcal C_\alpha\).

For the uniform value estimate on \([0,1]\), the same cubic expansion is
used without differentiation.  The coefficient power sums and all centered
moments above have uniform continuous extensions to the endpoints.  Since
\(1\leq\sigma^2\leq2\) and \(0\leq r\leq4\),
Eq.\eqref{eq:high-dimensional-fourth-moments} controls the complement of
\(|z|\leq\sigma^2/2\) uniformly.  This proves
Eq.\eqref{eq:high-dimensional-first-order}.
\end{proof}


\begin{thebibliography}{99}\small\sloppy\hbadness=2000

\bibitem{AlonAzar1993}
N. Alon and Y. Azar,
On-line Steiner trees in the Euclidean plane,
\emph{Discrete \& Computational Geometry} 10 (1993), 113--121.
\href{https://doi.org/10.1007/BF02573969}
{doi:10.1007/BF02573969}.

\bibitem{BonnefontJoulinMa2016}
M. Bonnefont, A. Joulin, and Y. Ma,
Spectral gap for spherically symmetric log-concave probability measures,
and beyond,
\emph{Journal of Functional Analysis} 270 (2016), 2456--2482.
\href{https://doi.org/10.1016/j.jfa.2016.02.007}
{doi:10.1016/j.jfa.2016.02.007}.

\bibitem{CastroDevillers2011}
P. M. M. de Castro and O. Devillers,
On the asymptotic growth rate of some spanning trees embedded in
\(\mathbb R^d\),
\emph{Operations Research Letters} 39 (2011), 44--48.

\bibitem{CastroGammaInsertion}
P. M. M. de Castro,
Sequential Euclidean connections with exponential memory:
distributional performance and adversarial robustness,
arXiv:2608.25298, 2026.
\href{https://doi.org/10.48550/arXiv.2608.25298}
{doi:10.48550/arXiv.2608.25298}.

\bibitem{Cogger1974}
K. O. Cogger,
The optimality of general-order exponential smoothing,
\emph{Operations Research} 22 (1974), 858--867.
\href{https://doi.org/10.1287/opre.22.4.858}
{doi:10.1287/opre.22.4.858}.

\bibitem{deBergMarkovicUmboh2023}
M. de Berg, A. Markovic, and S. W. Umboh,
The online broadcast range-assignment problem,
\emph{Algorithmica} 85 (2023), 3928--3956.
\href{https://doi.org/10.1007/s00453-023-01166-4}
{doi:10.1007/s00453-023-01166-4}.

\bibitem{deBergSadhukhanSpieksma2024}
M. de Berg, A. Sadhukhan, and F. Spieksma,
Stable approximation algorithms for the dynamic broadcast range-assignment
problem,
\emph{SIAM Journal on Discrete Mathematics} 38 (2024), 790--827.
\href{https://doi.org/10.1137/23M1545975}
{doi:10.1137/23M1545975}.

\bibitem{DiaconisFreedman1999}
P. Diaconis and D. Freedman,
Iterated random functions,
\emph{SIAM Review} 41 (1999), 45--76.
\href{https://doi.org/10.1137/S0036144598338446}
{doi:10.1137/S0036144598338446}.

\bibitem{Folland1999}
G. B. Folland,
\emph{Real Analysis: Modern Techniques and Their Applications},
2nd ed., Wiley, New York, 1999.

\bibitem{GuGuptaKumar2016}
A. Gu, A. Gupta, and A. Kumar,
The power of deferral: maintaining a constant-competitive Steiner tree online,
\emph{SIAM Journal on Computing} 45 (2016), 1--28.
\href{https://doi.org/10.1137/140955276}
{doi:10.1137/140955276}.

\bibitem{Hoeffding1963}
W. Hoeffding,
Probability inequalities for sums of bounded random variables,
\emph{Journal of the American Statistical Association} 58 (1963), 13--30.

\bibitem{ImaseWaxman1991}
M. Imase and B. M. Waxman,
Dynamic Steiner tree problem,
\emph{SIAM Journal on Discrete Mathematics} 4 (1991), 369--384.
\href{https://doi.org/10.1137/0404033}
{doi:10.1137/0404033}.

\bibitem{ZhuZhouXu2014}
B. Zhu, J. Zhou, and W. Xu,
Dual Orlicz--Brunn--Minkowski theory,
\emph{Advances in Mathematics} 264 (2014), 700--725.

\bibitem{Karp1977}
R. M. Karp,
A characterization of the minimum cycle mean in a digraph,
Technical Report UCB/ERL M77/47, University of California, Berkeley, 1977.
\href{https://www2.eecs.berkeley.edu/Pubs/TechRpts/1977/29121.html}
{Berkeley report page}.

\bibitem{LopesOliveira2009}
A. O. Lopes and E. R. Oliveira,
Entropy and variational principles for holonomic probabilities of IFS,
\emph{Discrete and Continuous Dynamical Systems} 23 (2009), 937--955.
\href{https://doi.org/10.3934/dcds.2009.23.937}
{doi:10.3934/dcds.2009.23.937}.

\bibitem{MengueOliveira2017}
J. K. Mengue and E. R. Oliveira,
Duality results for iterated function systems with a general family of
branches,
\emph{Stochastics and Dynamics} 17 (2017), 1750021.
\href{https://doi.org/10.1142/S0219493717500216}
{doi:10.1142/S0219493717500216}.

\bibitem{OlkinTong1988}
I. Olkin and Y. L. Tong,
Peakedness in multivariate distributions,
in \emph{Statistical Decision Theory and Related Topics IV}, vol.~2,
S. S. Gupta and J. O. Berger, eds., Springer, New York, 1988, 373--383.

\bibitem{Oliveira2019}
E. R. Oliveira,
On the connection between a skew product IFS and the ergodic optimization for
a finite family of potentials,
\emph{Dynamical Systems} 34 (2019), 685--709.
\href{https://doi.org/10.1080/14689367.2019.1606896}
{doi:10.1080/14689367.2019.1606896}.

\bibitem{PenroseWade2008}
M. D. Penrose and A. R. Wade,
Limit theory for the random on-line nearest-neighbour graph,
\emph{Random Structures \& Algorithms} 32 (2008), 125--156.
\href{https://doi.org/10.1002/rsa.20185}
{doi:10.1002/rsa.20185}.

\bibitem{Roberts1959}
S. W. Roberts,
Control chart tests based on geometric moving averages,
\emph{Technometrics} 1 (1959), 239--250.
\href{https://doi.org/10.1080/00401706.1959.10489860}
{doi:10.1080/00401706.1959.10489860}.

\bibitem{Steele1989}
J. M. Steele,
Cost of sequential connection for points in space,
\emph{Operations Research Letters} 8 (1989), 137--142.
\href{https://doi.org/10.1016/0167-6377(89)90039-4}
{doi:10.1016/0167-6377(89)90039-4}.

\bibitem{Wade2009}
A. R. Wade,
Asymptotic theory for the multidimensional random on-line nearest-neighbour
graph,
\emph{Stochastic Processes and their Applications} 119 (2009), 1889--1911.
\href{https://doi.org/10.1016/j.spa.2008.09.006}
{doi:10.1016/j.spa.2008.09.006}.

\bibitem{Xu2000}
Y. Xu,
Funk--Hecke formula for orthogonal polynomials on spheres and on balls,
\emph{Bulletin of the London Mathematical Society} 32 (2000), 447--457.
\href{https://doi.org/10.1112/S0024609300007001}
{doi:10.1112/S0024609300007001}.

\end{thebibliography}
\end{document}